%% file: main-arxiv2023.tex
\documentclass[11pt,twoside]{article}
\usepackage[english]{babel}
\usepackage[T1]{fontenc}
\usepackage[utf8]{inputenc}
\usepackage[margin=1in]{geometry}

\usepackage{amsmath}
\usepackage{amsthm}
\usepackage{mathtools,amssymb,latexsym} 
\usepackage{amsxtra} 
\usepackage[mathscr]{eucal}
\usepackage{turnstile}
\usepackage{upgreek}
\usepackage{pifont}
\usepackage[10pt]{moresize}

\usepackage{tikz}
\usetikzlibrary{shapes,shapes.multipart, calc,matrix,arrows,arrows,positioning,automata}
\tikzset{
	>=stealth',
	punkt/.style={
		circle,
		rounded corners,
		draw=black, thick, 
		text width=1.5em,
		minimum height=2em,
		text centered},
	punkts/.style={
		circle,
		rounded corners,
		draw=black, thick, 
		text width=1em,
		minimum height=1em,
		text centered},
	invisible/.style={
		draw=none,
		text width=1.5em,
		minimum height=0em,
		text centered},
	inv/.style={
		draw=none,
		text width=2.5em,
		minimum height=3em,
		text centered},
	pil/.style={
		->,
		thick,
		shorten <=2pt,
		shorten >=2pt,}
}

\usepackage[ruled,titlenumbered,linesnumbered,boxed,vlined,noend,algo2e]{algorithm2e}
\usepackage{algorithm}
\usepackage{changepage}
\usepackage{ragged2e}

\newcommand{\ysi}{\sigma}

\newcommand{\ySis}{\Sigma^\star}
\newcommand{\yal}{\alpha}
\newcommand{\ybe}{\beta}
\newcommand{\yga}{\gamma}
\newcommand{\yGa}{\Gamma}
\newcommand{\yGas}{\Gamma^\star}

\newcommand{\yde}{\delta}
\newcommand{\yDe}{\Delta}

\newcommand{\ytau}{\varsigma}

\newcommand{\yemp}{\emptyset}
\newcommand{\ysse}{\subseteq}
\newcommand{\ypow}[1]{\mathcal{P}(#1)}
\newcommand{\yst}{\,|\,} 

\newcommand{\yvptsn}[1]{\mathscr{#1}}  

\newcommand{\yvpts}[5]{\langle #1, #2, #3, #4, #5 \rangle}  
\newcommand{\yvptsS}{\yvpts{S}{S_{in}}{L}{\yGa}{T}}  

\newcommand{\yltsconf}[1]{\yltsn{C}_{#1}}  

\newcommand{\yiovpts}[6]{\langle #1, #2, #3, #4,#5, #6\rangle}  
\newcommand{\yiovptsS}{\yiovpts{S}{S_{in}}{L_I}{L_U}{\yGa}{T}}  

\newcommand{\yiovp}[2]{\yvptsn{IOVP}(#1,#2)}  

\newcommand{\yait}{\sharp}  
\newcommand{\yvpa}[6]{\langle #1, #2, #3, #4, #5, #6\rangle}  
\newcommand{\yvpaA}{\yvpa{S}{S_{in}}{A}{\yGa}{\rho}{F}}  
\newcommand{\yvpaB}{\yvpa{Q}{Q_{in}}{A}{\yDe}{\mu}{G}}  

\newcommand{\yltsn}[1]{\mathscr{#1}}  
\newcommand{\yS}{\yltsn{S}}  
\newcommand{\yA}{\yltsn{A}}  
\newcommand{\yB}{\yltsn{B}}  
\newcommand{\yI}{\yltsn{I}}  
\newcommand{\yQ}{\yltsn{Q}}  
\newcommand{\yP}{\yltsn{P}}  
\newcommand{\yV}{\yltsn{V}}  

\newcommand{\ytr}[3]{#1\overset{#2}{\rightarrow}#3} 
\newcommand{\ytrt}[3]{#1\overset{#2}{\Rightarrow}#3} 

\newcommand{\ytrtf}[4]{#1\overset{#2}{\underset{#3}{\rightarrow}}#4} 

\newcommand{\ypdaconf}[1]{\yltsn{C}_{#1}}  

\newcommand{\ypdatrt}[3]{#1\underset{#2}{\mapsto}#3} 
\newcommand{\ypdatrtf}[4]{#1\overset{#2}{\underset{#3}{\mapsto}}#4} 

\newcommand{\ycfgtrtf}[4]{#1\overset{#2}{\underset{#3}{\leadsto}}#4} 

\newcommand{\ycfgtrtfl}[4]{#1\overset{#2}{\underset{#3}{\hookrightarrow}}#4} 

\newcommand{\yioltsn}[1]{\mathscr{#1}} 
\newcommand{\yT}{\yioltsn{T}}  
\newcommand{\yD}{\yioltsn{D}}  
\newcommand{\yF}{\yioltsn{F}}  

\newcommand{\ytru}[4]{#1\overset{#2}{\underset{#3}{\rightarrow}}#4} 
\newcommand{\ytrut}[4]{#1\overset{#2}{\underset{#3}{\Rightarrow}}#4} 

\newcommand{\yconf}[2]{\,\text{\bf vconf}_{#1,#2}\,\,} 

\newtheorem{theo}{Theorem}
\newtheorem{lemm}[theo]{Lemma} 

\newtheorem{prop}[theo]{Proposition}
\newtheorem{defi}[theo]{Definition} 
\newtheorem{exam}[theo]{Example}

\newcommand{\yfun}[3]{#1:#2\rightarrow #3} 
\newcommand{\yoh}[1]{\mathcal{O}(#1)} 

\newcommand{\yeps}{\varepsilon} 
\newcommand{\yfim}{\hfill$\Box$} 

\newcommand{\ycomp}[1]{\overline{#1}}

\newcommand{\ypush}[1]{#1_+}
\newcommand{\ypop}[1]{#1_-}

\begin{document}
	
\input{title-arxiv2023-vconf}
\input{introduction}
\input{notation-vpa}

\input{vpts}

\input{vconf}

\input{conclusions}

\input{main-arxiv2023.bbl}
\end{document}

%% file: title-arxiv2023-vconf.tex
\title{Conformance Checking for Pushdown Reactive Systems based on Visibly Pushdown Languages}

\author{Adilson Luiz Bonifacio\thanks{Computing Department, University of Londrina, Londrina, Brazil.}}

\date{} 

\maketitle

\input{abstract}

%% file: abstract.tex
\begin{abstract}
Testing pushdown reactive systems is deemed important to guarantee a precise and robust software development process. 
Usually, such systems can be specified by the formalism of Input/Output Visibly Pushdown Labeled Transition System (IOVPTS), where the interaction with the environment is regulated by a pushdown memory. 
Hence a conformance checking can be applied in a testing process to verify whether an implementation is in compliance to a specification using an appropriate conformance relation. 
In this work we establish a novelty conformance relation based on Visibly Pushdown Languages (VPLs) that can model sets of desirable and undesirable behaviors of systems. 
Further, we show that test suites with a complete fault coverage can be generated using this conformance relation for pushdown reactive systems.
\end{abstract}

%% file: introduction.tex
\section{Introduction}

A large number of real-world systems, such as communication protocols (internet), vehicle and aircraft control systems, and most of industrial control systems, are indeed reactive systems, where their behavior is dictated by their interaction with the environment. 
That is, reactive systems receive input stimuli from an external environment and, eventually, send back outputs in responses. 

Reactive systems, in general, are critical applications and then their development process must be precise and robust. 
Automatic supporting techniques are then deemed important, particularly in the testing activity, where high costs are associated to maintenance when the testing step is poorly conducted during the system development process.
Model-based testing meets such requirement relying on formal methods, then supporting several models and testing approaches for reactive systems. 

Usually, reactive behaviors can be appropriately specified by the formalism of Input/Output Labeled Transition Systems (IOLTSs)~\cite{tret-model-2008}, where the exchange of input and output stimuli can occur asynchronously. 
Hence a conformance checking can be applied in a testing process to verify whether an implementation is in compliance to a specification~\cite{bonifacio2020cleiej}. 
This conformance checking process is established over a certain conformance relation~\cite{ananbcc-orchestrated-2013,bonifacio2020cleiej} for IOLTS models, such as the Input/Output Conformance ({\bf ioco}) relation~\cite{simap-generating-2014,tret-model-2008} and its variations~\cite{bonifacio2020cleiej,stg,uppaal,torXAkis}. 

In this work we turn into a more complex reactive systems, where the interaction with the environment is regulated by a pushdown memory through the communication channel.  
We study aspects of conformance testing and test suite generation for pushdown reactive systems that can be treated by Input/Output Visibly Pushdown Labeled Transition Systems (IOVPTSs)~\cite{bonifacio2022cleiej} and also by Visibly Pushdown Automata (VPAs)~\cite{alurm-visibly-2004}. 
In this setting a more general conformance relation has been given in the same spirit of the classical {\bf ioco} relation, where any observable implementation behavior that was not already present in the corresponding specification must also be prevented in a conformance checking process~\cite{bonifacio2022cleiej}. 

Here we go further and address this problem in a more general setting, where particular desirable and undesirable behaviors can be specified by Visibly Pushdown Languages (VPLs)~\cite{alurm-visibly-2004}. 
The class of VPLs is strict more powerful than the classical regular  languages treated by {\bf ioco} relation and some extensions. 
We then propose and prove the correctness of an efficient algorithm that can check conformance in this scenario. 
The language-based approach is established on a white-box testing scenario, where the structure of implementations under test (IUTs) are previously known by the tester, and the participating models have an auxiliary stack memory.

The remaining of this paper is organized as follows. 
Section~\ref{sec:models} establishes some notations, definitions and preliminary results over VPAs. 
VPTS and IOVPTS models, and their relationship to VPAs are defined in Section~\ref{sec:conformance}. 
Section~\ref{sec:cfl-suites} defines a visibly conformance relation, a fault model over VPLs, and shows how to obtain complete test suites, of polynomial complexity, for checking this class of models. 
Section~\ref{sec:conclusion} gives concluding remarks. 

%% file: notation-vpa.tex
\section{Notation and Preliminary Results}\label{sec:models}

This section establishes concepts and notation to ease  the reference when testing complex reactive systems. 
We also give some preliminary results that will be useful later and precisely define all models and formalism studied in this work, such as  Visibly Pushdown Automata (VPA) and closure properties over Visibly Pushdown Languages (VPLs). 


We start with some basic notation. 
So let $X$ and $Y$ be sets we indicate by $\ypow{X}=\{Z\yst Z\ysse X\}$ the power set of $X$, and $X-Y=\{z\yst z \in X \text{ and } z \not\in Y\}$ the set difference. 
Also, we assume that $X_Y=X\cup Y$ and when $Y=\{y\}$ is a singleton we may write $X_y$ for $X_{\{y\}}$.
If $X$ is a finite set, the size of $X$ will be indicated by $|X|$.

An alphabet is any non-empty set of symbols.
Let $A$ be an alphabet, a word over $A$ is any finite sequence $\ysi=x_1\ldots x_n$ of symbols in $A$, that is, $n\geq 0$ and $x_i\in A$, for all $i=1,2,\ldots, n$.
When $n=0$, $\ysi$ is the empty sequence, also indicated by $\yeps$.
The set of all finite sequences, or words, over $A$ is denoted by $A^{\star}$, and the set of all nonempty finite words over $A$ is indicated by $A^+$.
When we write $x_1x_2\ldots x_n\in A^{\star}$, it is implicitly assumed that $n\geq 0$ and that $x_i\in A$, $1\leq i\leq n$, unless explicitly noted otherwise.
The length of a word $\alpha$ over $A$ is indicated by $|\alpha|$.
Hence, $|\yeps|=0$.

Further, let $\ysi=\ysi_1\ldots \ysi_n$ and $\rho=\rho_1\ldots \rho_m $ be two words over $A$. 
The concatenation of $\ysi$ and $\rho$, indicated by $\ysi\rho$, is the word $\ysi_1\ldots\ysi_n\rho_1\ldots\rho_m$.
Clearly, $|\ysi\rho|=|\ysi|+|\rho|$.
A language $G$ over $A$ is any set $G\ysse A^\star$ of words over $A$.
Also let $G_1$, $G_2\ysse A^\star$ be languages over $A$. 
The product of $G_1$ and $G_2$, indicated by $G_1 G_2$, is the language 
$\{\ysi\rho\yst \ysi\in G_1, \rho\in G_2\}$. 
If $G\ysse A^\star$ is a language over $A$, then its complement is the language $\ycomp{G}=A^\star-G$.

We will also use the notion of morphism to treat alphabets.
So let $A$, $B$ be alphabets. 
A \emph{homomorphism}, or just a \emph{morphism}, from $A$ to $B$ is any function $h:A\rightarrow B^\star$.
A morphism $h:A\rightarrow B^\star$ can be extended in a natural way to a function 
$\widehat{h}:A^\star\rightarrow B^\star$, thus
either $\widehat{h}(\ysi) = \yeps$ if $\ysi=\yeps$ or $\widehat{h}(\ysi) = h(a)\widehat{h}(\rho)$ if $\ysi=a\rho$ with $a\in A$.
We can further lift $\widehat{h}$ to a function $\widetilde{h}: \ypow{A^\star}\rightarrow \ypow{B^\star}$ in a natural way, by letting $\widetilde{h}(G)=\bigcup\limits_{\ysi\in G}\widehat{h}(\ysi)$,
for all $G\ysse A^\star$. 
In order to avoid cluttering the notation, we may write $h$ instead of $\widehat{h}$, or of $\widetilde{h}$, when no confusion can arise.
When $a\in A$, we define the simple morphism $h_a:A\rightarrow A-\{a\}$ by letting $h_a(a)=\yeps$, and $h_a(x)=x$ when $x\neq a$.  
Hence, $h_a(\ysi)$ erases all occurrences of  $a$ in the word $\ysi$.
%

We now turn to transition systems.

 
 \subsection{Visibly Pushdown Automata}\label{subsec:vpa}
 
 A Visibly Pushdown Automaton (VPA)~\cite{alurm-visibly-2004} is, basically, a Pushdown Automaton (PDA)~\cite{hopcu-introduction-1979}, with a transition relation over an alphabet and a pushdown stack (or just a stack, for short) associated to it. 
 Thus a VPA can make use of a potentially infinite memory as with PDA models. 
 Any alphabet $L$ is always partitioned into three disjoint subsets $L=L_c\cup L_r\cup L_i$.
 Elements in the set $L_c$ are ``call symbols'', or ``push symbols'', and specify push actions on the stack.
 Elements in $L_r$ are ``return symbols'', or ``pop symbols'', and indicate pop actions, and in $L_i$ we find ``simple symbols'',  that do not change the stack%
 \footnote{In~\cite{alurm-visibly-2004} symbols in $L_i$ are called ``internal action symbols''.
 	In this text we will reserve that denomination for another special symbol as will be apparent later.}.%
 
 The next definition is a slight extension of the similar notion appearing in~\cite{alurm-visibly-2004}. 
 Here, we also allow $\yeps$-moves, that is, the VPA can change states without reading any symbol from the input.
 \begin{defi}\label{def:pda}
 	A \emph{Visibly Pushdown Automaton} (VPA)~\cite{alurm-visibly-2004} over a finite input alphabet $A$ is a tuple $\yA = \yvpaA$, where:
 	\begin{itemize}
 		\item[---] $A=A_c\cup A_r\cup A_i$ and $A_c$, $A_r$, $A_i$ are pairwise disjoint; 
 		\item[---] $S$ is a finite set of \emph{states};
 		\item[---] $S_{in}\ysse S$ is set of \emph{initial states};
 		\item[---] $\Gamma$ is a finite \emph{stack alphabet}, with $\bot\not \in \yGa$ the \emph{initial stack symbol};
 		\item[---] The \emph{transition relation} is $\rho=\rho_c\cup\rho_r\cup\rho_i$, where $\rho_c\ysse S\times A_c \times \yGa  \times S$, $\rho_r\ysse S\times A_r \times \yGa_\bot  \times S$, and 
 		$\rho_i\ysse S\times (A_i\cup\{\yeps\}) \times \{\yait\} \times S$,  where $\yait\not\in\yGa_\bot$ is a place-holder symbol;
 		\item[---] $F\ysse S$ is the set of \emph{final states}. 
 	\end{itemize}
 	A transition $(p,\yeps,\yait,q)\in \rho_i$ is called an \emph{$\yeps$-move} of $\yA$. 
 	A \emph{configuration} of $\yA$ is any triple $(p,\ysi,\yal)\in S\times A^\star\times(\yGas\{\bot\})$, and the set of all configurations of $\yA$ it is indicated by $\ypdaconf{\yA}$.
 \end{defi}
 
 A tuple $(p,a,Z,q)\in \rho_c$ specifies a \emph{push-transition}.
 We have a \emph{pop-transition} if $(p,a,Z,q)\in \rho_r$, and a \emph{simple-transition} if $(p,a,Z,q)\in \rho_i$.
 The intended meaning for a push-transition $(p,a,Z,q)\in \rho_c$ is that $\yS$, in  state $p$ and reading input $a$,  changes to state $q$ and pushes $Z$ in the stack. 
 A pop-transition $(p,a,Z,q)\in \rho_r$ makes $\yS$, in state $p$ and reading $a$,  pop $Z$ from the stack and change to state $q$. 
 A simple-transition $(p,a,\yait,q)\in \rho_i$  has the intended meaning of reading $a$ and changing  from state $p$ to state $q$, no matter what the topmost stack symbol is. 
 An $\yeps$-transition  $(p,\yeps,\yait,q)\in \rho_i$ indicates that $\yS$ must move from state $p$ to state  $q$, while reading no symbol from the input. 
 
 \begin{exam}
 	Figure~\ref{vpafig} depicts a VPA with $L_c=\{a\}$, $L_r=\{b\}$ and $L_i=\emptyset$.
 	Also, $S_{in}=\{s_0\}$, $F =\{s_f\}$ and $\yGa=\{A,B\}$. 
 	\input{figs/vpafig.tex}
 	It is clearly that the language $L(A)=\{a^nb^n\}$ with $n\geq 0$ is accepted by the VPA $A$. \yfim
 \end{exam}


 We can now define the relation $\ypdatrt{}{}{}\ysse\ypdaconf{\yA}\times\ypdaconf{\yA}$ which captures simple moves of a VPA $\yA$.
 Note that the model can make pop moves when the stack contains only the initial stack symbol $\bot$. 
 The semantics of a VPA is the language comprised by all input strings it accepts.
 
 \begin{defi}\label{def:fsa-move}
 	Let $\yA=\yvpaA$ be a VPA. 
 	For all $\ysi\in A^\star$:
 	\begin{enumerate}
 		\item If $(p,a,Z,q)\in\rho_c$  then $\ypdatrt{(p,a\ysi,\yal\bot)}{}{(q,\ysi,Z\yal\bot)}$;
 		\item If $(p,a,Z,q)\in\rho_r$ then $\ypdatrt{(p,a\ysi, Z\yal\bot)}{}{(q,\ysi,\yal\bot)}$ when $Z\neq \bot$, or  $\ypdatrt{(p,a\ysi, \bot)}{}{(q,\ysi,\bot)}$ when $Z=\bot$; 
 		\item If $(p,a,\yait,q)\in\rho_i$  then $\ypdatrt{(p,a\ysi, \yal\bot)}{}{(q,\ysi, \yal\bot)}$. 
 	\end{enumerate}
 	The set $L(\yA)=\big\{ \ysi\in A^\star\yst \ypdatrtf{(s_0,\ysi,\bot)}{\star}{\yA}{(p,\yeps,\ybe)},\,\,s_0\in S_{in},\,\,p\in F\big\}$ 
 	is the \emph{language accepted} by $\yA$.
 	Two VPAs $\yA$ and $\yB$ are said to be \emph{equivalent} when $L(\yA)=L(\yB)$. 
 \end{defi}
 We may also write $\ypdatrt{(p,a\ysi,\yal)}{\yA}{(q,\ysi,\yga)}$ if it is important to explicitly mention  the VPA $\yA$.
 It is clear that when $\ypdatrt{(p,a\ysi,\yal)}{}{(q,\mu,\ybe)}$ then we get  $(q,\mu,\ybe)\in \ypdaconf{\yA}$, so that $\mapsto$ is a relation on the set $\ypdaconf{\yA}$. 
 The $n$-th power of the relation $\ypdatrt{}{}{}$ will be indicated by $\ypdatrtf{}{n}{}{}$ for all $n\geq 0$, and its reflexive and transitive closure will be indicated by $\ypdatrtf{}{\star}{}{}$.
 
A language is said to be a Visibly Pushdown Language (VPL) when it is accepted by a VPA. 
That is, let $A$ be an alphabet and let $G\ysse A^\star$ be a language over $A$, then $G$ is a \emph{Visibly Pushdown Language} if there is a VPA $\yA$  such that $L(\yA)=G$.
 
 
 Given any VPA, we can always construct and equivalent VPA with no $\yeps$-moves and the same number of states as the original VPA. 
 \begin{prop}\label{prop:no-eps}
 	For any  VPA  we can effectively construct an equivalent VPA with no $\yeps$-moves and with the same number of states.
 \end{prop}
\begin{proof}
	Let $\yA=\yvpaA$.
	Define the mapping $\yfun{E}{S}{\ypow{S}}$ by
	$$E(s)=\{p\yst \ypdatrtf{(s,\yeps,\bot)}{n}{\yA}{(p,\yeps,\bot)}, n\geq 0\},$$ 
	that is, $E(s)$ is the set of all states that can be reached from $s$ through $\yeps$-moves.
	A simple inductive algorithm guarantees that we can effectively compute $E$.
	
	Consider the VPA $\yB=\yvpa{S}{Q_{in}}{A}{\yGa}{\mu}{F}$ where 
	$Q_{in}=\bigcup\limits_{s\in S_{in}} E(s)$, and 
	the set of transitions $\mu$ is obtained from $\rho$ by removing all $\yeps$-transitions from $\rho$, and then adding to $\mu$ all transitions $(r,a,Z,p)$ where $a\neq \yeps$, $(s,a,Z,q)$ is a transition in $\rho$ and, by means of $\yeps$-transitions alone, we can reach $s$ from $r$ and  $p$ from $q$, that is,
	$$\mu=\left[\strut\rho-\{(s,\yeps,\yait,p)\yst s,p\in S\}\right]\bigcup \left\{\strut(r,a,Z,p)\yst (s,a,Z,q)\in \rho, a\neq \yeps, s\in E(r), p\in E(q) \right\}.$$
	Clearly, $\yB$ has no $\yeps$-moves, and has the same number of states as $\yA$.
	
	The next two claims will be used to show that $\yA$ and $\yB$ are equivalent VPAs. 
	\begin{description}
		\item[\sc Claim 1:] If $\ysi\neq \yeps$ and $\ypdatrtf{(s,\ysi,\yal)}{\star}{\yA}{(t,\yeps,\ybe)}$, then 
		$\ypdatrtf{(s,\ysi,\yal)}{\star}{\yB}{(t,\yeps,\ybe)}$.\\
		{\sc Proof:} Let $\ypdatrtf{(s,\ysi,\yal)}{n}{\yA}{(t,\yeps,\ybe)}$, $n\geq 1$ and $\ysi=a\yde$ where $a\in A$. 
		When $n=1$, we have $\yde=\yeps$ and there is a transition $(s,a,Z,t)$ in $\rho$. 
		Since $s\in E(s)$ and $t\in E(t)$, the construction gives $(s,a,Z,t)$ in $\mu$ and the result follows immediately. 
		When $n>1$ we may write
		$$(s,\ysi,\yal)=\ypdatrtf{(s,a\yde,\yal)}{k}{\yA}{(r,a\yde,\yal)}\ypdatrtf{}{}{\yA}{(q,\yde,\yga)}\ypdatrtf{}{m}{\yA}{(t,\yeps,\ybe)},$$
		where $k\geq 0$ and $k+1+m=n$, that is, the first $k$ moves of $\yA$ are $\yeps$-moves, while the next one is a move over $a$, through a transition $(r,a,Z,q)\in\rho$, for some $Z\in\yGa_\bot$. 
		Then, $r\in E(s)$. 
		If $\yde=\yeps$, the last $m$ moves of $\yA$ are also $\yeps$-moves, and we get $\yga=\ybe$ and $t\in E(q)$.
		In this case, the construction gives $(s,a,Z,t)\in\mu$.
		Since $\ysi=a\yde=a$ we have  
		$(s,\ysi,\yal)=\ypdatrtf{(s,a,\yal)}{}{\yB}{(t,\yeps,\yga)}=(t,\yeps,\ybe)$, as desired.
		Now let $\yde\neq \yeps$. Since $q\in E(q)$ and we already have $r\in E(s)$, the construction gives $(s,a,Z,q)\in \mu$, so that $(s,\ysi,\yal)=\ypdatrtf{(s,a\yde,\yal)}{}{\yB}{(q,\yde,\yga)}$.
		But $m<n$, and since $\yde\neq \yeps$,  the induction gives $\ypdatrtf{(q,\yde,\yga)}{\star}{\yB}{(t,\yeps,\ybe)}$.
		Hence, $\ypdatrtf{(s,\ysi,\yal)}{}{\yB}{(t,\yeps,\yga)}$ and the claim holds.
	\end{description}
	Now let $\ysi\in L(\yA)$.
	Then, $\ypdatrtf{(s_0,\ysi,\bot)}{\star}{\yA}{(f,\yeps,\yga)}$ where $s_0\in S_{in}$ and $f\in F$.
	When $\ysi=\yeps$ we get $\yga=\bot$ and $f\in E(s_0)$.
	Since $s_0\in S_{in}$ we get $f\in Q_{in}$.
	Clearly, $\ypdatrtf{(f,\yeps,\bot)}{0}{\yB}{(f,\yeps,\bot)}$ and so $\ysi=\yeps\in L(\yB)$.
	When $\ysi\neq \yeps$,  Claim 1 gives $\ypdatrtf{(s_0,\ysi,\bot)}{\star}{\yB}{(f,\yeps,\yga)}$.
	Since $s_0\in E(s_0)$, the construction gives $s_0\in Q_{in}$.
	Hence, $\ysi\in L(\yB)$.
	We may now conclude that $L(\yA)\ysse L(\yB)$.
	
	For the converse, we state
	\begin{description}
		\item[\sc Claim 2:] If $\ypdatrtf{(s,\ysi,\yal)}{n}{\yB}{(t,\yeps,\ybe)}$ with $n\geq 0$, then 
		$\ypdatrtf{(s,\ysi,\yal)}{\star}{\yA}{(t,\yeps,\ybe)}$.\\
		{\sc Proof:} When $n=0$ the result follows immediately. Let $n>0$.
		Since $\yB$ has no $\yeps$-moves, we must have $\ysi=a\yde$ with $a\in A$, and 
		$$(s,\ysi,\yal)=\ypdatrtf{(s,a\yde,\yal)}{}{\yB}{(r,\yde,\yga)}\ypdatrtf{}{n-1}{\yB}{(t,\yeps,\ybe)},$$
		where the first move was through a transition $(s,a,Z,r)\in \mu$. 
		By the construction, we must have $(p,a,Z,q)\in \rho$ with $p\in E(s)$ and $r\in E(q)$.
		This gives $\ypdatrtf{(s,a\yde,\yal)}{\star}{\yA}{(p,a\yde,\yal)}$ and $\ypdatrtf{(q,\yde,\yga)}{\star}{\yA}{(r,\yde,\yga)}$.
		Composing we get 
		$$(s,\ysi,\yal)=\ypdatrtf{(s,a\yde,\yal)}{\star}{\yA}{(p,a\yde,\yal)}\ypdatrtf{}{}{\yA}{(q,\yde,\yga)}\ypdatrtf{}{\star}{\yA}{(r,\yde,\yga)}.$$
		Using the induction hypothesis we get $\ypdatrtf{(r,\yde,\yga)}{\star}{\yA}{(t,\yeps,\ybe)}$.
		So, $\ypdatrtf{(s,\ysi,\yal)}{\star}{\yA}{(t,\yeps,\ybe)}$ and the claim holds.
	\end{description}
	Now let $\ysi\in L(\yB)$, so that $\ypdatrtf{(q_0,\ysi,\bot)}{\star}{\yB}{(t,\yeps,\yga)}$, with $q_0\in Q_{in}$ and $t\in F$.
	Using Claim 2, we may write $\ypdatrtf{(q_0,\ysi,\bot)}{\star}{\yA}{(t,\yeps,\yga)}$.
	Also, the construction gives some $s_0\in S_{in}$ with $q_0\in E(s_0)$. 
	Hence,  $\ypdatrtf{(s_0,\ysi,\bot)}{\star}{\yA}{(q_0,\ysi,\bot)}$. 
	Composing we get $\ypdatrtf{(s_0,\ysi,\bot)}{\star}{\yA}{(t,\yeps,\yga)}$, that is $\ysi\in L(\yA)$.
	This shows that $L(\yB)\ysse L(\yA)$.
	
	We now have $L(\yA)=L(\yB)$, and the proof is complete. 
\end{proof}
 
 Next we define deterministic VPAs.
 Determinism captures the idea that there is at most one computation for a given input string.
 Since our notion of a VPA extends the original~\cite{alurm-visibly-2004}, we also have to deal with $\yeps$-moves.
 \begin{defi}\label{def:vpa-determinism}
 	Let $\yA=\yvpaA$ be a VPA. 
 	We say that $\yA$ is \emph{deterministic} if $\vert S_{in}\vert\le 1$,  and the following conditions hold: 
 	\begin{enumerate}
 		\item $(p,x,Z_i,q_i)\in\rho_c$ with $i=1,2$, implies $Z_1=Z_2$ and $q_1=q_2$; 
 		\item $(p,x,Z,q_i)\in\rho_r\cup \rho_i$ with $i=1,2$,  implies $q_1=q_2$; 
 		\item  $(p,x,Z,q_1)\in\rho$ with $x\neq \yeps$, implies $(p,\yeps,\yait,q_2)\not\in\rho$ for all $q_2\in S$.
 	\end{enumerate}
 	A language $L$ is a \emph{deterministic VPL} if $L=L(\yA)$ for some deterministic VPA $\yA$.
 \end{defi}
 
 It is worth noticing that Definition~\ref{def:vpa-determinism} does not prohibit $\yeps$-moves in deterministic VPAs.
 With deterministic VPAs, however,  computations are always  unique, as expected.
 \begin{prop}\label{prop:vpa-determ}
 	Let $\yA=\yvpaA$ be a deterministic VPA, let $(p,\ysi,\yal)\in\ypdaconf{\yA}$, and take $n\geq 0$.
 	Then $\ypdatrtf{(p,\ysi,\yal)}{n}{}{(q_i,\mu_i,\yga_i)}$, $i=1,2$, implies 
 	$(q_1,\mu_1,\yga_1)=(q_2,\mu_2,\yga_2)$.
 \end{prop}
 \begin{proof}
 	By induction on $n\geq 0$. When $n=0$ the result is immediate.
 	Take $n>0$. 
 	Using the induction hypothesis we can write $\ypdatrtf{(p,\ysi,\yal)}{n-1}{}{(q,\mu,\yga)}\ypdatrtf{}{1}{}{(q_i,\mu_i,\yga_i)}$, where the last single steps 
 	used the transitions $(q,x_i,Z_i,q_i)\in\rho$, $i=1,2$.
 	
 	First assume that $x_1=\yeps$. Then Definition~\ref{def:pda} gives $(q,x_1,Z_1,q_1)=(q,\yeps,\yait,q_1)\in\rho_i$.
 	From $\ypdatrtf{(q,\mu,\yga)}{}{}{(q_1,\mu_1,\yga_1)}$ and Definition~\ref{def:fsa-move}, we conclude that $\mu_1=\mu$ and $\yga=\yga_1$.
 	Since $\yA$ is deterministic, Definition~\ref{def:vpa-determinism}(3) forces
 	$x_2=\yeps=x_1$, and then Definition~\ref{def:pda} gives $(q,x_2,Z_2,q_2)=(q,\yeps,\yait,q_2)\in \rho_i$.
 	So, from $\ypdatrtf{(q,\mu,\yga)}{}{}{(q_2,\mu_2,\yga_2)}$ we now get $\mu=\mu_2$ and $\yga=\yga_2$.
 	Since $\yA$ is deterministic, using Definition~\ref{def:vpa-determinism}(2) we get $q_1=q_2$, and then we have $(q_1,\mu_1,\yga_1)=(q_2,\mu_2,\yga_2)$.
 	
 	Likewise when $x_2=\yeps$.
 	We can now assume  $x_1$, $x_2\in A$.
 	From $\ypdatrtf{(q,\mu,\yga)}{}{}{(q_i,\mu_i,\yga_i)}$ and Definition~\ref{def:fsa-move}  we get  $x_1\mu_1=\mu=x_2\mu_2$.
 	Hence, $x_1=x_2=x$, and $\mu_1=\mu_2=\yde$.
 	It suffices to verify that $q_1=q_2$ and $\yga_1=\yga_2$. 
 	We have $\ypdatrtf{(q,x\yde,\yga)}{}{}{(q_i,\yde,\yga_i)}$ using the transitions $(q,x,Z_i,q_i)\in\rho$. 
 	There are three cases:
 	\begin{description}
 		\item[{\sc Case 1:}] $x\in A_c$.
 		Since $\yA$ is deterministic and $(q,x,Z_i,q_i)\in\rho_c$,  Definition~\ref{def:vpa-determinism}(1) guarantees that $Z_1=Z_2$ and $q_1=q_2$.
 		From $\ypdatrtf{(q,x\yde,\yga)}{}{}{(q_i,\yde,\yga_i)}$  and Definition~\ref{def:fsa-move}(1) we get 
 		$\yga_1=Z_1\yga$ and $\yga_2=Z_2\yga$, so that $\yga_1=\yga_2$, and we are done.
 		
 		\item[{\sc Case 2:}] $x\in A_i$.
 		From Definition~\ref{def:pda} we get $(q,x_i,Z_i,q_i)=(q,x,\yait,q_i)\in\rho$, $i=1,2$.
 		Then, Definition~\ref{def:vpa-determinism}(2) forces $q_1=q_2$.
 		Since  $\ypdatrtf{(q,x\yde,\yga)}{}{}{(q_i,\yde,\yga_i)}$ Definition~\ref{def:fsa-move}(3) implies
 		$\yga_1=\yga=\yga_2$, concluding this case.
 		
 		\item[{\sc Case 3:}] $x\in A_r$. 
 		First, assume $\yga=\bot$.
 		Then $\ypdatrtf{(q,x\yde,\bot)}{}{}{(q_i,\yde,\yga_i)}$ using the transitions $(q,x,Z_i,q_i)\in\rho_r$.
 		Definition~\ref{def:fsa-move}(2) gives $\yga_1=\bot=\yga_2$ and $Z_1=\bot=Z_2$.
 		So, $(q,x,\bot,q_1)$, $(q,x,\bot,q_2)\in\rho_r$, and Definition~\ref{def:vpa-determinism}(2) implies $q_1=q_2$.
 		Now, assume $\yga\neq \bot$.
 		Definition~\ref{def:fsa-move}(2) now gives $\yga=Z_1\yga_1$ and $\yga=Z_2\yga_2$, so that $Z_1=Z_2$ and $\yga_1=\yga_2$.
 		Once again  Definition~\ref{def:vpa-determinism}(2) implies $q_1=q_2$, and we conclude this case. 
 	\end{description}
 	The proof is complete. 
 \end{proof}
 
 The next result shows that we can remove $\yeps$-moves, while still preserving determinism.
 \begin{prop}\label{prop:no-eps-determ}
 	Given a deterministic VPA $\yA$, we can obtain an equivalent deterministic VPA $\yB$ with no $\yeps$-moves and with the same number of states. 
 \end{prop}
\begin{proof}
	Let $\yA=\yvpaA$.
	First we eliminate cycles of $\yeps$-transitions, and in a second step we eliminate the remaining $\yeps$-transitions.
	Let 
	\begin{align}\label{eq:2-23a}
		(s_1,\yeps,\yait,s_{2}),   (s_2,\yeps,\yait,s_{3}), \ldots, (s_{k-1},\yeps,\yait,s_{k}),  (s_k,\yeps,\yait,s_{k+1}),
	\end{align}
	with $s_{k+1}=s_1$ and $k\geq 1$, be a $\yeps$-cycle in $\yA$.
	We construct the new VPA $\yB=\yvpa{Q}{Q_{in}}{A}{\yGa}{\mu}{E}$, mapping the cycle into a state that occurs in it, say into $s_1$.
	Let $T=\{ (s_j,\yeps,\yait,s_{j+1}): j=1,\ldots, k\}$, $J=\{s_j: j=1,\ldots, k\}$.
	Start with $\yB=\yA$ and transform $\yB$ as follows:
	\begin{enumerate}
		\item[(a)] Transitions: (i) remove $T$ from $\mu$; (ii) for all $(p,x,Z,q)\in\rho$ with $p\not\in J$ and $q\in J$
		remove  $(p,x,Z,q)$ from $\mu$ and add $(p,x,Z,s_1)$ to $\mu$.
		\item[(b)] States: (i) remove $s_j$ from $Q$ and $E$, $j=2,\ldots, k$; (ii) if $F\cap J\neq \yemp$, then add $s_1$ to $E$; (iii) 
		if $S_{in}\cap J\neq \yemp$, then let $Q_{in}=\{s_1\}$.
	\end{enumerate}
	We have to show that $\yB$ is equivalent to $\yA$, and that $\yB$ is deterministic.
	\begin{description}
		\item[\sf Claim 1:] $\yB$ is deterministic.
		
		{\sc Proof.}
		Since $\yA$ is deterministic, $\vert S_{in}\vert\leq 1$.
		Hence, it is clear that $\vert Q_{in}\vert\leq 1$.
		We now  look at each conditions in Definition~\ref{def:vpa-determinism}.
		
		Suppose that $t_i=(p,x,Z_i,q_i)\in \mu_c$, $i=1,2$.
		If $t_1, t_2\in\rho$, the determinism of $\yA$ immediately gives $Z_1=Z_2$, $q_1=q_2$ and so condition (1) in Definition~\ref{def:vpa-determinism} holds.
		With no loss of generality assume $t_1\not\in \rho$.
		Then item (a) of the construction forces $p\not\in J$, $q_1=s_1$, and $t'_1=(p,x,Z_1,s_j)\in \rho$ for some $s_j\in J$.  
		If $t_2\in\rho$ then, because $\yA$ is deterministic, condition (1) in Definition~\ref{def:vpa-determinism} implies $Z_1=Z_2$ and $q_2=s_j$, so that $t_2=(p,x,Z_2,s_j)\in \mu_c$.
		If $j\geq 2$, we see that $(p,x,Z_1,s_j)\in \mu_c$ contradicts item (a) of the construction.
		Hence, $j=1$ and then $q_2=s_j=s_1=q_1$.
		Together with $Z_1=Z_2$ we see that  condition (1) in Definition~\ref{def:vpa-determinism} holds.
		Assume now $t_2\not\in \rho$.
		We obtain again $q_2=s_1$ and $t'_2=(p,x,Z_2,s_\ell)\in\rho$ with $s_\ell\in J$.
		Now $t'_1,t'_2\in\rho$ and the determinism of $\yA$ forces $Z_1=Z_2$.
		We now have $q_1=s_1=q_2$, $Z_1=Z_2$ and condition (1) in Definition~\ref{def:vpa-determinism} holds again. 
		
		Suppose that $t_i=(p,x,Z,q_i)\in \mu_r\cup\mu_i$, $i=1,2$.
		An argument entirely similar to the one in the preceding paragraph shows that  condition (2) in Definition~\ref{def:vpa-determinism} holds. 
		
		Now let $t_1=(p,x,Z,q_1)\in \mu$ and $t_2=(p,\yeps,\yait,q_2)\in \mu$, with $x\neq \yeps$.
		If $t_1,t_2\in\rho$ we get an immediate contradiction to the determinism of $\yA$.
		Assume $t_1\not\in\rho$. 
		As before the construction gives $t'_1=(p,x,Z,s_j)\in \rho$ for some $s_j\in J$.
		Since  $x\neq \yeps$, if $t_2\in\rho$ we get a contradiction to the determinism of $\yA$.
		Hence, $t_2\not\in\rho$ implies, by the construction, that $t'_2=(p,\yeps,\yait,s_\ell)\in \rho$ for some $s_\ell\in J$.
		Again, $x\neq \yeps$ leads to a contradiction to the determinism of $\yA$.
		Finally, when $t_1\in\rho$ the determinism of $\yA$ will force $t_2\not\in\rho$, and the construction will give $t'_2=(p,\yeps,\yait,s_\ell)\in \rho$
		for some $s_\ell\in J$.
		Then, $t_1, t'_2\in\rho$ contradicts the determinism of $\yA$ again because $x\neq \yeps$.
		We conclude that  $t_1=(p,x,Z,q_1)\in \mu$ and $t_2=(p,\yeps,\yait,q_2)\in \mu$, with $x\neq \yeps$, cannot happen and condition (3) in Definition~\ref{def:vpa-determinism} also holds. 
	\end{description}
	
	Next, we want to argue for language equivalence.	
	First we show that $\yA$ can imitate runs of $\yB$.
	\begin{description}
		\item[\sf Claim 2:] If $\ypdatrtf{(q,\ysi,\yal_1\bot)}{\star}{\yB}{(p,\yeps,\yal_2\bot)}$ then we also have $\ypdatrtf{(q,\ysi,\yal_1\bot)}{\star}{\yA}{(p,\yeps,\yal_2\bot)}$.
		
		{\sc Proof:}	Write $\ypdatrtf{(q,\ysi,\yal_1\bot)}{n}{\yB}{(p,\yeps,\yal_2\bot)}$, where $n\geq 0$.
		If all transitions used in the run over $\yB$ are in $\rho$, there is nothing to prove.
		Let $t=(u,x,Z,v)$, with $x \in A$ and $Z\in\yGa\cup\{\yeps\}$, be the first transition not in $\rho$  used in the run over $\yB$.
		By item (a) of the construction we get $v=s_1$ and $t'=(u,x,Z,s_j)\in\rho$, where $s_j\in J$.
		Then, $\ysi=\ysi_1x\ysi_2$ and we may write
		$\ypdatrtf{(q,\ysi_1x\ysi_2,\yal_1\bot)}{\star}{\yA}{}{(u,x\ysi_2,\ybe\bot)}\ypdatrtf{}{1}{\yB}{(s_1,\ysi_2,\yga\bot)}\ypdatrtf{}{k}{\yB}{(p,\yeps,\yal_2\bot)}$, with $0\leq k<n$.
		Using the transitions in the $\yeps$-cycle at Eq.(\ref{eq:2-23a}) we get $\ypdatrtf{(s_j,\ysi_2,\yga\bot)}{\star}{\yA}{}{(s_1,\ysi_2,\yga\bot)}$.
		Thus, using $t'$ we have 
		$$\ypdatrtf{(q,\ysi_1x\ysi_2,\yal_1\bot)}{\star}{\yA}{}{(u,x\ysi_2,\ybe\bot)}\ypdatrtf{}{1}{\yA}{(s_j,\ysi_2,\yga\bot)}\ypdatrtf{}{\star}{\yA}{(s_1,\ysi_2,\yga\bot)}\ypdatrtf{}{k}{\yB}{(p,\yeps,\yal_2\bot)}.$$
		Inductively, we have $\ypdatrtf{(s_1,\ysi_2,\yga\bot)}{\star}{\yA}{(p,\yeps,\yal_2\bot)}$.
		Putting it together, we now have $\ypdatrtf{(q,\ysi_1x\ysi_2,\yal_1\bot)}{\star}{\yA}{(p,\yeps,\yal_2\bot)}$.
	\end{description}
	Next, we want to show that if we have a run of $\yA$ starting at $S_{in}$ and ending in $F$, then we can extend the run to start at $Q_{in}$ and end in $E$.
	\begin{description}
		\item[\sf Claim 3:] If $\ypdatrtf{(r,\ysi,\bot)}{\star}{\yA}{(f,\yeps,\yal\bot)}$ with $r\in S_{in}$, $f\in F$, then  $\ypdatrtf{(u,\ysi,\bot)}{\star}{\yA}{(v,\yeps,\yal\bot)}$ with $u\in Q_{in}$, $v\in E$.
		
		{\sc Proof.}
		Assume $\ysi \in L(\yA)$ so that  $\ypdatrtf{(r,\ysi,\bot)}{\star}{\yA}{(f,\yeps,\yal\bot)}$ with $r\in S_{in}$, $f\in F$.
		First, we want to show that we also have $\ypdatrtf{(u,\ysi,\bot)}{\star}{\yA}{(v,\yeps,\yal\bot)}$ with $u\in Q_{in}$ and $v\in E$.
		If $r\in Q_{in}$ let $u=r$.
		If $r\not\in Q_{in}$, because  $\vert S_{in}\vert\leq 1$ item (b) of the construction implies that $Q_{in}=\{s_1\}$ and $r=s_\ell\in S_{in}$ for some $s_\ell\in J$.
		Since $\yA$ is deterministic, condition (3) in Definition~\ref{def:vpa-determinism} says that $(s_i,\yeps,\yait,s_{i+1})$ is the unique transition out of $s_i$, for all $s_i\in J$.
		Let $u=s_1$.
		Hence, we must have $\ypdatrtf{(r,\ysi,\bot)}{\star}{\yA}{(u,\ysi,\bot)}\ypdatrtf{}{\star}{\yA}{(f,\yeps,\yal\bot)}$ with $u\in Q_{in}$.
		If $f\in E$, let $v=f$.
		If $f\not\in E$, item (b) of the construction says that $f=s_\ell$, for some $s_\ell\in J$, and $s_1\in E$.
		Let $v=s_1$.
		We now have $\ypdatrtf{(u,\ysi,\bot)}{\star}{\yA}{(f,\ysi,\yal\bot)}\ypdatrtf{}{\star}{\yA}{(v,\ysi,\yal\bot)}$ with $v\in E$.
		We can now assume $\ypdatrtf{(r,\ysi,\bot)}{\star}{\yA}{(f,\yeps,\yal\bot)}$ with $r\in Q_{in}$, $f\in E$. 
	\end{description}
	Now we are ready for language equivalence.
	\begin{description}
		\item[\sf Claim 4:] $L(\yA)=L(\yB)$.
		
		{\sc Proof.} 
		Now let $\ysi\in L(\yB)$, so that we have $\ypdatrtf{(q,\ysi,\bot)}{n}{\yB}{(p,\yeps,\yal\bot)}$ with $q\in Q_{in}$, $p\in E$, $n\geq 0$.
		Using Claim 2, we can write $\ypdatrtf{(q,\ysi,\bot)}{\star}{\yA}{(p,\yeps,\yal\bot)}$. 
		If $q\not\in S_{in}$ then by item (b) of the construction we must have $s_\ell\in S_{in}$ for some $s_\ell\in J$ and $q=s_1$.
		From the $\yeps$-cycle at Eq. (\ref{eq:2-23a}) we have $\ypdatrtf{(s_\ell,\ysi,\bot)}{\star}{\yA}{(s_1,\ysi,\bot)}=(q,\ysi,\bot)$, so that $\ypdatrtf{(s_\ell,\ysi,\bot)}{\star}{\yA}{(p,\yeps,\yal\bot)}$.
		Likewise, if $p\not\in F$ then by item (b) again we must have $p=s_1$ and $s_i\in F$ for some $s_i\in J$.
		Again,  the $\yeps$-cycle  gives $(p,\yeps,\yal\bot)=\ypdatrtf{(s_1,\yeps,\yal\bot)}{\star}{\yA}{(s_i,\yeps,\yal\bot)}$.
		Composing, we obtain $\ypdatrtf{(s_\ell,\ysi,\bot)}{\star}{\yA}{(s_i,\yeps,\yal\bot)}$ with $s_\ell\in S_{in}$ and $s_i\in F$.
		Thus we always have $\ypdatrtf{(u,\ysi,\bot)}{\star}{\yA}{(v,\yeps,\yal\bot)}$, with $u\in S_{in}$ and $v\in F$.
		Hence, $\ysi\in L(\yA)$.
		
		Now let $\ysi\in L(\yA)$.
		Using Claim 3, we can write $\ypdatrtf{(r,\ysi,\bot)}{\star}{\yA}{(f,\yeps,\yal\bot)}$ with $r\in Q_{in}$, $f\in E$. 
		If no state $s_j\in J$ occurs in this run, then all transitions are also in $\mu$, and we  get $\ypdatrtf{(r,\ysi,\bot)}{\star}{\yB}{(f,\yeps,\yal\bot)}$.
		Next, let  $\ysi=\ysi_1\ysi_2$ with $\ypdatrtf{(r,\ysi_1\ysi_2,\bot)}{\star}{\yA}{(s_j,\ysi_2,\ybe\bot)}\ypdatrtf{}{\star}{\yA}{(f,\yeps,\yal\bot)}$, where this is the first occurrence of a state of $J$ in the run  over $\yA$.
		Invoking condition (3) in Definition~\ref{def:vpa-determinism} again, we must have $f=s_i\in J$, $\ysi_2=\yeps$ and $\ybe=\yal$.
		Since $f\in E$, we must have $f=s_1\in E$.
		We now have $\ysi=\ysi_1$ and $\ypdatrtf{(r,\ysi,\bot)}{\star}{\yA}{(s_1,\yeps,\yal\bot)}$, $r\in Q_{in}$, $s_1\in E$.
		When $n=0$ we get  $\ypdatrtf{(r,\ysi,\bot)}{0}{\yB}{(s_1,\yeps,\yal\bot)}$ and we are done.
		Else, we must have $\ysi =\ysi'x$ with $x\in\yGa\cup\{\yeps\}$ and $\ypdatrtf{(r,\ysi'x,\bot)}{\star}{\yA}{(t,x,\yga\bot)}\ypdatrtf{}{1}{\yA}{(s_1,\yeps,\yal\bot)}$ with a transition $(t,x,Z,s_1)\in\rho$ used in the last move.
		Because $s_j$ is the first occurrence of a state in $J$ we get $t\not\in J$, and all transitions in $\ypdatrtf{(r,\ysi'x,\bot)}{\star}{\yA}{(t,x,\yga\bot)}$ are in $\mu$.
		Hence, $\ypdatrtf{(r,\ysi'x,\bot)}{\star}{\yB}{(t,x,\yga\bot)}$.
		Item (a) of the construction readily gives $(t,x,Z,s_1)\in\mu$.
		Thus, $\ypdatrtf{(r,\ysi,\bot)}{\star}{\yB}{(t,x,\yga\bot)}\ypdatrtf{}{1}{\yB}{(s_1,\yeps,\yal\bot)}$, $r\in Q_{in}$, $s_1\in E$. 
		We conclude that $L(\yA)\ysse L(\yB)$.
	\end{description}	
	At this point we know how to remove an $\yeps$-cycle from $\yA$, while maintain language equivalence and determinism. 
	Thus, we can repeat the construction for all $\yeps$-cycles in $\yA$, so that we can now assume that there is no $\yeps$-cycle in $\yA$.
	As a final step, we show how to remove all remaining $\yeps$-transitions from $\yA$.
	
	Let $\yA=\yvpaA$ be deterministic and with no $\yeps$-cycles.
	Let $t=(p,\yeps,\yait,q)\in\rho$. 
	Since $\yA$ has no $\yeps$-cycles, we can assume that for all transitions $(q,x,Z,r)\in \rho$ we have $x\neq \yeps$.
	Start with $\yB=\yA$ and transform $\yB$ as follows.
	\begin{itemize}
		\item[(c)] Transitions: (i) $\mu=\rho-\{t\}$; (ii) for all $(q,a,Z,r)\in\rho$, add $(p,a,Z,r)$ to $\mu$.
		\item[(d)] States: (i) if $p\in S_{in}$, $p\not\in F$, let $Q_{in}=\{q\}$; (ii) if $q\in F$, let $E=F\cup\{p\}$.
	\end{itemize}
	We still have determinism.
	\begin{description}
		\item[\sf Claim 5:] $\yB$ is deterministic.\\
		{\sc Proof.} 
		We want to show that $\yB$ is also deterministic.
		It is clear that $\vert Q_{in}\vert\leq 1$ because $\yA$ is deterministic.
		Assume that $\yB$ is not deterministic.
		Since $\yA$ is already deterministic and the only transitions added to $\mu$ are of the form $(p,a,Z,r)$, we see that the only possibility for $\yB$ not deterministic is that we have two transitions $t_i=(p,a_i,Z_i,r_i)\in\mu$, $i=1,2$ in violation  Definition~\ref{def:vpa-determinism}. 
		Since $t\in\rho$ and $\yA$ is deterministic, according to condition (3) of Definition~\ref{def:vpa-determinism}, this is the \emph{only} transition out of $p$ in $\rho$. 
		Moreover $t$ was removed from $\rho$ according to item (c) of the construction.
		Therefore, $t_1$, $t_2$ are new transitions added to $\mu$ by the construction.
		Hence, we must have transitions $t'_i=(q,a_i,Z_i,r_i)\in\rho$, $i=1,2$.
		But these two transitions in $\rho$ also violate Definition~\ref{def:vpa-determinism}, and so they contradict the determinism of $\yA$.
		We conclude that $\yB$ must be deterministic too.
	\end{description}	
	We complete the proof arguing for language equivalence.
	Again, we must first relate runs in $\yA$ to runs in $\yB$, and vice-versa.
	\begin{description}
		\item[\sf Claim 6:] If $\ypdatrtf{(s,\ysi,\ybe\bot)}{\star}{\yB}{(f,\yeps,\yal\bot)}$, then we also have $\ypdatrtf{(s,\ysi,\ybe\bot)}{\star}{\yA}{(f,\yeps,\yal\bot)}$.\\
		{\sc Proof.} 
		First, we show that for any run in $\yB$ there exists a similar run in $\yA$.
		Let $\ypdatrtf{(s,\ysi,\ybe\bot)}{n}{\yB}{(f,\yeps,\yal\bot)}$.
		If $n=0$ we can immediately write $\ypdatrtf{(s,\ysi,\ybe\bot)}{0}{\yA}{(f,\yeps,\yal\bot)}$.
		Now let $x\in A_\yeps$, $\ysi=x\ysi'$ and write
		\begin{align}\label{eq:vpadet1}
			\ypdatrtf{(s,x\ysi',\ybe\bot)}{1}{\yB}{(r,\ysi',\yga\bot)}\ypdatrtf{}{n-1}{\yB}{(f,\yeps,\yal\bot)},
		\end{align}
		with $t'=(s,x,Z,r)$ the first transition used in this run.
		Inductively,  $\ypdatrtf{(r,\ysi',\yga\bot)}{\star}{\yA}{(f,\yeps,\yal\bot)}$.
		If $t'\in\rho$, we have  $\ypdatrtf{(s,x\ysi',\ybe\bot)}{1}{\yA}{(r,\ysi',\yga\bot)}$, so that 
		$\ypdatrtf{(s,\ysi,\ybe\bot)}{\star}{\yA}{(f,\yeps,\yal\bot)}$.
		If $t'\not\in\rho$, the construction gives $s=p$, $t'=(p,x,Z,r)$ and $(q,x,Z,r)\in\rho$.
		Then, using $t=(p,\yeps,\yait,q)$ we have
		$$(s,x\ysi',\ybe\bot)=\ypdatrtf{(p,x\ysi',\ybe\bot)}{1}{\yA}{(q,x\ysi',\ybe\bot)}\ypdatrtf{}{1}{\yA}{(r,\ysi',\yga\bot)}\ypdatrtf{}{\star}{\yA}{(f,\yeps,\yal\bot)}.$$
	\end{description}	
	Now we examine how $\yB$ can imitate runs in $\yA$.
	\begin{description}
		\item[\sf Claim 7:] 
		Let $\ypdatrtf{(s_1,\ysi,\yal_1\bot)}{n}{\yA}{(s_2,\yeps,\yal_2\bot)}$ with $n\geq 0$.
		We show that $\ypdatrtf{(s_1,\ysi,\yal_1\bot)}{\star}{\yB}{(s'_2,\yeps,\yal_2\bot)}$ where $s'_2=p$ if the last transition in the run from $\yA$ was $t$, else $s_2'=s_2$.\\
		{\sc Proof.}
		If $n=0$ the result is immediate, with $s'_2=s_2$. 
		Next, let $n\geq 1$, $x\in A_\yeps$, $\ysi=\ysi'x$ and
		\begin{align}\label{eq:vpadet3}
			\ypdatrtf{(s_1,\ysi' x,\yal_1\bot)}{n-1}{\yA}{(s_3,x,\yal_3\bot)}\ypdatrtf{}{1}{\yA}{(s_2,\yeps,\yal_2\bot)},
		\end{align}
		where $t'=(s_3,x,Z,s_2)$ is the transition used in the last move in Eq. (\ref{eq:vpadet3}).
		
		The first simple case is when $t'=t$.
		Then, $s_3=p$, $x=\yeps$, $s_2=q$ and $\yal_3=\yal_2$.
		Since $s_3=p$, the last transition in $\ypdatrtf{(s_1,\ysi x,\yal_1\bot)}{n-1}{\yA}{(s_3,x,\yal_3\bot)}$ cannot be $t$, otherwise we would have $s_3=p=q$ and then $t$ would be a simple $\yeps$-cycle in $\yA$, a contradiction.
		Inductively from Eq. (\ref{eq:vpadet3}) we 
		can now write $\ypdatrtf{(s_1,\ysi' x,\yal_1\bot)}{\star}{\yB}{(s_3,x,\yal_3\bot)}=(p,\yeps,\yal_2\bot)$.
		Picking $s'_2=p$, we are done with this case.
		
		Let now that $t'\neq t$.
		This gives $t'=(s_3,x,Z,s_2)\in \yB$.
		If $n-1=0$ in Eq. (\ref{eq:vpadet3}) we get $s_1=s_3$, $\ysi'=\yeps$, $\yal_1=\yal_3$.
		The last move in Eq. (\ref{eq:vpadet3}) gives $(s_3,x,\yal_3\bot)=\ypdatrtf{(s_1,x,\yal_1\bot)}{1}{\yA}{(s_2,\yeps,\yal_2\bot)}$.
		Since $t'\in \yB$ we also have $\ypdatrtf{(s_1,x,\yal_1\bot)}{1}{\yB}{(s_2,\yeps,\yal_2\bot)}$, and the result follows because $t'\neq t$.
		We now proceed under the hypothesis that $n-1>0$.
		Inductively, from Eq. (\ref{eq:vpadet3}) we get
		$\ypdatrtf{(s_1,\ysi' x,\yal_1\bot)}{\star}{\yB}{(s_3',x,\yal_3\bot)}$, where $s'_3=p$ if the last transition, say $t''$, in $\ypdatrtf{(s_1,\ysi' x,\yal_1\bot)}{n-1}{\yA}{(s_3,x,\yal_3\bot)}$ was $t$, otherwise we must have $s'_3=s_3$.
		If $t''\neq t$, then $t'\in \yB$ and $s'_3=s_3$ give  $(s'_3,x,\yal_3\bot)=\ypdatrtf{(s_3,x,\yal_3\bot)}{1}{\yB}{(s_2,\yeps,\yal_2\bot)}$.
		Thus, $\ypdatrtf{(s_1,\ysi' x,\yal_1\bot)}{\star}{\yB}{(s_2,\yeps,\yal_3\bot)}$.
		Picking $s'_2=s_2$ and recalling that $t'$ is not $t$, we have the desired result.
		Lastly, let $t''=t$.
		
		We can rewrite the first $n-1$ moves in Eq. (\ref{eq:vpadet3}) thus
		\begin{align}\label{eq:vpadet4}
			\ypdatrtf{(s_1,\ysi',\yal_1\bot)}{n-2}{\yA}{(p,\yeps,\yal_3\bot)}\ypdatrtf{}{1}{\yA}{(q,\yeps,\yal_3\bot)}=(s_3,\yeps,\yal_3\bot),
		\end{align}
		so that $q=s_3$.
		Inductively, we get 
		$\ypdatrtf{(s_1,\ysi',\yal_1\bot)}{\star}{\yB}{(p,\yeps,\yal_3\bot)}$, so that we also have $\ypdatrtf{(s_1,\ysi' x,\yal_1\bot)}{\star}{\yB}{(p,x,\yal_3\bot)}$.
		We now have $t'=(s_3,x,Z,s_2)=(q,x,Z,s_2)$ in $\yA$.
		By item (c) of the construction, we must have $(p,x,Z,s_2)$ in $\yB$. 
		Since $\ypdatrtf{(s_3,x,\yal_3\bot)}{1}{\yA}{(s_2,\yeps,\yal_2\bot)}$ using $t'$, we also get
		$\ypdatrtf{(p,x,\yal_3\bot)}{1}{\yB}{(s_2,\yeps,\yal_2\bot)}$.
		Composing, we obtain $\ypdatrtf{(s_1,\ysi' x,\yal_1\bot)}{\star}{\yB}{(s_2,\yeps,\yal_2\bot)}$.
		Picking $s'_2=s_2$ and remembering that we assumed $t'\neq t$, we have the result again.
	\end{description}
	The first half of language equivalence now follows. 
	\begin{description}
		\item[\sf Claim 8:] $L(\yB)\ysse L(\yA)$.\\
		{\sc Proof.}
		Let  $\ypdatrtf{(s,\ysi,\bot)}{\star}{\yB}{(f,\yeps,\yal\bot)}$ with
		$s\in Q_{in}$ and $f\in E$.
		Claim 6 gives 
		\begin{align}\label{eq:vpadet2}
			\ypdatrtf{(s,\ysi,\bot)}{\star}{\yA}{(f,\yeps,\yal\bot)}.
		\end{align}
		If $f\in F$, let $f'=f$ and we get $\ypdatrtf{(s,\ysi,\bot)}{\star}{\yA}{(f',\yeps,\yal\bot)}$ with $f'\in F$.
		If $f\not\in F$, item (d) of the construction says that $f=p$ and $q\in F$. Picking $f'=q$ we have $f'\in F$. Using $t$ again, we have
		$(f,\yeps,\yal\bot)=\ypdatrtf{(p,\yeps,\yal\bot)}{1}{\yA}{(q,\yeps,\yal\bot)}=(f',\yeps,\yal\bot)$ with $f'\in F$.
		So, we can assume that $f\in F$ in Eq. (\ref{eq:vpadet2}).
		If $s\in S_{in}$, then Eq. (\ref{eq:vpadet2}) says that $\ysi\in L(\yA)$, as desired.
		Now assume $s\not\in S_{in}$.
		Since $s\in Q_{in}$, item (d) of the construction implies that $s=q$, $p\in S_{in}$.
		Using transition $t$ again, we now have 
		$\ypdatrtf{(p,\ysi,\bot)}{1}{\yA}{(q,\ysi,\bot)}=(s,\ysi,\bot)$.
		Composing with  Eq. (\ref{eq:vpadet2}) we get $\ypdatrtf{(p,\ysi,\bot)}{\star}{\yA}{(f,\yeps,\yal\bot)}$.
		Because $p\in S_{in}$ and $f\in F$, we see that $\ysi\in L(\yA)$.
	\end{description}	
	The next claim completes the proof.
	\begin{description}
		\item[\sf Claim 9:] $L(\yA)\ysse L(\yB)$.\\
		{\sc Proof.}  
		Let $s\in S_{in}$, $f\in F$, $n\geq 0$, and 		$\ypdatrtf{(s,\ysi,\bot)}{n}{\yA}{(f,\yeps,\yal\bot)}$.
		Assume that $s\not\in Q_{in}$.
		Since $s\in S_{in}$, then item (d) of the construction gives $Q_{in}=\{q\}$, $p\in S_{in}$ and $p\not\in F$.
		Since $\yA$ is deterministic, we get $S_{in}=\{p\}=\{s\}$, so that $p=s$.
		Hence, $s\not\in F$ and then $s\neq f$.
		Thus, $n\geq 1$.
		Since $p=s$ and $\yA$ is deterministic, we see that $t$ is the first transition in the run $\ypdatrtf{(s,\ysi,\bot)}{n}{\yA}{(f,\yeps,\yal\bot)}$.
		Thus, we can write $(s,\ysi,\bot)=\ypdatrtf{(p,\ysi,\bot)}{1}{\yA}{(q,\ysi,\bot)}\ypdatrtf{}{n-1}{\yA}{(f,\yeps,\yal\bot)}$, with $q\in Q_{in}$ and $n-1\geq 0$.
		We can, thus, take $s\in Q_{in}$ in the run  	$\ypdatrtf{(s,\ysi,\bot)}{n}{\yA}{(f,\yeps,\yal\bot)}$.
		
		If $n=0$ we get $\ysi=\yeps$, $\yal=\yeps$ and $s=f$.
		We can then  write 	$(s,\ysi,\bot)=\ypdatrtf{(s,\yeps,\bot)}{0}{\yB}{(f,\yeps,\bot)}$.
		Because $F\ysse E$ we get $f\in E$. 
		Since $s\in Q_{in}$ we get $\ysi=\yeps\in L(\yB)$.
		
		We proceed now with $s\in Q_{in}$, $f\in E$, $n\geq 1$ and 		$\ypdatrtf{(s,\ysi,\bot)}{n}{\yA}{(f,\yeps,\yal\bot)}$.
		Use Claim 7 to write 
		$\ypdatrtf{(s,\ysi,\bot)}{\star}{\yB}{(f',\yeps,\yal\bot)}$,
		where $f'=p$ if $t$ was the last transition  in the run over $\yA$, otherwise we have $f'=f$.
		If $f'=f$, we get $f'\in E$ and then $\ysi\in L(\yB)$.
		Assume now that $f'=p$ and $t$ was the last transition  in the run over $\yA$.
		This gives $q=f$, and then item (d) of the construction says that $p\in E$.
		Thus, $f'\in E$ and we get again $\ysi\in L(\yB)$.
	\end{description}	
\end{proof}

 %
 
 \subsection{The Synchronous Product of VPAs}\label{subsec:vpa-product}
 
 The product of two VPAs captures their synchronous behavior. 
 It will be useful when testing conformance between pushdown memory models. 
 \begin{defi}\label{def:productVPA}
 	Let $\yS=\yvpaA$  and $\yQ=\yvpa{Q}{Q_{in}}{A}{\yDe}{\mu}{G}$ be  VPAs with  a common alphabet. 
 	Their product is the VPA
 	$\yS\times \yQ=\yvpa{S\times Q}{S_{in}\times Q_{in}}{A}{\yGa\times\yDe}{\nu}{F\times G}$, where $((s_1,q_1),a,Z,(s_2,q_2))\in\nu$ if and only if either:
 	\begin{enumerate}
 		\item  $a\neq\yeps$ and we have $(s_1,a,Z_1,s_2)\in \rho$ and  $(q_1,a,Z_2,q_2)\in \mu$ with $Z=(Z_1,Z_2)$ or $Z=Z_1=Z_2\in\{\bot,\yait\}$; or 
 		\item $a=\yeps$, $Z=\yait$ and we have either $s_1=s_2$ and $(q_1,\yeps,\yait,q_2)\in\mu$; or $q_1=q_2$ and  $(s_1,\yeps,\yait,s_2)\in\rho$. 
 	\end{enumerate}
 \end{defi}
 
 We first note that the product construction preservers both determinism and the absence of $\yeps$-moves when the original VPAs satisfy such conditions.
 \begin{prop}\label{prop:epsilon-deterministic}
 	Let $\yS=\yvpaA$  and $\yQ=\yvpaB$ be VPAs and let $\yP=\yS\times\yQ$ be their product as in Definition~\ref{def:productVPA}.
 	Then the following conditions hold:
 	\begin{enumerate}
 		\item If $\yS$ and $\yQ$ have no $\yeps$-moves, then $\yP$ has no $\yeps$-moves;
 		\item If $\yS$ and $\yQ$ are deterministic with no $\yeps$-moves, then $\yP$ is also deterministic with no $\yeps$-moves.
 	\end{enumerate}
 \end{prop}
 \begin{proof}
 	Write $\yP=\yS\times \yQ=\yvpa{S\times Q}{S_{in}\times Q_{in}}{A}{\Gamma \times \yDe}{\nu}{F\times G}$.
 	
 	When $\yS$ and $\yQ$ have no $\yeps$-moves, only item (1) of Definition~\ref{def:productVPA} applies, so that $\yP$ has no $\yeps$-moves too.
 	Hence, assertion (1) holds.
 	
 	It remains to show that $\yP$ is deterministic when $\yS$ and $\yQ$ are also deterministic and have no $\yeps$-moves.
 	We just argued that $\yP$ has no $\yeps$-moves, so that $\yP$ can not violate condition (3) of Definition~\ref{def:vpa-determinism}.
 	For the sake of contradiction, assume $\yP$ has transitions 
 	$((s,q),a,Z_i,(p_i,r_i))\in\nu$, where $Z_i\in\yGa\times \yDe\cup\{\bot,\yait\}$, $i=1,2$, and $a\neq \yeps$.
 	Using Definition~\ref{def:productVPA} (1), from $((s,q),a,Z_1,(p_1,r_1))$ we get  
 	\begin{align}
 		&(s,a,X_1,p_1)\in\rho,\,\,\, (q,a,Y_1,r_1)\in\mu\label{prop2.26a}
 	\end{align}
 	where $X_1\in \yGa\cup{\{\bot,\yait\}}$, $Y_1\in \yDe\cup{\{\bot,\yait\}}$.
 	Likewise, from $((s,q),a,Z_2,(p_2,r_2))$ we get 
 	\begin{align}
 		&(s,a,X_2,p_2)\in\rho,\,\,\, (q,a,Y_2,r_2)\in\mu\label{prop2.26b}
 	\end{align}
 	where $X_2\in \yGa\cup{\{\bot,\yait\}}$, $Y_2\in \yDe\cup{\{\bot,\yait\}}$.
 	
 	If $a\in A_c$, we must have $Z_1,Z_2\in \yGa\times \yDe$, and Definition~\ref{def:productVPA} (1) forces $Z_i=(X_i,Y_i)$, $i=1,2$.
 	Since $a\in A_c$, the determinism of $\yS$, together with Definition~\ref{def:vpa-determinism} (1) applied to Eqs. (\ref{prop2.26a}, \ref{prop2.26b}) gives $X_1=X_2$, $p_1=p_2$.
 	Likewise, the determinism of $\yQ$ gives $Y_1=Y_2$, $r_1=r_2$.
 	But then $Z_1=(X_1,Y_1)=(X_2,Y_2)=Z_2$ and $(p_1,r_1)=(p_2,r_2)$.
 	We conclude that when $a\in A_c$, $\yP$ cannot violate condition (1) of Definition~\ref{def:vpa-determinism}.
 	
 	Now let $a\in A_r\cup A_i$ and $Z_1=Z_2$.
 	If $Z_1=(X,Y)\in \yGa\times \yDe$, Eq. (\ref{prop2.26a}) and Definition~\ref{def:productVPA} (1) imply $X_1=X$ and $Y_1=Y$.
 	Since $Z_2=Z_1$, Eq. (\ref{prop2.26b}) gives $X_2=X$ and $Y_2=Y$.
 	Hence, $X_1=X_2$, $Y_1=Y_2$.
 	Now, the determinism of $\yS$ applied to Eqs. (\ref{prop2.26a}, \ref{prop2.26b}), together with Definition~\ref{def:vpa-determinism} (2) implies $p_1=p_2$.
 	Likewise, $r_1=r_2$.
 	Thus, $(p_1,r_1)=(p_2,r_2)$ and we conclude that $\yP$ does not violate Definition~\ref{def:vpa-determinism}(2).
 	Finally let $Z_1\in \{\bot,\yait\}$.
 	Because $a\neq \yeps$, Definition~\ref{def:productVPA} (1) and Eq. (\ref{prop2.26a}) say that
 	$Z_1=X_1=Y_1$.
 	Likewise, now with Eq. (\ref{prop2.26b}) and knowing that $Z_1=Z_2$, we get and $Z_1=Z_2=X_2= Y_2$.
 	Hence, with $a\in A_r\cup A_i$, using Definition~\ref{def:vpa-determinism} (2) applied to Eqs. (\ref{prop2.26a}, \ref{prop2.26b}), the determinism of $\yS$ and $\yQ$  implies  $p_1=r_1$  and $p_2=r_2$.
 	Again we have $(p_1,r_1)=(p_2,r_2)$ and we conclude that $\yP$ can not violate Definition~\ref{def:vpa-determinism} in any circumstance.
 	That is, $\yP$ is deterministic.
 \end{proof}
 
 It is not hard to see that when $\yS$ and $\yQ$ are deterministic and $\yeps$-moves are allowed in both of them, then the product $\yP$ need not be deterministic.
The following result links moves in the product VPA to moves in the original constituent VPAs.
\begin{prop}\label{prop:product-behavior}
Let $\yS=\yvpaA$  and $\yQ=\yvpa{Q}{Q_{in}}{A}{\yDe}{\mu}{G}$ be VPAs over the same input alphabet. 
Then,  for all $\ysi\in A^\star$, and all $i$, $k\geq 0$:
\begin{quote}
	$$\ypdatrtf{((s,q),\ysi,(X_1,Y_1)\ldots (X_i,Y_i)\bot)}{\star}{\yS\times \yQ}{((p,r),\yeps,(Z_1,W_1)\ldots (Z_k,W_k)\bot)}$$ 
	\begin{center}if and only if\end{center} 
	$$\ypdatrtf{(s,\ysi,X_1\ldots X_i\bot)}{\star}{\yS}{(p,\yeps,Z_1\ldots Z_k\bot)}\,\,\text{and}\,\, \ypdatrtf{(q,\ysi,Y_1\ldots Y_i\bot)}{\star}{\yQ}{(r,\yeps,W_1\ldots W_k\bot)}.$$
\end{quote}
\end{prop}
\begin{proof}
	Let $\ypdatrtf{((s,q),\ysi,\yga_0)}{n}{\yS\times \yQ}{((p,r),\yeps,\yga_n)}$, for some $n\geq 0$, and where
	$\yga_0=(X_1,Y_1)\ldots (X_i,Y_i)\bot$ and $\yga_n=(Z_1,W_1)\ldots (Z_k,W_k)\bot$. 
	When $n=0$ the result is immediate.
	Proceeding inductively, let $n\geq 1$, $\ysi=\yde a$ with $a\in A\cup \{\yeps\}$, and 
	\begin{align}\label{prop2.28a}
		\ypdatrtf{((s,q),\yde a,\yga_0)}{n-1}{\yS\times \yQ}{((u,v),a,\yga_{n-1})}\ypdatrtf{}{1}{\yS\times \yQ}{((p,r),\yeps,\yga_n)},
	\end{align}
	where $\yga_{n-1}=(U_1,V_1)\ldots (U_j,V_j)\bot$, for some $j\geq 0$. 
	In order to ease the notation, let $\yal_0=X_1 \ldots X_i\bot$, $\yal_{n-1}=U_1 \ldots U_j\bot$, $\yal_n=Z_1\ldots Z_k\bot$, $\ybe_0=Y_1\ldots Y_i\bot$, $\ybe_{n-1}=V_1 \ldots V_j\bot$,  and $\ybe_n=W_1\ldots W_k\bot$. 
	Then, using the induction hypothesis we get 
	\begin{align}\label{prop2.28b}
		\ypdatrtf{(s,\yde,\yal_0)}{\star}{\yS}{(u,\yeps,\yal_{n-1})},\quad  \ypdatrtf{(q,\yde,\ybe_0)}{\star}{\yQ}{(v,\yeps,\ybe_{n-1})}. 
	\end{align}
	Let $((u,v),a,Z,(p,r))$ be the transition used in the last $\yS\times \yQ$ move.
	From Definition~\ref{def:productVPA}, there are two possibilities: 
	\begin{description}
		\item[\sc Case 1:] 
		$a\neq \yeps$, $(u,a,X,p)\in\rho$, $(v,a,Y,r)\in\mu$, and either $Z=(X,Y)$, or $Z=X=Y\in\{\bot,\yait\}$. 
		
		If $a\in A_c$, then from Eq. (\ref{prop2.28a}) we get $Z\not\in\{\bot,\yait\}$, so that we must have $Z=(X,Y)$.
		With $a\in A_c$, Eq (\ref{prop2.28a}) gives $\yga_n=(X,Y)\yga_{n-1}=(X,Y)(U_1,V_1)\cdots (U_j,V_j)\bot$.
		But now, with $(u,a,X,p)\in\rho$ and Eq. (\ref{prop2.28b}) we also get 
		$\ypdatrtf{(u,a,\yal_{n-1})}{1}{\yS}{(p,\yeps,X\yal_{n-1} )}$.
		Using Eq. (\ref{prop2.28b}) and composing, we get  $\ypdatrtf{(s,\ysi,\yal_0)}{\star}{\yS}{(p,\yeps,X\yal_{n-1})}$, where $X\yal_{n-1}=XU_1\cdots U_j\bot$.
		Likewise, we obtain $\ypdatrtf{(q,\ysi,\ybe_0)}{\star}{\yQ}{(r,\yeps,Y\ybe_{n-1})}$, with 
		$Y\ybe_{n-1}=YV_1\cdots V_j\bot$, as needed.
		
		When $a\in A_r$,  given that  $((u,v),a,Z,(p,r))$ was the last transition used Eq. (\ref{prop2.28a}) with
		$\yga_{n-1}=(U_1,V_1)\ldots (U_j,V_j)\bot$, we get either (i) $j\geq 1$, $Z=(U_1,V_1)$ and $\yga_n=(U_2,V_2)\cdots (U_j,V_j)\bot$, or (ii) $j=0$ and $Z=\bot$ and $\yga_n=\bot$.
		Since $Z=(X,Y)$, the first case gives  $X=U_1$ and $Y=V_1$.
		We can now repeat the argument above when $a\in A_c$ and reach the desired result.
		In the second case, $Z=\bot$ forces $Z=X=Y=\bot$.
		Also, $\yal_{n-1}=U_1\cdots U_j\bot$ reduces to $\yal_{n-1}=\bot$.
		Since $a\in A_r$ and $(u,a,X,p)\in\rho$, we can write $\ypdatrtf{(u,a,\yal_{n-1})}{1}{\yS}{(p,\yeps,
			\bot )}$.
		Using Eq. (\ref{prop2.28b}) and composing, we get $\ypdatrtf{(s,\ysi,\yal_0)}{\star}{\yS}{(p,\yeps,\bot)}$.
		Since $\yga_n=\bot$ we have the desired result for $\yS$.
		A similar reasoning also gives  $\ypdatrtf{(q,\ysi,\ybe_0)}{\star}{\yQ}{(r,\yeps,\bot)}$, as needed.
		
		If $a\in A_i$, since    $((u,v),a,Z,(p,r))$ was the last transition used Eq. (\ref{prop2.28a}), we get $Z=\yait$ and $\yga_n=\yga_{n-1}=(U_1,V_1)\cdots (U_j,V_j)\bot$.
		Since $Z\in\{\bot,\yait\}$, we must also have $Z=X=Y=\yait$.
		Thus, $(u,a,\yait,p)\in\rho$.
		Using Eq. (\ref{prop2.28b}) and composing, we get $\ypdatrtf{(s,\ysi,\yal_0)}{\star}{\yS}{(p,\yeps,\yal_{n-1})}$.
		Since $\yal_{n-1}=U_1 \ldots U_j\bot$, we get the desired result for $\yS$.
		By a similar reasoning we also get $\ypdatrtf{(q,\ysi,\ybe_0)}{\star}{\yQ}{(r,\yeps,\ybe_{n-1})}$, completing this case.
		
		\item[\sc Case 2:] 
		$a=\yeps$, $Z=\yait$ with either $u=p$ and $(v,\yeps,\yait,r)\in\mu$, or $v=r$ and $(u,\yeps,\yait,p)\in\rho$.
		We look at the first case, the other being entirely similar.
		Since $((u,v),\yeps,\yait,(u,r))$ was the transition used in Eq. (\ref{prop2.28a}) we get $\yga_{n-1}=\yga_n=(U_1,V_1)\cdots (U_j,V_j)\bot$.
		From Eq. (\ref{prop2.28b}) and $u=p$, we know that $\ypdatrtf{(s,\ysi,\yal_0)}{\star}{\yS}{(p,\yeps,\yal_{n-1})}$,
		and because $\yal_{n-1}=U_1 \ldots U_j\bot$ we get the desired result for $\yS$.
		Since $(v,\yeps,\yait,r)\in\mu$, from Eq. (\ref{prop2.28b}) we obtain
		$\ypdatrtf{(q,\ysi,\ybe_0)}{\star}{\yQ}{(r,\yeps,\ybe_{n-1})}$, and we have the result also for $\yQ$ because $\ybe_{n-1}=V_1 \ldots V_j\bot$.
		This case is now complete.
	\end{description}
	
	Next, we look at the converse.
	Let $\ypdatrtf{(s,\ysi,\yal_0)}{n}{\yS}{(p,\yeps,\yal_n)}$  and 
	$\ypdatrtf{(q,\ysi,\ybe_0)}{m}{\yQ}{(r,\yeps,\ybe_m)}$, with $n,m\geq 0$, where $\yal_0=X_1\ldots X_i\bot$, $\ybe_0=Y_1\ldots Y_i\bot$, $\yal_n=Z_1\ldots Z_k\bot$ and  $\ybe_m=W_1\ldots W_k\bot$, for some $i,k\geq 0$.
	Write $\yga_0=(X_1,Y_1)\ldots(X_i,Y_i)\bot$ and $\yga=(Z_1,W_1)\ldots(Z_k,W_k)\bot$.
	
	When $n+m=0$ we get $m=0$ and $n=0$, so that  the result is immediate.
	
	With no loss, let $n\geq 1$.
	Then, $\ysi=a\yde$ and  a transition $(s,a,Z,t)\in\rho$ used in the first step in $\yS$. We now have 
	\begin{equation}\label{eq:1.24a}
		\ypdatrtf{(s,a\yde,\yal_0)}{1}{\yS}{(t,\yde,\yal_1)}\ypdatrtf{}{n-1}{\yS}{(p,\yeps,\yal_n)},\quad \ypdatrtf{(q,\ysi,\ybe_0)}{m}{\yQ}{(q,\yeps,\ybe_m)}.
	\end{equation}
	We first look at the case when $a=\yeps$.  
	Then, $\ysi=\yde$ , $\yal_0=\yal_1$ and  Definition~\ref{def:fsa-move} implies $Z=\yait$.
	Because $(s,\yeps,\yait,t)\in \rho$, Definition~\ref{def:productVPA} item (2) says $((s,q),\yeps,\yait,(t,q))$ is in $\yS\times \yQ$.
	Thus,
	\begin{align}\label{prop2.28c} \ypdatrtf{((s,q),\ysi,\yga_0)}{1}{\yS\times\yQ}{((t,q),\yde,\yga_0)}.
	\end{align}
	Now we have $\ypdatrtf{(t,\yde,\yal_0)}{n-1}{\yS}{(p,\yeps,\yal_n)}$ and $\ypdatrtf{(q,\yde,\ybe_0)}{m}{\yQ}{(r,\yeps,\ybe_m)}$.
	Since $n-1+m<n+m$, inductively, we may write  
	$\ypdatrtf{((t,q),\yde,\yga_0)}{\star}{\yS\times\yQ}{((p,r),\yeps,\yga)}$.
	From Eq. (\ref{prop2.28c}) we get   $\ypdatrtf{((s,q),\ysi,\yga_0)}{\star}{\yS\times\yQ}{((p,r),\yeps,\yga)}$.
	Recalling the definitions of $\yal_0$, $\ybe_0$, $\yal_n$, $\ybe_n$, and of $\yga_0$, $\yga$, we see that the result holds in this case.
	
	We now turn to the case $a\neq \yeps$.
	If $m=0$, from $\ypdatrtf{(q,\ysi,\ybe_0)}{m}{\yQ}{(r,\yeps,\ybe_m)}$ we obtain $\ysi=\yeps$. 
	But then $\ysi=a\yde$ gives $a=\yeps$, a contradiction.
	Thus, $m\geq 1$. 
	This gives 
	\begin{equation}\label{eq:1.24b}
		\ypdatrtf{(q,a\yde,\ybe_0)}{1}{\yQ}{(u,\yde,\ybe_1)}\ypdatrtf{}{m-1}{\yQ}{(r,\yeps,\ybe_m)},
	\end{equation}
	and we must have a transition $(q,a,W,u)\in\mu$ which was used in the first step. 
	There are three cases.
	\begin{description}
		\item[\sc Case 3:] $a\in A_i$. 
		Then, Eqs.~(\ref{eq:1.24a}) and~(\ref{eq:1.24b}) imply  $(s,a,\yait,t)\in\rho$, $\yal_0=\yal_1$ and $(q,a,\yait,u)\in\mu$, $\ybe_0=\ybe_1$.
		From the same equations, we now obtain 
		$\ypdatrtf{(t,\yde,\yal_0)}{n-1}{\yS}{(p,\yeps,\yal_n)}$ and $\ypdatrtf{(u,\yde,\ybe_0)}{n-1}{\yQ}{(r,\yeps,\ybe_m)}$.
		The induction hypothesis now implies $\ypdatrtf{((t,u),\yde,\yga_0)}{\star}{\yS\times\yQ}{((p,r),\yeps,\yga)}$.
		Since $(s,a,\yait,t)\in\rho$ and $(q,a,\yait,u)\in\mu$ and $a\neq \yeps$, Definition~\ref{def:productVPA} item (1)  
		gives $((s,q),a,\yait,(t,u))$ in $\yS\times \yQ$, and we may write $\ypdatrtf{((s,q),a\yde,\yga_0)}{1}{\yS\times\yQ}{((t,q),\yde,\yga_0)}$.
		Composing, we get the result.
		
		\item[\sc Case 4:] $a\in A_c$.
		From  Eq.~(\ref{eq:1.24a}) we get $(s,a,X,t)\in\rho$, $\ypdatrtf{(s,a\yde,\yal_0)}{1}{\yS}{(t,\yde,X\yal_0)}$ for some $X\in\yGa$.
		And from Eq. (\ref{eq:1.24b}) we get $(q,a,Y,u)\in\mu$, $\ypdatrtf{(q,a\yde,\ybe_0)}{1}{\yQ}{(u,\yde,Y\ybe_0)}$for some  $Y\in\yDe$.
		
		Definition~\ref{def:productVPA} item (1) says that $((s,q),a,(X,Y),(t,u))$ is in $\yS\times \yQ$, and we may write 
		$\ypdatrtf{((s,q),a\yde,\yga_0)}{1}{\yS\times\yQ}{((t,u),\yde,(X,Y)\yga_0)}$.
		Inductively, Eqs.~(\ref{eq:1.24a}) and~(\ref{eq:1.24b}) also imply $\ypdatrtf{((t,u),\yde,(X,Y)\yga_0)}{\star}{\yS\times\yQ}{((p,r),\yeps,\yga)}$.
		Composing again, the result follows.
		
		\item[\sc Case 5:] $a\in A_r$, and recall that $\yal_0=X_1\ldots X_i\bot$, $\ybe_0=Y_1\ldots Y_i\bot$, for some $i\geq 0$.
		When $i\geq 1$, the reasoning is very similar to Case 4.
		
		Let $i=0$. 
		From Eq. (\ref{eq:1.24a}) we see that $\yal_0=\yal_1=\bot$, $(s,a,\bot,t)\in\rho$ and $\ypdatrtf{(t,\yde,\bot)}{n-1}{\yS}{(p,\yeps,\yal_n)}$.
		Likewise, From Eq. (\ref{eq:1.24b}) gives $\ybe_0=\ybe_1=\bot$, $(q,a,\bot,u)\in\mu$ and $\ypdatrtf{(u,\yde,\bot)}{m-1}{\yQ}{(r,\yeps,\ybe_m)}$.
		Inductively, we have
		$\ypdatrtf{((t,u),\yde,\bot)}{\star}{\yS\times\yQ}{((p,r),\yeps,\yga)}$.
		Because $(s,a,\bot,t)\in\rho$, $(q,a,\bot,u)\in\mu$,
		Definition~\ref{def:productVPA} item (1) says that $((s,q),a,\bot,(t,u))$
		is in $\yS\times \yQ$.
		Since $\ysi=a\yde$, composing we get $\ypdatrtf{((s,q),\ysi,\bot)}{\star}{\yS\times\yQ}{((p,r),\yeps,\yga)}$, which is the desired result. 
	\end{description}
	The proof is now complete.
\end{proof}


\subsection{Closure Properties over Visibly Pushdown Languages}\label{subsec:properties}

Now we look at some closure properties involving VPLs.
Similar results appeared elsewhere~\cite{alurm-visibly-2004}, but here the VPA models are somewhat more general because they allow for $\yeps$-moves,
which can make a difference in the results.
Moreover, it will prove important to investigate if determinism, when present in the participant VPAs, can also be guaranteed for the resulting VPAs.
Further, since later on we will be analyzing the complexity of certain constructions, we also note how the sizes of the resulting models vary 
as a function of the size of the given models.

First we report on a simple result on the stack size during runs of VPAs. 
\begin{prop}\label{prop:stack-size}
Let $\yA=\yvpaA$ and $\yB=\yvpaB$ be VPAs over a common alphabet $A$.
Consider starting configurations $(s_1,\ysi,\yal_1\bot)$ and $(q_1,\ysi,\ybe_1\bot)$ of $\yA$ and $\yB$, respectively,
and where $\vert \yal_1\vert=\vert\ybe_1\vert$.
If $\ypdatrtf{(s_1,\ysi,\yal_1\bot)}{\star}{\yA}{(s_2,\yeps,\yal_2\bot)}$ and $\ypdatrtf{(q_1,\ysi,\ybe_1\bot)}{\star}{\yB}{(q_2,\yeps,\ybe_2\bot)}$, then we must have $\vert \yal_2\vert=\vert\ybe_2\vert$.
\end{prop}
\begin{proof}
	A simple induction on $n=\vert\ysi\vert$. 
	When $n=0$ we get $\ysi=\yeps$.
	According to Definition~\ref{def:fsa-move},  $\yeps$-moves do not change the stack, and so we get $\yal_1=\yal_2$ and $\ybe_1=\ybe_2$, and the result follows.
	
	Assume $n\geq 1$ and $\ysi=\yde a$, where $a\in A $ and $\vert\yde\vert=n-1$.
	Then computations can be separated thus
	\begin{align*}
		&\ypdatrtf{(s_1,\yde a,\yal_1\bot)}{\star}{\yA}{(s_3,a,\yal_3\bot)}\ypdatrtf{}{1}{\yA}{(s_4,\yeps,\yal_4\bot)}
		\ypdatrtf{}{\star}{\yA}{(s_2,\yeps,\yal_2\bot)}\\ 
		&\ypdatrtf{(q_1,\yde a,\ybe_1\bot)}{\star}{\yB}{(q_3,a,\ybe_3\bot)}\ypdatrtf{}{1}{\yB}{(q_4,\yeps,\ybe_4\bot)}
		\ypdatrtf{}{\star}{\yB}{(q_2,\yeps,\ybe_2\bot)}.
	\end{align*}
	Clearly, we get $\ypdatrtf{(s_1,\yde,\yal_1\bot)}{\star}{\yA}{(s_3,\yeps,\yal_3\bot)}$ and $\ypdatrtf{(q_1,\yde,\ybe_1\bot)}{\star}{\yB}{(q_3,\yeps,\ybe_3\bot)}$.
	The induction hypothesis implies $\vert\yal_3\vert=\vert\ybe_3\vert$.
	The proof will be complete if we can show that $\vert\yal_4\vert=\vert\ybe_4\vert$, because from these configurations onward we have only $\yeps$-moves, which would imply that $\vert \yal_2\vert = \vert\ybe_2\vert$.
	
	Let $(s_3,a,Z,s_4)$ and $(q_3,a,W,q_4)$ be the $\yA$ and $\yB$ transitions, respectively, used in the one-step 
	computations, as indicated above.
	From Definition~\ref{def:fsa-move} we have three simple cases: (i) when $a\in A_i\cup\{\yeps\}$ we get $\yal_3=\yal_4$ and $\ybe_3=\ybe_4$; 
	(ii) when $a\in A_c$ we get $\yal_4=Z\yal_3$ and $\ybe_4=W\ybe_3$; 
	and (iii) when $a\in A_r$, since $\vert \yal_3\vert=\vert\ybe_3\vert$, we get either $\yal_3\neq \yeps\neq\ybe_3$ and then $\yal_3=Z\yal_4$ and $\ybe_3=W\ybe_4$, or 
	$\yal_3= \yeps=\ybe_3$ and then $\yal_4= \yeps=\ybe_4$.
	In any case we see that $\vert\yal_4\vert=\vert\ybe_4\vert$, because we already have $\vert\yal_3\vert=\vert\ybe_3\vert$.
\end{proof}

Using the product construction, we can show that VPLs are also closed under intersection.
\begin{prop}\label{prop:cap-vpa}\label{prop:determ-complement}
Let $\yS=\yvpaA$ and $\yQ=\yvpa{Q}{Q_{in}}{A}{\yDe}{\mu}{G}$ be VPAs with $n$ and $m$ states, respectively. 
Then $L(\yS)\cap L(\yQ)$ can be accepted by a VPA $\yP$ with $mn$ states.
Moreover, if $\yS$ and $\yQ$ are deterministic, then $\yP$ is also deterministic.
\end{prop}
\begin{proof}
Let $\yP=\yS\times \yQ=\yvpa{S\times Q}{S_{in}\times Q_{in}}{A}{\Gamma \times \yDe}{\nu}{F\times G}$
be the product of $\yS$ and $\yQ$.
See Definition~\ref{def:productVPA}.
It is clear that $\yP$ has $nm$ states.

For the language equivalence, assume that $\ysi\in L(\yP)$, so that 
$\ypdatrtf{((s,q),\ysi,\bot)}{\star}{\yP}{((p,r),\yeps,\yga\bot)}$ with $(s,q)\in S_{in}\times Q_{in}$ and 
$(p,r)\in F\times G$, for some $\yga\in (\yGa\times\yDe)^\star$.
Using Proposition~\ref{prop:product-behavior} we get $\ypdatrtf{(s,\ysi,\bot)}{\star}{\yS}{(p,\yeps,\yal\bot)}$ 
for some $\yal\in\yGa^\star$. 
Since $s\in S_{in}$ and $p\in F$, we conclude that $\ysi\in\ L(\yS)$.
Likewise, $\ysi\in\ L(\yQ)$, so that $\ysi\in\ L(\yS)\cap L(\yQ)$.
For the converse, assume that $s\in S_{in}$, $p\in F$ and $\ypdatrtf{(s,\ysi,\bot)}{\star}{\yS}{(p,\yeps,\yal\bot)}$ 
for some $\yal\in\yGa^\star$. Likewise, let $q\in Q_{in}$, $r\in G$ and $\ypdatrtf{(q,\ysi,\bot)}{\star}{\yQ}{(r,\yeps,\ybe\bot)}$ 
for some $\ybe\in\yDe^\star$. 
From Proposition~\ref{prop:stack-size} we get $\vert \yal\vert=\vert \ybe\vert$.
We can now apply Proposition~\ref{prop:product-behavior} and write 
$\ypdatrtf{((s,q),\ysi,\bot)}{\star}{\yP}{((p,r),\yeps,\yga\bot)}$ for some $\yga\in (\yGa\times\yDe)^\star$. 
Since $(p,r)\in F\times G$ we get $\ysi\in L(\yP)$, and the equivalence holds.

Finally, Proposition~\ref{prop:epsilon-deterministic} guarantees that $\yP$ is deterministic when $\yS$ and $\yQ$ are deterministic. 
\end{proof}


Now, we investigate the closure of VPLs under union. 
But first, in order to complete the argument for the union, we need to consider VPAs that can always read any string of input symbols when started at any state and with any stack configuration.
\begin{defi}\label{def:forward} 
Let $\yA=\yvpaA$ be a VPA.
We say that $\yA$ is a \emph{non-blocking} VPA if, for all $s\in S$, all $\ysi\in A^\star$ and all $\yal\in\yGa^\star$, there are $p\in S$ and $\ybe\in\yGa^\star$ such that $\ypdatrtf{(s,\ysi,\yal\bot)}{\star}{\yA}{(p,\yeps,\ybe\bot)}$.
\end{defi}

Any VPA can be easily turned into a non-blocking VPA, with almost no cost in the number of states.
\begin{prop}\label{prop:non-blocking}
Let $\yS=\yvpaA$ be a VPA with $n$ states. 
Then we can construct an equivalent VPA $\yQ$ with at most $n+1$ states and which is also a non-blocking VPA.
Moreover, $\yQ$ is deterministic if $\yS$ is deterministic, and $\yQ$ has no $\yeps$-moves if so does $\yS$.
\end{prop}
\begin{proof}
Let $\yQ=\yvpa{S\cup\{p\}}{S_{in}}{A}{\yGa}{\mu}{F}$ where $p\not\in S$ is a new state.
In order to construct the new transition set $\mu$ as an extension of $\rho$, first pick some stack symbol $Z\in\yGa$.
Then, for all $s\in S$ such that there is no $\yeps$-transition out of $s$, that is $(s,\yeps,\yait,q)$ is not in $\rho$ for any $q\in S$, we proceed as follows.
For any input symbol $a\in A$:
\begin{enumerate}
	\item $a\in A_i$: if $(s,a,\yait,r)\not\in\rho$ for all $r\in S$, add the
	transition $(s,a,\yait,p)$ to $\mu$;
	\item $a\in A_c$: if $(s,a,W,r)\not\in\rho$ for all $W\in\yGa$ and all $r\in S$, add the
	transition $(s,a,Z,p)$ to $\mu$;
	\item $a\in A_r$: if $(s,a,W,r)\not\in\rho$ for some $W\in\yGa_\bot$ and all $r\in S$, add the
	transition $(s,a,W,p)$ to $\mu$.
\end{enumerate}
Finally, add to $\mu$ the self-loops $(p,a,\yait,p)$ for all $a\in A_i$, $(p,a,Z,p)$ for all $a\in A_c$, and 
$(p,a,W,p)$ for all $a\in A_r$ and all $W\in\yGa_\bot$. 

It is clear now that $\yS$ has $n+1$ states.
Moreover,  for all $s\in S$ and all $a\in A$ the construction readily allows for a move $\ypdatrtf{(s,a,\yal\bot)}{}{}{(r,\yeps,\ybe\bot)}$ for all $\yal\in \yGa^\star$.
Hence, an easy induction on $\vert \ysi\vert\geq 0$ shows that for all $s\in S$ and all $\yal\in\yGa^\star$ there will always be a computation $\ypdatrtf{(s,\ysi,\yal\bot)}{\star}{\yS}{(r,\yeps,\ybe\bot)}$, for some $r\in S$ and some $\ybe\in\yGa^\star$.
That is, the modified version is a non-blocking VPA.

The construction adds no $\yeps$-moves, so it is  clear that $\yQ$ has no $\yeps$-moves when we start  with a VPA $\yS$ that already has  no $\yeps$-moves.

Assume that $\yS$ is deterministic.
Since the construction adds no $\yeps$-moves, the new VPA $\yQ$ does not violate condition (3) of Definition~\ref{def:vpa-determinism}.
Also, none of the self-loops added at the new state $p$ violate any of the conditions of Definition~\ref{def:vpa-determinism}.
When $a\in A_c$ a new transition $(s,a,Z,p)$ is only added to $\yQ$ when there are no transition $(s,a,X,r)$ already in $\yS$, for any $X\in \yGa$, $r\in S$. 
Hence, $\yQ$ does not violate condition (1) of Definition~\ref{def:vpa-determinism}.
Likewise, we only add $(s,a,\yait,p)$, or $(s,a,W,p)$, to $\mu$ when we find no $(s,a,\yait,r)$, respectively we find no $(s,a,W,r)$, in $\rho$ for any $r$ in $S$. 
Hence, $\yQ$ does not violate condition (2) of Definition~\ref{def:vpa-determinism}. 
We conclude that $\yQ$ is deterministic when $\yS$ is already deterministic.

If  $\ypdatrtf{(s_0,\ysi,\bot)}{}{\yS}{(f,\yeps,\yal\bot)}$, with $s_0\in S_{in}$, $f\in F$, $\yal\in\yGas$ then we also have $\ypdatrtf{(s_0,\ysi,\bot)}{}{\yQ}{(f,\yeps,\yal\bot)}$, because all transitions in $\rho$ are also in $\mu$.
Hence, $L(\yS)\ysse L(\yQ)$.
For the converse, assume $\ypdatrtf{(s_0,\ysi,\bot)}{}{\yQ}{(f,\yeps,\yal\bot)}$, with $s_0\in S_{in}$, $f\in F$, $\yal\in\yGas$.
Note that all new transitions added to $\mu$ have the new state $p$ as a target state.
Thus, since the new state $p$ is not in $S_{in}$ nor in $F$, we see that all transitions used in this run over $\yQ$ are also in $\rho$, and
we then get $\ypdatrtf{(s_0,\ysi,\bot)}{}{\yS}{(f,\yeps,\yal\bot)}$.
Thus we also have $L(\yQ)\ysse L(\yS)$, showing that $L(\yS)= L(\yQ)$.
\end{proof}


Now the closure of VPLs under union is at hand.
\begin{prop}\label{prop:cup-vpa}
Let $\yS$ and $\yQ$ be two VPAs over an alphabet $A$, with $n$ and $m$ states, respectively. 
Then, we can construct a non-blocking VPA $\yP$ over $A$ with at most $(n+1)(m+1)$ states and such that $L(\yP)=L(\yS)\cup L(\yQ)$. 
Moreover, if $\yS$ and $\yQ$ are deterministic, then $\yP$ is also deterministic and has no $\yeps$-moves.
\end{prop}
%
\begin{proof}
	Let $\yS=\yvpaA$ and $\yQ=\yvpaB$.
	Using Proposition~\ref{prop:non-blocking} we can assume that $\yS$ and $\yQ$ are non-blocking VPAs with $n+1$ and $m+1$ states, respectively.
	
	Let $\yP$ be the product of $\yS$ and $\yQ$ as in Definition~\ref{def:productVPA}, except that we redefine the final states of $\yP$ as $(F\times Q)\cup (S\times G)$.
	Clearly, $\yP$ has $(n+1)(m+1)$ states.
	
	We now argue that $\yP$ is also a non-blocking VPA.
	Let $\ysi\in A^\star$, $(s,q)\in S\times Q$, and let $\yga=(Z_1,W_1)\ldots(Z_k,W_k)\in (\yGa\times\yDe)^\star$.
	Since $\yS$ is a non-blocking VPA, we get $n\geq 0$, $p\in S$, $\yal\in\yGa^\star$ such that  
	$\ypdatrtf{(s,\ysi,Z_1\ldots Z_k\bot)}{n}{\yS}{(p,\yeps,\yal\bot)}$. 
	Likewise, $\ypdatrtf{(q,\ysi,W_1\ldots W_k\bot)}{m}{\yQ}{(r,\yeps,\ybe\bot)}$, for some
	$m\geq 0$, $r\in Q$, $\ybe\in\yDe^\star$.
	Applying Proposition~\ref{prop:stack-size} we get $\vert \yal\vert=\vert\ybe\vert$, and then using  Proposition~\ref{prop:product-behavior} we have 
	$\ypdatrtf{((s,q),\ysi,\yga\bot)}{\star}{\yP}{((p,r),\yeps,\yde\bot)}$ where $\yga=(Z_1,W_1)\ldots (Z_k,W_k)$ and  $\yde\in (\yGa\times\yDe)^\star$.
	This shows that $\yP$ is a non-blocking VPA.
	
	Now suppose that $\ysi\in L(\yP)$, that is $\ypdatrtf{((s_0,q_0),\ysi,\bot)}{\star}{\yP}{((p,r),\yeps,\yal\bot)}$, where $(s_0,q_0)\in S_{in}\times Q_{in}$,  $(p,r)\in (F\times Q)\cup(S\times G)$, and $\yal\in(\yGa\times\yDe)^\star$.
	Take the case when $(p,r)\in (F\times Q)$.
	We get $p\in F$ and $s_0\in S_{in}$. 
	Using Proposition~\ref{prop:product-behavior} we can also write $\ypdatrtf{(s_0,\ysi,\bot)}{\star}{\yS}{(p,\yeps,\ybe\bot)}$, for some $\ybe\in \yGa^\star$.
	This shows that $\ysi\in L(\yS)$.
	By a similar reasoning, when $(p,r)\in (S\times G)$ we get $\ysi\in L(\yQ)$.
	Thus, $L(\yP)\ysse L(\yS)\cup L(\yQ)$. 
	
	Now let $\ysi\in  L(\yS)\cup L(\yQ)$.
	Take the case $\ysi\in L(\yS)$, the case $\ysi\in L(\yQ)$ being similar.
	Then we must  have $\ypdatrtf{(s_0,\ysi,\bot)}{n}{\yS}{(p,\yeps,\yal\bot)}$ for some $n\geq 0$, some $\yal\in \yGa^\star$, and some $p\in F$.
	Pick any $q_0\in Q_{in}$.
	Since $\yQ$ is a non-blocking VPA, Definition~\ref{def:forward} gives some $r\in Q$ and some $\ybe\in \yDe^\star$ such that $\ypdatrtf{(q_0,\ysi,\bot)}{m}{\yQ}{(r,\yeps,\ybe\bot)}$ for some $m\geq 0$, some $r\in Q$ and some $\ybe\in\yDe^\star$. 
	Using Proposition~\ref{prop:stack-size} we conclude that $\vert \yal\vert=\vert\ybe\vert$.
	The only-if part of Proposition~\ref{prop:product-behavior} now yields  
	$\ypdatrtf{((s_0,q_0),\ysi,\bot)}{\star}{\yP}{((p,r),\yeps,\yga\bot)}$, where $\yga\in (\yGa\times\yDe)^\star$.
	Clearly, $(s_0,q_0)\in S_{in}\times Q_{in}$ is an initial state of $\yP$ and $(p,r)\in  (F\times Q)\cup (S\times G)$ is a final state of $\yP$.
	Hence $\ysi\in L(\yP)$. 
	Thus $ L(\yS)\cup L(\yQ)\ysse L(\yP)$, and so $ L(\yS)\cup L(\yQ)=L(\yP)$.
	
	Applying Proposition~\ref{prop:epsilon-deterministic} (2) we see that when $\yS$ and $\yQ$ are deterministic, then $\yP$ is also deterministic and has no $\yeps$-moves.
	
	The proof is now complete.
\end{proof}


We can also show that deterministic VPLs are closed under complementation.
A related result appeared in~\cite{alurm-visibly-2004}, but here we also allow for arbitrary $\yeps$-moves in any model.
\begin{prop}\label{prop:compl-vpa}
Let $\yS=\yvpaA$ be a deterministic VPA  with $n$ states. 
Then, we can construct a non-blocking and deterministic VPA $\yQ$ over $A$ with no $\yeps$-moves,  $n + 1$ states, and such that $L(\yQ)=\ycomp{L(\yS)} = \ySis - L(\yS)$.
\end{prop}
%
%
\begin{proof}
	Applying Propositions~\ref{prop:no-eps-determ} and~\ref{prop:non-blocking}, we can assume that $\yS$ is a non-blocking and deterministic VPA with $n+1$ states and no $\yeps$-moves.
	
	Let $\yQ=\yvpa{S}{S_{in}}{A}{\yGa}{\rho}{S-F}$, that is, we switch the final states of $\yS$.
	Clearly, since $\yS$ and $\yQ$ have the same set of initial states and the same transition relation, 
	we see that $\yQ$ is also a non-blocking and deterministic VPA with $n+1$ states and no $\yeps$-moves.
	
	Pick any $\ysi\in L(\yQ)$, and let $n=\vert \ysi\vert$.
	Since $\yQ$ has no $\yeps$-moves, we must have $\ypdatrtf{(q_0,\ysi,\bot)}{n}{\yQ}{(r,\yeps,\yal\bot)}$ for some $q_0\in S_{in}$, $\yal\in \yGa^\star$, and some $r\in S-F$.
	For the sake of contradiction, assume that $\ysi\in L(\yS)$.
	Then, because $\yS$ has no $\yeps$-moves, we must have $\ypdatrtf{(s_0,\ysi,\bot)}{n}{\yS}{(p,\yeps,\ybe\bot)}$ for some $s_0\in S_{in}$, $\ybe\in \yGa^\star$, and some $p\in F$.
	Since $\yS$ is deterministic we have $\vert S_{in}\vert=1$, so that $q_0=s_0$.
	Since $\yS$ and $\yQ$ have exactly the same set of transitions, we can now write
	$\ypdatrtf{(s_0,\ysi,\bot)}{n}{\yS}{(r,\yeps,\yal\bot)}$.
	Together with $\ypdatrtf{(s_0,\ysi,\bot)}{n}{\yS}{(p,\yeps,\ybe\bot)}$ and Proposition~\ref{prop:vpa-determ}, we conclude that $p=r$, which is a contradiction since $r\in S-F$ and $p\in F$.
	Thus, $L(\yQ)\ysse \ycomp{L(\yS)}$.
	
	Now let $\ysi\not\in L(\yS)$.
	Since $\yS$ is non-blocking, we have $\ypdatrtf{(s_0,\ysi,\bot)}{\star}{\yS}{(p,\yeps,\yal\bot)}$ for some $p\in S$ and $\yal\in \yGa^\star$.
	Because $\ysi\not\in L(\yS)$ we need $p\in S-F$.
	Since $\yS$ and $\yQ$ have the same set of transitions, we can write $\ypdatrtf{(s_0,\ysi,\bot)}{n}{\yQ}{(p,\yeps,\yal\bot)}$.
	Thus, $\ysi\in L(\yQ)$ because $p\in S-F$. 
	Hence, $ \ycomp{L(\yS)}\ysse L(\yQ)$, and the proof is complete. 
\end{proof}

%% file: figs/vpafig.tex
\begin{figure}[tb]
\center

\begin{tikzpicture}[node distance=1cm, auto,scale=.6,inner sep=1pt]
  \node[ initial by arrow, accepting, initial text={}, punkt] (q0) {$s_0$};
  \node[punkt, inner sep=1pt,right=1.5cm of q0] (q1) {$s_1$};  
    \node[punkt, inner sep=1pt,right=2cm of q1] (q2) {$s_2$};  
  \node[punkt, accepting, inner sep=1pt,below=2cm of q1] (qf) {$s_f$};  
      
\path (q0)    edge [ pil, left=50]
                	node[pil,above]{$a/\ypush{B}$} (q1);
\path (q1)    edge [ pil, left=50]
                	node[pil,above]{$b/\ypop{A}$} (q2);
\path (q1)    edge [ pil, left=50]
					node[pil,left]{$b/\ypop{B}$} (qf);
\path (q2)    edge [ pil, left=50]
					node[pil,right]{$b/\ypop{B}$} (qf);
                          	
\path (q1)    edge [loop above] node   {$a/\ypush{A}$} (q1);
\path (q2)    edge [loop above] node   {$b/\ypop{A}$} (q2);


\end{tikzpicture}
\caption{The VPA $\yA$ accepting $a^nb^n$ with $n\geq 0$.}

\label{vpafig}
\end{figure}

%% file: vpts.tex
\section{Reactive Pushdown  Models}\label{sec:conformance}

In this section we introduce the Visibly Pushdown Labeled Transition System (VPTS), and its variation Input/Output VPTS, both  models that can properly specify pushdown reactive systems. 
We also discuss the notion of contracted VPTSs, an important property to allow the model-based testing process, and state the relationship with its associated VPA. 

\subsection{The formalism of VPTS}\label{subsec:lts}\label{subsec:pda}

Visibly Pushdown Labeled Transition System (VPTS) is an appropriate formalism with a potentially infinite memory that allow us to model pushdown reactive systems. 
VPTS models can specify the asynchronous exchange of messages between a system and the environment, where the outputs do not have to occur synchronously with inputs, \emph{i.e.}, output messages are generated as separated events. 
Next we formally define VPTS models. 
\begin{defi}\label{def:plts}
	A \emph{Visibly Pushdown Labeled Transition System} (VPTS)  over an input alphabet $L$ is a tuple $\yS=\yvptsS$, where:
	\begin{itemize}
		\item[---] $S$ is a finite set of \emph{states} or \emph{locations};
		\item[---] $S_{in}\ysse S$ is the set of \emph{initial states}; 
		\item[---] 
		There  is a special symbol $\ytau\notin L$, the \emph{internal action symbol}; 
		\item[---] $\Gamma$ is a  set of \emph{stack symbols}.
		There is a special symbol $\bot\not \in \Gamma$, the \emph{bottom-of-stack symbol}; 
		\item[---] 
		$T=T_c \cup T_r \cup T_i$, where $T_c\ysse S\times L_c \times \yGa  \times S$, $T_r\ysse S\times L_r \times \yGa_\bot  \times S$, and $T_i\ysse S\times (L_i\cup\{\ytau\}) \times \{\yait\} \times S$,  where $\yait\not\in\yGa_\bot$ is a place-holder symbol.
	\end{itemize}
\end{defi}
Assume that $t=(p,x,Z,q)$ is a transition of $T$. 
We call by \emph{push-transition} when $t\in T_c$ and the meaning is reading an input $x$ when moving from the state $p$ to $q$ in $\yS$, and pushes $Z$ onto the stack. 
We also call \emph{pop-transition} if $t\in T_r$, and in this case, the intended meaning is that, changing $\yS$ from $p$ to $q$, $x\in L_r$ is read in $\yS$,  and pops $Z$ from the stack. 
Notice that a pop move can be taken when the stack is empty, leaving the stack unchanged, if the pop symbol is $\bot$. 
Finally, we call by a \emph{simple-transition} when $t\in T_i$ and $x\in L_i$, 
and by an \emph{internal-transition} when $t\in T_i$ and $x=\ytau$. 
A simple-transition $t$ that reads $x$ when moving from $p$ to $q$ does not change the stack. 
Likewise an internal-transition does not change the stack, but further does not read any symbol from the input.

Next we give the notion of configurations and precisely define elementary moves in VPTS models. 
\begin{defi}\label{def:simplemove}
	Let $\yS=\yvptsS$ be a VPTS.
	A \emph{configuration} of $\yS$ is a pair $(p,\yal)\in S\times (\yGas\{\bot\})$. 
	When $p\in S_{in}$ and $\yal=\bot$, we say that  $(p,\yal)$ is an \emph{initial configuration} of $\yS$.
	The set of all configurations of $\yS$ is indicated by $\yltsconf{\yS}$.
	Let $(q,\yal)\in\yltsconf{\yS}$, and let $\ell \in L_\ytau$. 
	Then we write $\ytr{(p,\yal)}{\ell}{(q,\ybe)}$ if there is a transition $(p,\ell,Z,q)\in T$, and either:
	\begin{enumerate}
		\item $\ell \in L_c$, and $\ybe=Z\yal$; 
		\item  $\ell \in L_r$, and  either (i) $Z\neq \bot$ and $\yal=Z\ybe$, or (ii) $Z=\yal=\ybe=\bot$; 
		\item $\ell \in L_i \cup \{\ytau\}$ and $\yal=\ybe$. 
	\end{enumerate}
	Then an \emph{elementary move} of $\yS$ is represented by $\ytr{(p,\yal)}{\ell}{(q,\ybe)}$ when the transition $(p,\ell,Z,q)\in T$ is \emph{used} in this move. 
	Further, after any elementary move $\ytr{(p,\yal)}{\ell}{(q,\ybe)}$, $(q,\ybe)\in\yltsconf{\yS}$ is also a configuration of $\yS$.
\end{defi} 

From now on, when graphically depicting VPTSs, we will represent a push-transition $(s,x,Z,q)$ by {\rm $x/\ypush{Z}$} next to the corresponding arc from $s$ to $q$ in the figure. 
Similarly, a pop-transition $(s,x,Z,q)$ will be indicated by the label {\rm $x/\ypop{Z}$} next to the arc from $s$ to $q$.
Simple- or internal-transitions over $(s,x,\yait,q)$ will be indicated by the label {\rm $x$} next to the corresponding arc. 	

\begin{exam}\label{example-vpts1}
	Figure~\ref{vpts1} represents a VPTS $\yS$ where the set of states is $S=\{s_0,s_1\}$, $S_{in}=\{s_0\}$.
	Also, we have $L_c=\{b\}$, $L_r=\{c,t\}$, $L_i=\{\}$, and $\yGa=\{Z\}$.
	\input{figs/vpts1.tex}
	We have a push-transition $(s_0,b,Z,s_0)$, the pop-transitions  
	$(s_0,c,Z,s_1)$, $(s_0,t,Z,s_1)$, $(s_1,c,Z,s_1)$, $(s_1,t,Z,s_1)$, and the internal-transition $(s_1,\ytau,\yait,s_0)$. 
	Intuitively, the behavior of $\yS$ says that we can have the symbol $b$ as many as we want, while pushing the symbol $Z$ on the stack. 
	Next at least one corresponding $c$ or $t$ must occur, and then one symbol $Z$ must be popped from the stack, or several symbols $c$ and $t$ can occur while the stack memory is not empty. 
	After that, this process can be restarted moving $\yS$  back to state $s_0$, over the internal label $\ytau$. \yfim
\end{exam} 

When a sequence of events is induced over a VPTS model we can observe its behavior. 
That is, we can obtain the semantics of a VPTS by its traces, or behaviors. 
\begin{defi}\label{def:path}
	Let $\yS=\yvptsS$ be a VPTS and let $(p,\yal),(q,\ybe)\in\yltsconf{\yS}$ . 
	\begin{enumerate}
		\item Let $\ysi=l_1,\ldots,l_n$ be a word in $L_\ytau^\star$. 
		We say that $\ysi$ is a \emph{trace} from  $(p,\yal)$ to $(q,\ybe)$ if there are configurations  $(r_i, \yal_i)\in\yltsconf{\yS}$ , $0\leq i\leq n$, such that $\ytr{(r_{i-1},\yal_{i-1})}{l_i}{(r_i,\yal_i)}$, $1\leq i\leq n$,  with $(r_0,\yal_0)=(p,\yal)$ and $(r_n,\yal_n)=(q,\ybe)$.
		\item Let $\ysi\in L^\star$. 
		We say that $\ysi$ is an \emph{observable trace}  from $(p,\yal)$ to $(q,\ybe)$ in $\yS$ if there is a trace $\mu$ from $(p,\yal)$ to $(q,\ybe)$ in $\yS$ such that $\ysi=h_\ytau(\mu)$.
	\end{enumerate}
	In both cases we also say that the trace starts at $(p,\yal)$ and ends at $(q,\ybe)$, and we say that the configuration $(q,\ybe)$ is \emph{reachable} from $(p,\yal)$.
	We also say that $(q,\ybe)$ is \emph{reachable in} $\yS$ if it is reachable from an initial configuration of $\yS$.
\end{defi}


Note that moves with internal symbol $\ytau$ can occur in a trace, but when $\ytau$-labels are removed we just say that it is an observable trace. 
If $\ysi$ is a trace from $(p,\yal)$ to $(q,\ybe)$, we can also write  $\ytr{(p,\yal)}{\ysi}{(q,\ybe)}$.
We may also write $\ytr{(p,\yal)}{\ysi}{}$, $\ytr{(p,\yal)}{}{(q,\ybe)}$, and $\ytr{(p,\yal)}{}{}$ when $(q,\ybe)\in \yltsconf{\yS}$, $\ysi\in L_\ytau^*$, or both, respectively, are not important. 
Also we write $\ytru{(p,\yal)}{\ysi}{\yS}{(q,\ybe)}$ to emphasize that the underlying VPTS is $\yS$. 
If $\ysi$ is an observable trace from $(p,\yal)$ to $(q,\ybe)$, we may also write $\ytrt{(p,\yal)}{\ysi}{(q,\ybe)}$, with similar shorthand notation also carrying over to the $\ytrt{}{}{}$ relation.

We call the traces of $(p,\yal)$, or the traces starting at $(p,\yal)$, for all traces starting at a given configuration $(p,\yal)$. 
Now we can define the semantics of a VPTS by all traces starting at an initial configuration. 
\begin{defi}\label{def:trace}
	Let $\yS=\yvptsS$ be a VPTS and let $(p,\yal)\in \yltsconf{\yS}$. 
	\begin{enumerate}
		\item The set of \emph{traces} of $(p,\yal)$ is   $tr(p,\yal)=\{\ysi\yst \ytr{(p,\yal)}{\ysi\,\,}{}\}$.
		The set of \emph{observable traces} of $(p,\yal)$ is $otr(p,\yal)= \{\ysi\yst \ytrt{(p,\yal)}{\ysi\,\,}{}\}$.
		\item The \emph{semantics} of $\yS$ is $\bigcup\limits_{q\in S_{in}}\!\!\!\!tr(q,\bot)$, and the \emph{observable semantics} of $\yS$ is 
$\bigcup\limits_{q\in S_{in}}\!\!\!\!otr(q,\bot)$.
	\end{enumerate}
\end{defi}
We will also indicate the semantics and, respectively, the observable semantics, of $\yS$ by $tr(\yS)$ and $otr(\yS)$.
If $\ytrt{(s,\yal)}{}{(p,\ybe)}$ then we also have $\ytr{(s,\yal)}{}{(p,\ybe)}$ in $\yS$, for all $(s,\yal)$, $(p,\ybe)\in \yltsconf{\yS}$.
Moreover, $otr(\yS)=h_\ytau(tr(\yS))$ and when $\yS$ has no internal transitions we already have $otr(\yS)=tr(\yS)$.

We also note that $\ytau$-labeled self-loops can also be eliminated in a VPTS since they play no role when considering any system behaviors.  
So given a VPTS $\yS=\yvptsS$, for any $s\in S$ we postulate that $\ytr{(s_0,\bot)}{\ysi}{(s,\yal\bot)}$, for some $\yal\in\yGas$, $\ysi\in L_\ytau^\star$,  and $s_0\in S_{in}$. 
Further, if $(s,\ytau,\yait,q)\in T$ then $s\neq q$.
In general, $\ytau$-moves indicate that a VPTS can autonomously move along $\ytau$-transitions, without consuming any input symbol.
But in some cases such moves may be undesirable, or  simply we might want no observable behavior leading to two distinct states.
Then we need the notion of determinism in VPTS models.
\begin{defi}\label{def:vpts-determinism}
Let $\yS=\yvptsS$ be a VPTS. 
We say that $\yS$ is \emph{deterministic} if, for all   $s$, $p\in S_{in}$, $s_1$, $s_2\in S$, $\ybe_1,\ybe_2 \in \yGas$, and $\ysi\in L^\star$,  
we have that 	$\ytrut{(s,\bot)}{\ysi}{}{(s_1,\ybe_1\bot)}$ and $\ytrut{(p,\bot)}{\ysi}{}{(s_2,\ybe_2\bot)}$ imply	$s_1=s_2$ and $\ybe_1=\ybe_2$.
\end{defi}
Net we show that deterministic VPTSs do not have internal moves.
\begin{prop}\label{prop:vpts-deterministic}
	Let $\yS=\yvptsS$ be a deterministic VPTS.
	Then $\yS$ has no $\ytau$-labeled transitions.
\end{prop}
\begin{proof}
	By contradiction, assume that $(s,\ytau,\yait,q)\in T$. 
	Since there is no self-loop with label $\ytau$ we get $s\neq q$ and then $\yal\in\yGas$, $\ysi\in\ L^\star$ such that $\ytr{(s_0,\bot)}{\ysi}{(s,\yal\bot)}$, with $s_0\in S_{in}$.
	Hence, $\ytr{(s_0,\bot)}{\ysi}{(s,\yal\bot)}\ytr{}{\ytau}{(q,\yal\bot)}$. 
	Using Definition~\ref{def:trace} we get $\ytrt{(s_0,\bot)}{\mu}{(s,\yal\bot)}$ and $\ytrt{(s_0,\bot)}{\mu}{(q,\yal\bot)}$, where $\mu=h_\ytau(\ysi)$. 
	Since $s\neq q$, this contradicts Definition~\ref{def:vpts-determinism}. 
\end{proof}


\subsection{Contracted VPTSs}\label{subsec:contracted}

The syntactic description of VPTSs can also be reduced, without loosing any semantic capability, by removing states that are not reachable from any initial state. 
Moreover, we can remove transitions in a VPTS model that cannot be exercised by some trace. 
Since every transition, except possibly for pop transitions, can always be taken, we concentrate on the pop transitions.
\begin{defi}\label{def:vpts-reduced}
	We say that a VPTS $\yS=\yvpts{S}{S_{in}}{L}{\yGa}{T}$ is \emph{contracted} if for every transition $(p,b,Z,r) \in T$ with $b\in L_r$, there are some $s_0\in S_{in}$, $\yal\in\yGa^\star$ and $\ysi\in L ^\star$ such that $\ytrt{(s_0,\bot)}{\ysi}{(p,\yal\bot)}$, where either (i) $\yal=Z\ybe$ for some $\ybe\in\yGa^\star$, or (ii) $\yal=\yeps$ and $Z=\bot$.
\end{defi}

We  can obtain contracted VPTSs using the next Proposition~\ref{prop:contracted-vpts}. 
The idea is to construct a context free grammar (CFG) based on the given VPTS, in such a way that the CFG generates strings where terminals represent VPTS transitions.
The productions of the CFG will indicate the set of transitions that can be effectively used in a trace over the VPTS.
\begin{prop}\label{prop:contracted-vpts} 
	Let $\yS=\yvpts{S}{S_{in}}{L}{\yGa}{T}$ be a VPTS.
	We can effectively construct a contracted VPTS $\yQ=\yvpts{Q}{Q_{in}}{L}{\yGa}{R}$ with $\vert Q\vert\leq \vert S\vert$, and such that $tr(\yS)=tr(\yQ)$. 
	Moreover, if $\yS$ is deterministic, then $\yQ$ is also deterministic.
\end{prop}
%
\begin{proof}
	We have $\yS=\yvpts{S}{S_{in}}{L}{\yGa}{T}$.
	We first construct a context-free grammar (CFG)~\cite{sipser,hopcu-introduction-1979} $G$ whose leftmost
	derivations will simulate traces of $\yS$, and vice-versa.
	
	The terminals of $G$ are the transitions in $T$. 
	The non-terminals are of the form $[s,Z,p]$ where $s,p\in S$ are states of $\yS$ and $Z\in\yGa_\bot$ is a stack symbol. 
	If state $p$ is not important, we may write $[s,Z,-]$.
	For  the main idea, let  $t_i=[s_i,a_i,Z_i,p_i]$, $1\leq i\leq n$ be transitions of $\yS$ and let $\ysi=a_1a_2\cdots a_n$.
	Then, if in $G$ we have a leftmost derivation
	$$\ycfgtrtfl{[s_0,\bot,-]}{\star}{}{t_1\cdots t_n[r_1,W_1,r_2][r_2,W_2,r_3]\cdots[r_m,W_m,r_{m+1}][r_{m+1},\bot,-]}$$
	it must be the case that $\yS$ starting at the initial configuration $(s_0,\bot)$ can move along the transitions $t_1, \ldots, t_n$, in that order, to reach the configuration $(r_1,W_1W_2\cdots W_m\bot)$,
	that is, $\ytrtf{(s_0,\bot)}{\ysi}{\yS}{(r_1,W_1W_2\cdots W_m\bot)}$ and vice-versa.
	We are also guessing that, when $\yS$ removes  some $W_i$ from the stack  --- if it eventually does, --- then it will enter state $r_i$, $i=1,\ldots, m$.
	
	Formally, let $G=(V, T,P,I)$ where $V$ is the set of non-terminals, $T$ is the set of terminals $P$ is the set of productions and $I$ is the initial non-terminal of $G$. 
	The set of terminals is the same set $T$ of transition of $\yS$.
	The sets $V$ and $P$ are constructed as follows, and where $NV$ is an auxiliary set. 
	Start with $V=\{I\}$, 
	$NV=\{[s_0,\bot,-]:\,\,\text{for all } s_0\in S_{in}\}$, and 
	$P=\{\ycfgtrtf{I}{}{}{[s_0,\bot,-]}:\,\,[s_0,\bot,-]\in NV\}$.
	Next, we apply the following simple algorithm:
	
	\noindent\makebox[\textwidth]{\rule[-4pt]{.4pt}{4pt}\hrulefill\rule[-4pt]{.4pt}{4pt}}\\
	\texttt{\fontsize{8pt}{9pt}\selectfont\newline
		\noindent While $NV\neq\yemp$:
		\begin{enumerate} 
			\item Let $[s,Z,p]\in NV$. Remove $[s,Z,p]$ from $NV$ and add it to $V$.
			\item
			For all $t\in  T$: 
			\begin{enumerate} 
				\item If $t=(s,a,W,q)\in T_c$, add 
				$\ycfgtrtf{[s,Z,p]}{}{}{t[q,W,r][r,Z,p]}$ to $P$, for all $r\in S$. 
				If $[q,W,r] \notin V \cup NV$, add $[q,W,r]$ to $NV$, and if $[r,Z,p] \notin V \cup NV$ add $[r,Z,p]$ to $NV$; 
				\item If $t=(s,a,\yait,q)\in T_i$,  add $\ycfgtrtf{[s,Z,p]}{}{}{t[q,Z,p]}$ to $P$. 
				If $[q,Z,p] \notin V \cup NV$, add $[q,Z,p]$ to $NV$; 
				\item If $t=(s,a,Z,q)\in T_r$, then 
				\begin{enumerate}
					\item if $Z\neq \bot$ and $p=q$ add $\ycfgtrtf{[s,Z,p]}{}{}{t}$ to $P$; and 
					\item if $Z=\bot$ and $p=-$, add $\ycfgtrtf{[s,Z,p]}{}{}{t[q,\bot,-]}$ to $P$, and if
					$[q,\bot,-] \notin V \cup NV$ add $[q,\bot,-]$ to $NV$
				\end{enumerate}
			\end{enumerate}
		\end{enumerate}
	}
	\vspace*{-3ex}\noindent\makebox[\textwidth]{\rule{.4pt}{4pt}\hrulefill\rule{.4pt}{4pt}}
	
	We indicate by $\hookrightarrow$ the leftmost derivation relation induced by $G$
	over $(V\cup  T)^\star$.
	
	The next claim says that leftmost derivations of $G$ faithfully simulate traces of $\yS$.
	\begin{description}
		\item[{\sf Claim 1:}] Let $t_i=(p_i,a_i,Z_i,q_i)$, $1\leq i\leq n$, $n\geq 0$, 
		and $\ysi=a_1a_2\cdots a_n$.
		Assume that 
		\begin{equation}\label{prop2.37c1a}
			\ycfgtrtfl{[s,\bot,-]}{n}{G}{t_1t_2\cdots t_n[u_0,W_1,u_1][u_1,W_2,u_2]\cdots [u_{m-1},W_m,u_m]}.
		\end{equation}
		Then we must have:  
		\begin{align}
			&m\geq 1,  u_m=-, W_m=\bot,\,\text{and}\,\, W_i\neq \bot, 1\leq i<m \label{prop2.37c1i2}\\
			& \ytrtf{(s,\bot)}{\ysi}{\yS}{(u_0,W_1\cdots W_m)} \label{prop2.37c1i}\\
			& \text{either (i) $n=0$ with $s=u_0$; or (ii) $n\geq 1$ with $s=p_1$, $q_n=u_0$}\label{prop2.37c1i3}.
		\end{align}
		\item[{\sf Proof:}] 
		If $n=0$ then $\ysi=\yeps$, $[s,\bot,-]=[u_0,W_1,u_1]$ and $m=1$.
		Hence, $W_m=W_1=\bot$, $u_m=u_1=-$ and $s=u_0$.
		Since we can write $\ytrtf{(s,\bot)}{\ysi}{\yS}{(u_0,W_1)}$,
		the result follows.
		
		Proceeding inductively, fix some $n\geq 1$ and assume the assertive is true for $n-1$.
		Now suppose that Eq.~(\ref{prop2.37c1a}) holds.
		Since derivations in $G$ are leftmost, and each production in $G$ has exactly one terminal symbol as the leftmost symbol in the right-hand side, using the induction hypothesis we can write
		\begin{align}
			\ycfgtrtfl{[s,\bot,-]&}{n-1}{G}{t_1\cdots t_{n-1}[u_0,W_1,u_1][u_1,W_2,u_2]\cdots [u_{m-1},W_m,u_m]}\notag\\
			&\ycfgtrtfl{}{1}{G}{t_1\cdots t_{n-1}t_n\ybe[u_1,W_2,u_2]\cdots [u_{m-1},W_m,u_m]}\label{prop2.37c1e},
		\end{align}
		where $\ycfgtrtf{[u_0,W_1,u_1]}{}{}{t_n\ybe}$ was the production used in the last step. 
		We also get 
		\begin{align}
			&m\geq 1,  u_m=-, W_m=\bot,\,\text{and}\,\, W_i\neq \bot, 1\leq i<m \label{prop2.37c1k2}\\
			& \ytrtf{(s,\bot)}{\ysi_1}{\yS}{(u_0,W_1\cdots W_m)},\,\, \ysi_1=a_1a_2\cdots a_{n-1}\label{prop2.37c1d}\\
			& \text{either (i) $n-1=0$, $s=u_0$; or (ii) $n-1\geq 1$ with $s=p_1$, $q_{n-1}=u_0$}\label{prop2.37c1g}.
		\end{align}
		
		By the construction of $G$ there are three cases.
		
		As a first alternative, assume $a_n\in L_c$, so that $t_n=[p_n,a_n,Z_n,q_n]\in T_c$.
		Then by item (2a) of the construction of $G$ we must have 
		$\ycfgtrtf{[u_0,W_1,u_1]}{}{}{t_n[q_n,Z_n,r][r,W_1,u_1]}$, where $r\in S$, and $p_n=u_0$.
		Thus, $\ybe=[q_n,Z_n,r][r,W_1,u_1]$.
		Together with Eq. (\ref{prop2.37c1e}) we get 
		\begin{align}\label{prop2.37c1f}
			\ycfgtrtfl{[s,\bot,-]&}{n}{G}{t_1\cdots t_{n-1}t_n[q_n,Z_n,r][r,W_1,u_1][u_1,W_2,u_2]\cdots [u_{m-1},W_m,u_m]}.
		\end{align}
		Now define $v_0=q_n$, $X_1=Z_n$, $v_1=r$.
		Also, let $v_{i+1}=u_i$ and $X_{i+1}=W_i$, $1\leq i\leq m$.
		We get
		$$\ycfgtrtfl{[s,\bot,-]}{n}{G}{t_1\cdots t_n[v_0,X_1,v_1][v_1,X_2,v_2]\cdots [v_{m},X_{m+1},v_{m+1}]}.$$
		Clearly, using condition (\ref{prop2.37c1k2}) we get $m+1\geq 1$, $X_{m+1}=W_m=\bot$ and $v_{m+1}=u_m=-$.
		Also $X_1=Z_n$ and since $t_n\in T_c$ we get $Z_n\neq \bot$, so that $X_1\neq \bot$.
		Using condition (\ref{prop2.37c1k2}) we get $X_{i+1}=W_{i}\neq \bot$, $1\leq i\leq m$, and we conclude that condition (\ref{prop2.37c1i2}) holds.
		Next, we examine condition (\ref{prop2.37c1i3}).
		We already have $n\geq 1$, $p_n=u_0$ and $v_0=q_n$. 
		If condition (\ref{prop2.37c1g}.i) holds, then $n=1$  and $s=u_0$, so that $s=p_n=p_1$, and condition (\ref{prop2.37c1i3}) holds.
		On the other hand, if condition (\ref{prop2.37c1g}.ii) holds, we immediately get $s=p_1$ and, because $v_0=q_n$, we conclude 
		condition (\ref{prop2.37c1i3}) holds again.
		Finally, since  $t_n=[p_n,a_n,Z_n,q_n]\in T_c$ and $p_n=u_0$, together with condition (\ref{prop2.37c1d}), and because $\ysi=\ysi_1a_n$, we can write 
		\begin{align*}
			\ytrtf{(s,\bot)&}{\ysi_1}{\yS}{(u_0,W_1\cdots W_m)=(p_n,W_1\cdots W_m)}\\
			\ytrtf{&}{a_n}{\yS}{(q_n,Z_nW_1\cdots W_m)=(v_0,X_1\cdots X_{m+1})}.
		\end{align*}
		Hence, condition	(\ref{prop2.37c1i}) is verified, and we conclude that the claim holds in this case.
		
		For the second alternative, let $a_n\in L_i\cup\{\ytau\}$, so that $t_n=[p_n,a_n,Z_n,q_n]\in T_i$. 
		The reasoning is entirely similar to the preceding case.
		
		As a final alternative, let  $a_n\in L_r$, so that $t_n=[p_n,a_n,Z_n,q_n]\in T_r$.
		Looking at item (2c) of the construction of $G$, and recalling that $\ycfgtrtf{[u_0,W_1,u_1]}{}{}{t_n\ybe}$ is the production used to get Eq. (\ref{prop2.37c1e}), we need $u_0=p_n$ and $Z_n=W_1$.
		
		The first sub-case is when we followed step (2c.i).
		We must then have $W_1\neq \bot$, $u_1=q_n$ and $\ybe=\yeps$.
		From Eq. (\ref{prop2.37c1e}) we get
		$\ycfgtrtfl{[s,\bot,-]}{n}{G}{t_1\cdots t_{n-1}t_n[u_1,W_2,u_2]\cdots [u_{m-1},W_m,u_m]}$.
		Since $W_m=\bot$ and $W_1\neq \bot$ we must have $m\geq 2$.
		Define
		$v_{i-1}=u_{i}$, $X_{i}=W_{i+1}$, $1\leq i\leq m-1$.
		Now we have 
		\begin{align}\label{prop2.37c1j}
			\ycfgtrtfl{[s,\bot,-]}{n}{G}{t_1\cdots t_n[v_0,X_1,v_1][v_1,X_2,v_2]\cdots [v_{m-2},X_{m-1},v_{m-1}]}.
		\end{align}
		Note that $m\geq 2$ implies $m-1\geq 1$, and $v_{m-1}=u_m=-$ and $X_{m-1}=W_m=\bot$. 
		Together with condition (\ref{prop2.37c1k2}) we see that
		$X_i=W_{i+1}\neq\bot$, $1\leq i\leq m-2$, so that condition (\ref{prop2.37c1i2}) holds.
		Since $m\geq 2$,  $u_0=p_n$, $W_1=Z_n$, $u_1=q_n$, and $[p_n,a_n,Z_n,q_n]\in T_r$ we get 
		$$\ytrtf{(u_0,W_1W_2\cdots W_m)}{a_n}{\yS}{(u_1,W_2\cdots W_m)=(v_0,X_1X_2\cdots X_{m-1})}.$$
		Because $\ysi=\ysi_1a_n$, we can compose with Eq. (\ref{prop2.37c1d}) and get
		$$\ytrtf{(s,\bot)}{\ysi}{\yS}{(u_1,W_2\cdots W_m)=(v_0,X_1X_2\cdots X_{m-1})},$$
		so that the condition (\ref{prop2.37c1i}) is satisfied.
		We already have $v_0=u_1$ and $u_1=q_n$, so that $v_0=q_n$.
		It is clear that $n\geq 1$, so in order to verify condition (\ref{prop2.37c1i3}) we need $s=p_1$.
		If $n=1$, condition (\ref{prop2.37c1g}) gives $s=u_0$, and since we already have $u_0=p_n=p_1$, we get $s=p_1$. 
		If $n\geq 2$ we have $n-1\geq 1$ and condition (\ref{prop2.37c1g}) immediately gives $s=p_1$, 
		concluding this sub-case.
		
		For the last sub-case, assume we followed step (2c.ii).
		In this case we get $u_0=p_n$, $W_1=Z_n=\bot$, $u_1=-$, and $\ybe=[q_n,\bot,-]$.
		Since $W_1=\bot$, condition (\ref{prop2.37c1k2}) says that $m=1$.
		We see that condition (\ref{prop2.37c1i2}) is immediately satisfied.
		From Eq. (\ref{prop2.37c1e}) we get
		$\ycfgtrtfl{[s,\bot,-]}{n}{G}{t_1\cdots t_{n-1}t_n[q_n,\bot,-]}$.
		Let $v_0=q_n$, $v_1=-$, and $X_1=\bot$.
		Thus, 
		$\ycfgtrtfl{[s,\bot,-]}{n}{G}{t_1\cdots t_n[v_0,X_1,v_1]}$.
		Since $m=1$, condition (\ref{prop2.37c1d}) reduces to 
		$\ytrtf{(s,\bot)}{\ysi_1}{\yS}{(u_0,W_1)}$.
		Since 	$u_0=p_n$, $W_1=Z_n=\bot$, and $[p_n,a_n,Z_n,q_n]\in  T_r$ we get
		$$\ytrtf{(s,\bot)}{\ysi_1}{\yS}{(u_0,W_1)=(p_n,\bot)}\ytrtf{}{a_n}{\yS}{(q_n,\bot)=(v_0,\bot)},
		$$
		and we see that condition (\ref{prop2.37c1i}) holds.
		As a final step, we verify that condition (\ref{prop2.37c1i3}) also holds.
		Clearly, $n\geq 1$, and we already have $v_0=q_n$.
		Again, if $n=1$, condition (\ref{prop2.37c1g}) gives $s=u_0$, and since we already have $u_0=p_n=p_1$, we get $s=p_1$. 
		If $n\geq 2$ we have $n-1\geq 1$ and condition (\ref{prop2.37c1g}) immediately gives $s=p_1$, 
		concluding the argument for this sub-case.
		This completes the argument for the last alternative.
		
		Since we examined all three alternatives in case (2c), we conclude that the claim holds. \yfim
	\end{description}
	
	For the converse, we show that any trace of $\yS$ can be simulated by a leftmost derivation of $G$.
	\begin{description}
		\item[{\sf Claim 2:}] Let $\ytrtf{(s_0,\bot)}{\ysi}{\yS}{(u_0,W_1\cdots W_m\bot)}$ with $s_0\in S_{in}$, $m\geq 0$, 
		$\ysi=a_1a_2\cdots a_n\in L_\ytau^\star$, and $n\geq 0$. 
		Assume that the transitions used in this trace were, in order, $t_i=(p_i,a_i,Z_i,q_i)\in  T$, $1\leq i\leq n$.
		Then for all $u_i\in S$, $1\leq i\leq m$, we have
		\begin{align}
			&\ycfgtrtfl{[s_0,\bot,-]}{n}{G}{t_1t_2\cdots t_n[u_0,W_1,u_1][u_1,W_2,u_2]\cdots [u_{m-1},W_m,u_m][u_m,\bot,-]}\label{prop2.37-claim2a}\\
			&\text{If $n\geq 1$ then  $s_0=p_1$, $u_0=q_n$.} \label{prop2.37-claim2c}
		\end{align}
		\item[{\sf Proof:}]
		Assume first that $n=0$, so that $\ysi=\yeps$.
		Then $\ytrtf{(s_0,\bot)}{\ysi}{\yS}{(u_0,W_1\cdots W_m\bot)}$  implies $s_0=u_0$, $m=0$.
		Since we can write $s_0\in S_{in}$, we get that $[s_0,\bot,-]$ is a non-terminal of $G$.
		Also, 
		Since $\ycfgtrtfl{[s_0,\bot,-]}{0}{G}{[u_0,\bot,-]}$, condition (\ref{prop2.37-claim2a}) holds with $m=0$.
		Condition (\ref{prop2.37-claim2c}) holds vacuously.
		
		Now assume $n\geq 1$. 
		Since $t_n$ was the transition used in the last step, we can write 
		\begin{align}
			&\ytrtf{(s_0,\bot)}{\ysi_1}{\yS}{(p_n,X_1\cdots X_k\bot)}\ytrtf{}{a_n}{\yS}{(q_n,\ybe \bot)},\label{prop2.37-claim2d}
		\end{align}
		where $\ysi_1=a_1\cdots a_{n-1}$, $k\geq 0$, and $u_0=q_n$.
		From the induction hypothesis, condition (\ref{prop2.37-claim2a}), for all $u_i\in S$, $1\leq i\leq k$, we get 
		\begin{align}
			\ycfgtrtfl{[s_0,\bot,-]}{n-1}{G}{t_1\cdots t_{n-1}[p_n,X_1,u_1]\cdots [u_{k-1},X_k,u_k][u_k,\bot,-]}. \label{prop2.37-claim2e}
		\end{align}
		If $n=1$, so that $n-1=0$, Eq. (\ref{prop2.37-claim2e}) says that 
		$s_0=p_n=p_1$.
		If $n-1>0$, the induction hypothesis in Eq. (\ref{prop2.37-claim2e})
		says that $s_0=p_1$.
		In any case, $s_0=p_1$. 
		Since we already have $u_0=q_n$, we conclude that condition (\ref{prop2.37-claim2c}) always holds.
		
		Next, we argue that either $k>0$ and $[p_n,X_1,u_1]$ was added to the set $NV$, or $k=0$ and $[p_n,\bot,-]$ was added to $NV$  during the construction of $G$.
		If $n-1=0$ then $[p_n,X_1,u_1]=[s_0,\bot,-]$. 
		Since $s_0\in S_{in}$ we know that $[p_n,\bot,-]$ was added to the initial $NV$ set. 
		If $n-1>1$ then  $[p_n,X_1,u_1]$ was on the right-hand side of a production of $G$ and so it was also added to $NV$ during the construction of $G$.
		In any case, we can assume that $[p_n,X_1,u_1]$ was added to the set $NV$. 
		
		Following step (2) in the construction of $G$, we break the argument in there cases:
		\begin{description}
			\item[$a_n\in L_c$:] 
			in this case we have $t_n=(p_n,a_n,Z_n,q_n)\in T_c$.
			Pick any $v\in S$.
			Following item (2a) in the construction of $G$,  we must have  
			$\ycfgtrtf{[p_n,X_1,u_1]}{}{}{t_n[q_n,Z_n,v][v,X_1,u_1]}$ as a production of $G$.
			Together with Eq. (\ref{prop2.37-claim2e}) we can now write 	
			\begin{align}
				\ycfgtrtfl{[s_0,\bot,-]&}{n}{G}{t_1\cdots t_{n-1}t_n[q_n,Z_n,v][v,X_1,u_1]\cdots [u_{k-1},X_k,u_k][u_k,\bot,-]
				},\notag
			\end{align}
			for all $v\in S$, and all $u_i\in S$ $i=1,\cdots, k$.
			Because $t_n$ was the last transition used in Eq. (\ref{prop2.37-claim2d}), we must have $\ybe=Z_nX_1\cdots X_k$.
			Hence, we can use Eq. (\ref{prop2.37-claim2d}) and write 
			$\ytrtf{(s_0,\bot)}{\ysi}{\yS}{(q_n,Z_nX_1\cdots X_k\bot)}$, where $\vert \ysi \vert=n$.
			We conclude that condition (\ref{prop2.37-claim2a}) holds.
			\item[$a_n\in L_i\cup\{\ytau\}$:]  we have $t_n=(p_n,a_n,Z_n,q_n)\in T_i$ and we can reason as in the preceding case.
			
			\item[$a_n\in L_r$:]	
			now we have $t_n=(p_n,a_n,Z_n,q_n)\in T_r$.
			
			First assume that $k=0$ in Eq. (\ref{prop2.37-claim2d}).
			Since $t_n$ was the last transition used in that equation, we need
			$Z_n=\bot$ and $\ybe=\yeps$.
			We now have $\ytrtf{(s_0,\bot)}{\ysi}{\yS}{(q_n,\bot)}$.
			We already have $u_0=q_n$ and, with $k=0$, Eq. (\ref{prop2.37-claim2e}) reduces to 
			$\ycfgtrtfl{[s_0,\bot,-]}{n-1}{G}{t_1\cdots t_{n-1}[p_n,\bot,-]}$.
			We now have $[p_n,\bot,-]$ as a non-terminal of $G$, and $t_n=(p_n,a_n,\bot,q_n)\in T_r$.
			By item (2c.ii) of the construction of $G$, it follows that
			$\ycfgtrtf{[p_n,\bot,-]}{}{}{t_n[q_n,\bot,-]}$ is a production of $G$.
			Composing, we get   $\ycfgtrtfl{[s_0,\bot,-]}{n}{G}{t_1\cdots t_{n}[q_n,\bot,-]}$, and condition (\ref{prop2.37-claim2a}) holds with $m=0$.
			
			Lastly, assume $k\geq 1$ in Eq. (\ref{prop2.37-claim2d}).
			From Definition~\ref{def:simplemove}, it is clear that $X_1\neq \bot$.
			Again, because $t_n$ was the last transition used in Eq. (\ref{prop2.37-claim2d}), we must have $Z_n=X_1\neq \bot$ and $\ybe=X_2\cdots X_k$, so that we now have $\ytrtf{(s_0,\bot)}{\ysi}{\yS}{(q_n,X_2\cdots X_k\bot)}$.
			Recall that $[p_n,X_1,u_1]$ is a non-terminal of $G$, and Eq. (\ref{prop2.37-claim2e}) holds for all $u_i\in S$, $1\leq i\leq k$.
			Choose $u_1=q_n$, so that we have $[p_n,X_1,u_1]$ as a non-terminal of $G$, and $t_n=(p_n,a_n,X_1,u_1)$.
			Now,  using item (2c.i) of the construction of $G$ we see that
			$\ycfgtrtf{[p_n,X_1,u_1]}{}{}{t_n}$ is a production of $G$.
			Composing with Eq. (\ref{prop2.37-claim2e}) we obtain 
			$$\ycfgtrtfl{[s_0,\bot,-]}{n}{G}{t_1\cdots t_{n}[q_n,X_2,u_2]\cdots [u_{k-1},X_k,u_k][u_k,\bot,-]},$$ 
			and condition (\ref{prop2.37-claim2a}) holds with 
			$m=k-1\geq 0$, concluding this case.
		\end{description} 
	\end{description}
	
	Now we can extract a contracted VPTS $\yQ=\yvpts{Q}{Q_{in}}{A}{\yGa}{ R}$ from the original VPTS $\yS$.
	First, we determine the set $LN$ of all non-terminals of $G$ that
	appear in the leftmost position in a derivation of $G$, that is,
	\begin{align}
		LN=\big\{[s,Z,p]\in NV\yst \ycfgtrtfl{I}{\star}{G}{t_1t_2\cdots t_n[s,Z,p]\yal, \yal \in V^\star}\big\}.\label{prop2.37n}
	\end{align}
	This can be accomplished by a backward search on the productions of $G$ to find all non-terminals $[s,Z,p]$ of $G$ that generate at least one string
	of terminals, that is, $\ycfgtrtfl{[s,Z,p]}{\star}{G}{t_1t_2\cdots t_n}$ for some $t_i\in  R$, $1\leq i\leq n$.
	Next, a forward search on the productions of $G$ collects all non-terminals in $LN$.
	
	In a second step, we collect the transitions in $R$ as follows.
	
	\noindent\makebox[\textwidth]{\rule[-4pt]{.4pt}{4pt}\hrulefill\rule[-4pt]{.4pt}{4pt}}\\
	\texttt{\fontsize{8pt}{9pt}\selectfont\newline
		Start with $ R=\yemp$. \\
		Next, for all $[s,Z,p]\in LN$ and all 
		$t=(s,a,W,q)\in  T$, add $t$ to $ R$ if:
		\begin{enumerate}
			\item[3.] $a\in L_c\cup L_i\cup\{\ytau\}$, or
			\item[4.] $a\in L_r$ and either (i) $Z=W\neq \bot$ and $p=q$; or (ii) $Z=W=\bot$ and $p=-$.   
		\end{enumerate}
	}
	\vspace*{-3ex}\noindent\makebox[\textwidth]{\rule{.4pt}{4pt}\hrulefill\rule{.4pt}{4pt}}
	
	In the resulting directed graph formed by all productions in $ R$, remove any state that is not reachable from an initial state in $S_{in}$, and name $Q$ the set of  remaining states.
	Finally, let 
	$Q_{in}=S_{in}$.  
	
	We can now show that $\yQ$ is contracted and equivalent to $\yS$, as needed.
	\begin{description}
		\item[{\sf Claim 3:}] $\yQ$ is a contracted VPTS.
		\item[{\sf Proof:}]
		Let $t=(s,a,W,q)\in R_r$, so that $a\in\ L_r$.
		By step (4) above, we need $[s,W,p]$ leftmost in $G$, and either 
		(i) $W\neq \bot$ and $p=q$, or (ii) $W=\bot$ and $p=-$.
		
		Since $[s,W,p]$ is leftmost in $G$, by the form of the productions in $G$, we need $t_i=(p_i,a_i,Z_i,q_i)$, $1\leq i\leq n$, $n\geq 0$, $[u_{j-1},W_j,u_j]$, $1\leq j\leq m$, $m\geq 1$, such that
		$$\ycfgtrtfl{[s_0,\bot,-]}{n}{G}{t_1t_2\cdots t_n[u_0,W_1,u_1][u_1,W_2,u_2]\cdots [u_{m-1},W_m,u_m]},$$
		with $s_0\in Q_{in}$ and $[s,W,p]=[u_0,W_1,u_1]$.
		Hence, $s=u_0$ and $W_1=W$.
		
		Using Claim 1 we can write  $\ytrtf{(s_0,\bot)}{\ysi}{\yS}{(u_0,W_1W_2\cdots W_m)}
		=(s,WW_2\cdots W_m)$, and where $\ysi=a_1a_2\cdots a_n$.
		Let $\mu=h_{\ytau}(\ysi)$.
		We now have $\ytrut{(s_0,\bot)}{\mu}{\yS}{(s,WW_2\cdots W_m)}$.
		If $m>1$, and remembering that $t=(s,a,W,q)$, we see that condition (i) of Definition~\ref{def:vpts-reduced} is immediately satisfied.
		When $m=1$, condition \ref{prop2.37c1i2} in Claim 1 says that $W_m=\bot$, and we now have 
		$W_m=W_1=W=\bot$ and $\ytrut{(s_0,\bot)}{\mu}{\yS}{(s,\bot)}$.
		Since $t=(s,a,\bot,q)$, we see that condition (ii) of Definition~\ref{def:vpts-reduced} can also be satisfied.
	\end{description}
	
	\begin{description}
		\item[{\sf Claim 4:}] $tr(\yQ)=tr(\yS)$ and $otr(\yQ)=otr(\yS)$. 
		\item[{\sf Proof:}]
		We trivially obtain that  $tr(\yQ)\ysse tr(\yS)$ since $ R \subseteq  T$, $Q_{in}=S_{in}$, and $Q\ysse S$. 
		
		Now assume $\ysi\in tr(\yS)$ with  $\ysi=a_1\cdots a_n$, $n\geq 0$.
		Then we have 
		\begin{equation}\label{prop2.37c}
			\ytrtf{(s_0,\bot)}{\ysi}{\yS}{(f,W_1\cdots W_m\bot)}
		\end{equation}
		with $m\geq 0$, and $s_0\in S_{in}$. 
		Let $t_i=(p_i,a_i,Z_i,q_i)\in  T$, $1\leq i\leq n$, be the transitions used in this trace of $\yS$, and in this order.
		Using Claim 2, we get $r_i\in S$, $1\leq i\leq m$,  
		such that 
		\begin{equation}\label{prop2.37b}
			\ycfgtrtfl{[s_0,\bot,-]}{n}{G}{t_1t_2\cdots t_n[f,W_1,r_1][r_1,W_2,r_2]\cdots [r_{m-1},W_m,r_m][r_m,\bot,-]}.
		\end{equation}
		We want to show that $t_k$ is also a transition of $\yQ$, $1\leq k\leq n$.
		That is, we want to show that $t_k$ was added to $R$ according to rules (3) and (4) above.
		Fix some $k$, $1\leq k\leq n$, with $t_k=(p_k,a_k,Z_k,q_k)$.
		Since the derivation in Eq.~(\ref{prop2.37b}) is leftmost, and $G$ has exactly one terminal in the right-hand side of any production, we must have
		$$\ycfgtrtfl{[s_0,\bot,-]}{k-1}{G}{t_1\cdots t_{k-1}[u_1,X_1,u_2][u_2,X_2,u_3]\cdots [u_\ell,X_\ell,-]},$$
		and the next production used in Eq.~(\ref{prop2.37b}) was $\ycfgtrtf{[u_1,X_1,u_2]}{}{}{t_k\ybe}$, for some $\ybe\in V^\star$.
		From the construction of $G$ we always have $u_1=p_k$, so that  $[u_1,X_1,u_2]=[p_k,X_1,u_2]$ is leftmost in $G$.
		From Eq. (\ref{prop2.37n}) we get $[p_k,X_1,u_2]\in LN$.
		We now look at rules (3) and (4).
		If $a_k\in L_c\cup L_i\cup\{\ytau\}$, rule (3) gives $t_k\in  R$.
		Now $a_k\in L_r$ and recall that $\ycfgtrtf{[u_1,X_1,u_2]}{}{}{t_k\ybe}$ is a production of $G$.
		By rule (2c) we  have either: 
		\begin{itemize}
			\item[(i)]   $u_2=q_k$, $X_1=Z_k\neq \bot$ using rule (2c.i).
			Then, $[u_1,X_1,u_2]=[p_k,Z_k,q_k]\in LN$.
			Since $t_k=(p_k,a_k,Z_k,q_k)$ and  $a_k\in L_r$, rule (4.i) says that $t_k\in  R$. 
			\item[(ii)] $u_2=-$, $X_1=Z_k=\bot$ using rule (2c.ii).
			Now $[u_1,X_1,u_2]=[p_k,Z_k,-]\in LN$.
			Again, $t_k=(p_k,a_k,Z_k,q_k)$ and		$a_k\in L_r$ give $t_k\in R$ using rule (4.ii).
		\end{itemize}
		Now, we have  $t_k\in R$ for all $1\leq k\leq n$, $s_0\in S_{in}$ and $Q_{in}=S_{in}$.
		Then, as in Eq.~(\ref{prop2.37c}) we can now write
		$\ytrtf{(s_0,\bot)}{\ysi}{\yQ}{(f,W_1\cdots W_m\bot)}$.
		Thus $\ysi\in tr(\yQ)$, and we now have $tr(\yS)\ysse tr(\yQ)$. 
	\end{description}
	
	Now, if $tr(\yS)=tr(\yQ)$ then $otr(\yS)=otr(\yQ)$ follows easily from 
	Definitions~\ref{def:path} and \ref{def:trace}.
	
	Finally, assume that $\yS$ is deterministic. 
	Since $\yQ$ is constructed by removing transitions from $\yS$, it is clear from Definition~\ref{def:vpts-determinism} that $\yQ$ is also deterministic.
	This completes the proof.
\end{proof}


\subsection{Relating VPTS and VPA models}\label{subsec:vpts-vpa}

Now we show that any VPTS $\yS$  gives rise to an associated VPA $\yS_\yA$ in a natural way.
We convert any $\ytau$-transition of $\yS$  into a $\yeps$-transition of $\yS_\yA$.
The set of final states of $\yS_\yA$ is just the set of all locations of $\yS$.
Conversely, we can associate a VPA to any given VPTS, provided that all states in the given VPA are final states.
\begin{defi}\label{def:plts-pda}
	We have the following two associations:
	\begin{enumerate}
		\item Let $\yS=\yvptsS$ be a VPTS. 
		The VPA \emph{induced} by $\yS$ is  $\yA_\yS=\yvpa{S}{S_{in}}{L}{\yGa}{\rho}{S}$ where, for all $p$, $q\in S$,  $Z\in \yGa$,  $\ell\in L$, we have: 
		\begin{enumerate}
			\item $(p,\ell,Z,q)\in \rho$ if and only if $(p,\ell,Z,q)\in T$;
			\item $(p,\yeps,\yait,q)\in \rho$ if and only if $(p,\ytau,\yait,q)\in T$.
		\end{enumerate}
		\item Let $\yA=\yvpa{S}{S_{in}}{L}{\yGa}{\rho}{S}$ be a VPA.
		The VPTS \emph{induced} by $\yA$ is $\yS_\yA=\yvpts{S}{S_{in}}{L}{\yGa}{T}$ where:
		\begin{enumerate}
			\item $(p,\ell,Z,q)\in T$ if and only if $(p,\ell,Z,q)\in \rho$;
			\item $(p,\ytau,\yait,q)\in T$ if and only if $(p,\yeps,\yait,q)\in \rho$.
		\end{enumerate}
	\end{enumerate}
\end{defi}
The relationship between the associated models is given by the following result.
\begin{prop}	\label{prop:lang-vpa-vlpts}
	$\yS=\yvptsS$ be a VPTS and  $\yA=\yvpa{S}{S_{in}}{L}{\yGa}{\rho}{S}$ a VPA.
	Assume that either $\yA$ is the VPA induced $\yS$, or $\yS$ is the VPTS induced by $\yA$.
	Then, the following are equivalent, where $\ysi\in L^\star$, $\mu\in (L_\ytau)^\star$,  $s,p\in S$, $\yal,\ybe \in\yGa^\star$, $n\geq 0$:
	\begin{enumerate}
		\item $\ytrut{(s,\yal\bot)}{\ysi}{\yS}{(p,\ybe\bot)}$
		\item  $\ytrtf{(s,\yal\bot)}{\mu}{\yS}{(p,\ybe\bot)}$, $h_\ytau(\mu)=\ysi$
		\item $\ypdatrtf{(s,\ysi,\yal\bot)}{n}{\yA}{(p,\yeps,\ybe\bot)}$,  $n=\vert\mu\vert$
	\end{enumerate}
\end{prop}
\begin{proof}
	From Definition~\ref{def:path} we see that (1) and (2) are equivalent, for any VPTS $\yS$.
	
	If $\yA$ is the VPA induced by $\yS$ then, using Definition~\ref{def:plts-pda}(1), an easy induction on $\vert\mu\vert\geq 0$ shows that if (2) holds then (3) also holds with $n=\vert\mu\vert$.
	Likewise, an easy induction on $n\geq 0$ shows that if (3) holds, then we get some $\mu$ satisfying (2) and with $\vert\mu\vert=n$.	
	If $\yS$ is the VPTS induced by the VPA $\yA$ then the reasoning is very similar, now using  Definition~\ref{def:plts-pda}(2). 
\end{proof}

The observable semantics of $\mathcal{S}$ is just the language accepted by $A_\mathcal{S}$.
We note this as the next proposition.
\begin{prop}	\label{prop:plts-pda}
	Let $\yS$ be a VPTS and  $\yA_{\yS}$ the VPA induced by $\yS$.
	Then $otr(\yS)=L(\yA_{\yS})$ and, further, if $\yS$ is deterministic and contracted
	then $\yA_\yS$ is also deterministic.
	Conversely, let $\yA$ be a VPA and $\yS_\yA$ the VPTS induced by $\yA$.
	Then $L(\yA)=otr(\yS_\yA)$
	and, also, 
	if $\yA$ is deterministic and has no $\yeps$-moves, then $\yS_\yA$ is deterministic.
\end{prop}
\begin{proof}
	Let $\yS=\yvptsS$ and  $\yA_\yS=\yvpa{S}{S_{in}}{L}{\yGa}{\rho}{S}$.
	Then, $otr(\yS)=L(\yA_\yS)$ follows directly from Proposition~\ref{prop:lang-vpa-vlpts}.
		
	Now assume that $\yS$ is deterministic and contracted. 
	We show that $\yA_\yS$  satisfies Definition~\ref{def:vpa-determinism}.
	First, take $p$, $q\in S_{in}$. 
	Then, we have $\ytrt{(p,\bot)}{\yeps}{(p,\bot)}$ and $\ytrt{(q,\bot)}{\yeps}{(q,\bot)}$ in $\yS$.
	Since $\yS$ is deterministic, Definition~\ref{def:vpts-determinism} gives $p=q$.
	This shows that $\vert S_{in}\vert\leq 1$, as required by Definition~\ref{def:vpa-determinism}.
	
	Next, we look at the three assertions at Definition~\ref{def:vpa-determinism}.
	First, let $(p,\ell,Z_i,q_i)\in\rho_c$,  $i=1,2$.
	Since $\ell\in L_c$, Definition~\ref{def:plts-pda} gives $(p,\ell,Z_i,q_i)\in T$ for $i=1,2$. 
	Since there is no self-loop with label $\ytau$ we obtain some $\ysi\in L^\star$ and $s_0\in S_{in}$ such that $\ytrut{(s_0,\bot)}{\ysi}{\yS}{(p,\yal\bot)}$ for some $\yal\in\yGa^\star$. 
	Hence,
	$\ytrut{(s_0,\bot)}{\ysi\ell}{\yS}{(q_i,Z_i\yal\bot)}$ for $i=1,2$. 
	Since $\yS$ is deterministic, Definition~\ref{def:vpts-determinism} gives $q_1=q_2$ and $Z_1=Z_2$.
	We conclude that $A_\yS$ satisfies condition 1 at Definition~\ref{def:vpa-determinism}. 	
	Now suppose we have $(p,\ell,Z,q_i)\in\rho_r\cup \rho_i$, $i=1,2$.
	Since $\ell\in L_r\cup L_i$, Definition~\ref{def:plts-pda} gives $(p,\ell,Z,q_i)\in T$, $i=1,2$. 
	When $\ell\in L_i$ we get $Z=\yait$, and proceed exactly as in the first case.
	When $\ell\in L_r$, since $\yS$ is contracted, Proposition~\ref{prop:contracted-vpts}  gives  $s_0\in S_{in}$, $\ysi\in L^\star$ and $\yal\in\yGa^\star$ such that  $\ytrut{(s_0,\bot)}{\ysi}{\yS}{(p,\yal\bot)}$.
	Further, $\yal= Z\ybe$ when $Z\neq \bot$, or $\yal=\yeps$ when $Z=\bot$.
	In the first case, $\ytrut{(s_0,\bot)}{\ysi\ell}{\yS}{(q_i,\ybe\bot)}$ and, in the second case,  $\ytrut{(s_0,\bot)}{\ysi\ell}{\yS}{(q_i,\bot)}$, $i=1,2$.
	Thus, because $\yS$ is deterministic, Definition~\ref{def:vpts-determinism} forces
	gives $q_1=q_2$.
	We conclude that $A_\yS$ satisfies condition 2 at Definition~\ref{def:vpa-determinism}. 
	Finally suppose we have $(p,x,Z,q_1)\in\rho$ and $(p,\yeps,Z,q_2)\in\rho$ with $x\neq \yeps$.
	Then, Definition~\ref{def:plts-pda} gives $(p,\ytau,\yait,q_2)\in T$
	and $(p,x,Z,q_1)\in T$. 
	Since $\yS$ is deterministic, Proposition~\ref{prop:vpts-deterministic} says that it has no $\ytau$-transitions, and we reached a contradiction.
	Hence,  $A_\yS$ satisfies condition 3 at Definition~\ref{def:vpa-determinism}, and the proof is complete. 
	
	For the converse, let $\yA=\yvpa{S}{S_{in}}{L}{\yGa}{\rho}{S}$ and $\yS_\yA=\yvptsS$. 
	Again, $otr(\yA)=otr(\yS_\yA)$ is immediate from Proposition~\ref{prop:lang-vpa-vlpts}.
	Now assume that $\yA$ is deterministic and has no $\yeps$-moves.
	Clearly, from the construction, we know that $\yS_\yA$ has no $\ytau$-moves. 
	Let $\ytrut{(s,\bot)}{\ysi}{\yS_\yA}{(s_1,\ybe_1\bot)}$ and 
	$\ytrut{(p,\bot)}{\ysi}{\yS_\yA}{(s_2,\ybe_2\bot)}$ with $p,s\in S_{in}$ and $\ysi\in L^\star$.
	Since $\yA$ is deterministic, we get $\vert S_{in}\vert\leq 1$, so that $p=s$.
	Using Definition~\ref{def:path}, and since $\yS_\yA$ has no $\ytau$-moves, we get 
	$\ytrtf{(s,\bot)}{\ysi}{\yS_\yA}{(s_1,\ybe_1\bot)}$ and 
	$\ytrtf{(s,\bot)}{\ysi}{\yS_\yA}{(s_2,\ybe_2\bot)}$.
	From Proposition~\ref{prop:lang-vpa-vlpts} we get 
	$\ypdatrtf{(s,\ysi,\bot)}{n}{\yA}{(s_1,\yeps,\ybe_1\bot)}$ and 
	$\ypdatrtf{(s,\ysi,\bot)}{n}{\yA}{(s_2,\yeps,\ybe_2\bot)}$, where $n=\vert \ysi\vert$.
	Now, since $\yA$ is deterministic, we can use Proposition~\ref{prop:vpa-determ} and conclude that $s_1=s_2$ and $\ybe_1=\ybe_2$.
	This shows that $\yS_\yA$ satisfies  Definition~\ref{def:vpts-determinism}, completing the proof.
\end{proof}
Lemma~\ref{prop:plts-pda} also says that $otr(\yS)$ is a visibly pushdown language, and that
for any given VPTS $\yS$, we can easily construct a VPA $\yA$ with $L(\yA)=otr(\yS)$.

\subsection{Input Output Pushdown Transition Systems}\label{subsec:ioplts}

The VPTS formalism can be used to model systems with a potentially infinite memory and with a capacity to interact asynchronously with an external environment.
In such situations, we may want to treat some action labels as symbols that the VPTS ``receives'' from the environment, and some other action labels as symbols that the VPTS ``sends back'' to the environment. 
The next VPTS variation differentiates between input action symbols and output action symbols.
\begin{defi}\label{def:iolts}
	An Input/Output Visibly Pushdown Transition System (IOVPTS) over an alphabet $L$ is a tuple $\yI=\yiovptsS$, where
	\begin{itemize}
		\item $L_I$ is a finite set of \emph{input actions}, or \emph{input labels}; 
		\item $L_U$ is a finite set of \emph{output actions}, or \emph{output labels}; 
		\item $L_I\cap L_U=\emptyset$, and $L=L_I\cup L_U $ is the set of \emph{actions} or \emph{labels}; and
		\item $\yvptsS$ is an  \emph{underlying VPTS} over $L$, which is associated to $\yI$.
	\end{itemize}
\end{defi}

We 
denote the class of all IOVPTSs with input alphabet $L_I$ and output alphabet $L_U$ by $\yiovp{L_I}{L_U}$. 
Notice that in any reference to an IOVPTS model we can substitute it by its underlying VPTS. 
So if $\yS$ is an IOVPTS with $\yQ$ as its underlying VPTS, then the VPA induced by $\yS$ is the VPA induced by $\yQ$, according to Definition~\ref{def:plts-pda}. 
Likewise for any formal assertion involving IOVPTS models.


Next we define the semantics of an IOVPTS as the set of its observable traces, that is, observable traces of its underlying VPTS.  
\begin{defi}\label{def:iolts-semantics}
	Let $\yI=\yiovptsS$ be an IOVPTS.
	The \emph{semantics} of $\yI$ is the set $otr(\yI)=otr(\yS_\yI)$,
	where $\yS_\yI$ is the underlying VPTS associated to $\yI$.
\end{defi}
Also, when referring an IOVPTS $\yI$, the notation $\underset{\yI}{\rightarrow}$ and $\underset{\yI}{\Rightarrow}$ are to be understood as $\underset{\yS}{\rightarrow}$ and $\underset{\yS}{\Rightarrow}$, respectively, where $\yS$ is the underlying VPTS associated to $\yI$. 

\begin{exam}
	Recall Example~\ref{example-vpts1}. 
	Now, Figure~\ref{vpts1} represents an IOVPTS that describes a machine that dispenses drinks.
	Again we have $L_c=\{b\}$, $L_r=\{c,t\}$ and $L_i=\yemp$, with the initial state $s_0$, but here we also have $L_I=\{b\}$ and $L_U=\{c,t\}$. 
	In this context, symbol $b$ stands for a button where an user can press when asking for a drink, a cup of coffee or a cup of tea, represented by the labels $c$ and $t$, respectively.
	The user can hit the $b$ button while the machine stays at state $s_0$.
	Each time the $b$ button is activated, the machine pushes a symbol $Z$ on the stack, so that the stack is used to count how many times the $b$ button was hit by the user.

	At any instant, after the user has activated the $b$ button at least once, the machine moves to state $s_1$ and starts dispensing either coffee or tea, indicated by the $c$ and $t$ labels.
	It decrements the stack each time a drink is dispensed, so that it will never deliver more drinks than the user asked for.
	
	A move back to state $s_0$, over the internal label $\ytau$, may interrupt the delivery of drinks, so that the user can, possibly, receive less drinks than originally asked for.
	In this case, when the next user will operate  the machine with  residual number of $Z$ symbols in the stack he could, eventually collect more drinks than asked for.
	But the machine will never dispense more drinks than the total number of solicitations. 
	An alternative could be to use one more state $s_2$ to interrupt the transition from $s_1$ to $s_0$ and install a self-loop at $s_2$ that empties the stack.  
	\yfim
\end{exam}

%% file: figs/vpts1.tex
\begin{figure}[tb]
\center

\begin{tikzpicture}[node distance=1cm, auto,scale=.6,inner sep=1pt]
  \node[ initial by arrow, initial text={}, punkt] (q0) {$s_0$};
  \node[punkt, inner sep=1pt,right=2.5cm of q0] (q1) {$s_1$};  
\path (q0)    edge [ pil, left=50]
                	node[pil,above]{$c/\ypop{Z}$} (q1);
\path (q0)    edge [ pil, right=50]
                	node[pil,below]{$t/\ypop{Z}$} (q1);
                	
\path (q0)    edge [loop above] node   {$b/\ypush{Z}$} (q0);
\path (q1)    edge [loop above] node   {$c/\ypop{Z},t/\ypop{Z}$} (q1);

\path (q1)    edge [ pil, bend left=50]
                	node[pil]{$\ytau$} (q0);


\end{tikzpicture}
\caption{A VPTS $\yS_1$, with $L_c=\{b\}$, $L_r=\{c,t\}$, $L_i=\yemp$.}

\label{vpts1}
\end{figure}

%% file: vconf.tex
\section{Conformance Checking and Visibly Pushdown Languages}\label{sec:cfl-suites}

In this section we define a more general conformance relation based on  Visibly Pushdown Languages~\cite{alurm-visibly-2004}, a proper subset of the more general class of context-free languages~\cite{hopcu-introduction-1979}, but a proper superset of the regular languages.
Next we define fault models for VPTSs using this more general relation and 
study the notion of test suite completeness under this setting. 
In sequel we give a 
method to check conformance between an IUT and its specification, both given by VPTS models, and using the more general conformance relation over VPLs.

\input{our-vconf}
\input{our-fault-model}
\input{our-testing}

%% file: our-vconf.tex
\subsection{A General Conformance Relation for VPTS models }\label{subsec:conf}

The more general conformance relation is defined on subset of words specified by a tester. 
Informally, consider a language $D$, the set of ``desirable'' behaviors, and a language $F$, the set of ``forbidden'' behaviors, of a system. 
If we have a specification VPTS $\yS$ and an implementation VPTS $\yI$ we want to say that $\yI$  \emph{conforms} to $\yS$ according to $(D,F)$ if no undesired behavior in $F$ that is observable in $\yI$  is specified in $\yS$, and all desired behaviors in $D$ that are observable in $\yI$ are specified in $\yS$.
This leads to the following definition.
\begin{defi}\label{def:conf}
Let $L$ be an alphabet, and let $D, F\ysse L^\star$.
Let $\yS$ and $\yI$ be VPTSs over $L$.
We say that \emph{$\yI$ $(D,F)$-visibly conforms to $\yS$}, written $\yI\yconf{D}{F} \yS$, if and only if
\begin{enumerate}
\item $\ysi\in otr(\yI)\cap F$, then $\ysi\not\in otr(\yS)$;
\item $\ysi\in otr(\yI)\cap D$, then $\ysi\in otr(\yS)$.
\end{enumerate}
\end{defi}

We note an equivalent way of expressing these conditions that may also be useful.
Recall that the complement of $otr(\yS)$ is $\ycomp{otr}(\yS)=L^\star-otr(\yS)$.
\begin{prop}\label{prop:equiv-conf}
Let $\yS$ and $\yI$ be VPTSs over $L$ and let $D, F\ysse L^\star$.
Then $\yI\yconf{D}{F} \yS$ if and only if
$otr(\yI)\cap \big[(D\cap\ycomp{otr}(\yS))\cup(F\cap otr(\yS))\big]=\yemp$.
\end{prop}
\begin{proof}
From Definition~\ref{def:conf} we readily get $\yI\yconf{D}{F} \yS$ if and only if  $otr(\yI)\cap F\cap otr(\yS)=\yemp$ and 
$otr(\yI)\cap D\cap \ycomp{otr}(\yS)=\yemp$.
And this holds if and only if 
$$\yemp=\big[otr(\yI)\cap F\cap otr(\yS)\big]\cup\big[otr(\yI)\cap D\cap \ycomp{otr}(\yS)\big]=otr(\yI)\cap \big[(D\cap\ycomp{otr}(\yS))\cup(F\cap otr(\yS))\big],$$
as desired.
\end{proof}

\begin{exam}
\label{exemplo3}
Let $\yS$ be an IOVPTS specification depicted in Figure~\ref{iovpts2}, where $L_I=\{a,b\}$, $L_U=\{x\}$, $L_c=\{a\}$, $L_r=\{b,x\}$ and $L_i=\emptyset$.
Also, $S_{in}=\{s_0\}$ and $\yGa=\{A\}$. 
\input{figs/iovpts2.tex}

Take the languages $D=\{a^nb^nx: n\geq 1\}$ and $F=\{a^nb^{n+1}:n\geq 0\}$. 
This says that any behavior consisting of a block of $a$s followed by an equal length block of $b$s and terminating by an $x$, is a desirable behavior. 
Any block of $a$s followed by a lengthier block of $b$s is undesirable.
We  want to check whether the implementation $\yI$ conforms to the specification $\yS$ with respect to the  sets of behaviors described by $D$ and $F$.
That is, we want to check whether $\yI\yconf{D}{F}\yS$.  
\input{figs/impl}

First, we obtain the  VPA $\ycomp{\yS}$ depicted in Figure~\ref{compiovpts2}. 
Since $\yS$ is  deterministic and all its states are final, we just add a new state $err$ to $\ycomp{\yS}$, and for any missing transitions in $\yS$ we add corresponding transitions ending at $err$ in  $\ycomp{\yS}$. 
It is not hard to see that the language accepted by  $\ycomp{\yS}$ is $\ycomp{otr}(\yS)$. 
\input{figs/compiovpts2}
Again, it is easy to see from Figure~\ref{iovpts2} that  $a^nb^{n+1}\not\in otr(\yS)$, for all $n\geq 0$.
So, $F\cap otr(\yS) = \emptyset$. 
Also we see that the VPA $\yD$, depicted at Figure~\ref{vpaD}, accepts the language $D$ and that $D \subseteq \ycomp{otr}(\yS)$. 
Then the VPA $\yD$ accepts the language $T=D\cap \ycomp{otr}(\yS)=(D\cap\ycomp{otr}(\yS))\cup(F\cap otr(\yS))$.
\input{figs/vpaD}

Now  let  $\yI$ be the implementation  depicted in Figure~\ref{impl}. 
A simple inspection also shows that $aabbx$ is accepted by $\yI$, and we also have $aabbx\in D$.
Hence,  $otr(\yI) \cap D\cap\ycomp{otr}(\yS) =otr(\yI)\cap T\neq \yemp$, and Proposition~\ref{prop:equiv-conf}  implies that $\yI\yconf{D}{F} \yS$ does not hold.

On the other hand, if we assume an implementation $\yI$  that is isomorphic to $\yS$, $\yI$ would not have the transition $\ytr{q_2}{x/\bot}{q_1}$ and then $aabbx$ would not be an observable behavior of $\yI$. 
Actually, in this case,  $otr(\yI) \cap D\cap\ycomp{otr}(\yS) = \yemp$. 
So that now 
$otr(\yI)\cap \big[(D\cap\ycomp{otr}(\yS))\cup(F\cap otr(\yS))\big]=\yemp$, 
and therefore $\yI\yconf{D}{F} \yS$, as expected.  
\yfim
\end{exam} 

%% file: figs/iovpts2.tex
\begin{figure}[tb]
\center

\begin{tikzpicture}[node distance=1cm, auto,scale=.6,inner sep=1pt]
  \node[ initial by arrow, initial text={}, punkt] (q0) {$s_0$};
  \node[punkt, inner sep=1pt,right=2.5cm of q0] (q1) {$s_1$};  
  \node[punkt, inner sep=1pt,below right=1.5cm of q0] (q2) {$s_2$};  

\path (q0)    edge [ pil, left=50]
                	node[pil,above]{$b/\ypop{A}$} (q1);

\path (q2)    edge [ pil, left=50]
                 	node[pil,right]{$b/\ypop{A}$} (q1);               
\path (q1)    edge [ pil, bend left=50]
                	node[pil]{$b/\ypop{A}$} (q2);
                	                	
\path (q0)    edge [loop above] node   {$a/\ypush{A}$} (q0);
\path (q1)    edge [loop above] node   {$a/\ypush{A}$} (q1);

\path (q0)    edge [ pil, bend right=50]
                	node[below left]{$x/\ypop{A}$} (q2);
\path (q2)    edge [ pil,  right=90]
                	node[pil,above right]{$a/\ypush{A}$} (q0);


\end{tikzpicture}
\caption{An  IOVPTS specification $\yS$ with $L_I=\{a,b\}$ and $L_U=\{x\}$.}

\label{iovpts2}
\end{figure}
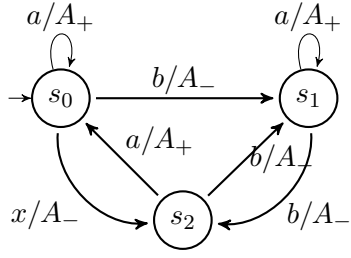

%% file: figs/impl.tex
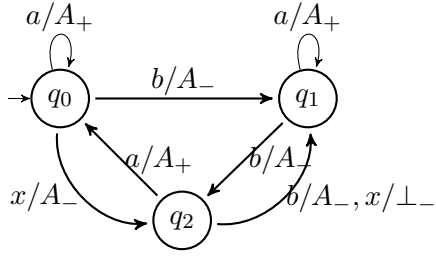
\begin{figure}[tb]
\center

\begin{tikzpicture}[node distance=1cm, auto,scale=.6,inner sep=1pt]
  \node[ initial by arrow, initial text={}, punkt] (q0) {$q_0$};
  \node[punkt, inner sep=1pt,right=2.5cm of q0] (q1) {$q_1$};  
  \node[punkt, inner sep=1pt,below right=1.5cm of q0] (q2) {$q_2$};  

\path (q0)    edge [ pil, left=50]
                	node[pil,above]{$b/\ypop{A}$} (q1);

\path (q2)    edge [ pil, bend right=50]
                 	node[pil,right]{$b/\ypop{A},x/\ypop{\bot}$} (q1);               
\path (q1)    edge [ pil, right=50]
                	node[pil]{$b/\ypop{A}$} (q2);
                	                	
\path (q0)    edge [loop above] node   {$a/\ypush{A}$} (q0);
\path (q1)    edge [loop above] node   {$a/\ypush{A}$} (q1);

\path (q0)    edge [ pil, bend right=50]
                	node[pil,left]{$x/\ypop{A}$} (q2);
\path (q2)    edge [ pil,  right=50]
                	node[pil,right]{$a/\ypush{A}$} (q0);


\end{tikzpicture}
\caption{An implementation IOVPTS $\yI$ with $L_I=\{a,b\}$ and $L_U=\{x\}$.}

\label{impl}
\end{figure}

%% file: figs/compiovpts2.tex
\begin{figure}[tb]
\center

\begin{tikzpicture}[node distance=1cm, auto,scale=.6,inner sep=1pt]
  \node[ initial by arrow, initial text={}, punkt] (q0) {$\bar{s_0}$};
  \node[punkt, inner sep=1pt,right=2.5cm of q0] (q1) {$\bar{s_1}$};  
  \node[punkt, inner sep=1pt,below right=1.5cm of q0] (q2) {$\bar{s_2}$};  
  \node[punkt, accepting, inner sep=1pt,below=3cm of q1] (q3) {$err$};  
  
\path (q0)    edge [ pil, bend left=30]
                	node[pil,above]{$b/\ypop{A}$} (q1);

\path (q2)    edge [ pil, bend left=30]
                 	node[pil,right]{$b/\ypop{A}$} (q1);               
\path (q1)    edge [ pil, bend left=30]
                	node[pil,right]{$b/\ypop{A}$} (q2);
                	                	
\path (q0)    edge [loop above] node   {$a/\ypush{A}$} (q0);
\path (q1)    edge [loop above] node   {$a/\ypush{A}$} (q1);

\path (q0)    edge [ pil, bend right=40]
                	node[pil,above right]{$x/\ypop{A}$} (q2);
\path (q2)    edge [ pil,  bend right=40]
                	node[pil,above right]{$a/\ypush{A}$} (q0);

\path (q2)    edge [ pil, right=30]
                 	node[pil,pos=.35,above right]{$x/\ypop{A}$} node[pil,below left]{$b/\ypop{\bot}$} (q3);  
\path (q2)    edge [ pil, right=30]
					node[pil,pos=.65,above right]{$x/\ypop{\bot}$} node[pil,below left]{$b/\ypop{\bot}$} (q3);

\path (q1)    edge [ pil, bend left=50]
                 	node[pil,pos=.3,right]{$x/\ypop{A}$} (q3);         
\path (q1)    edge [ pil, pos=.5,bend left=50]
					node[pil,right]{$x/\ypop{\bot}$} (q3);         
\path (q1)    edge [ pil, pos=.7,bend left=50]
					node[pil,right]{$b/\ypop{\bot}$} (q3);

\path (q0)    edge [ pil, bend right=50]
                 	node[pil,pos=.45, below  left]{$b/\ypop{\bot}$} (q3);         
\path (q0)    edge [ pil, bend right=50]
					node[pil,pos=.55,below left]{$x/\ypop{\bot}$} (q3);

  \path (q3)    edge [ loop below]
                   	node[pil]{$a/\ypush{A},b/\ypop{A},x/\ypop{A},b/\ypop{\bot},x/\ypop{\bot}$}  (q3);                        	
                	  

\end{tikzpicture}
\caption{The VPA accepting $\ycomp{otr}(\yS)$ for the IOVPTS $\yS$ of Figure~\ref{iovpts2}.}

\label{compiovpts2}
\end{figure}
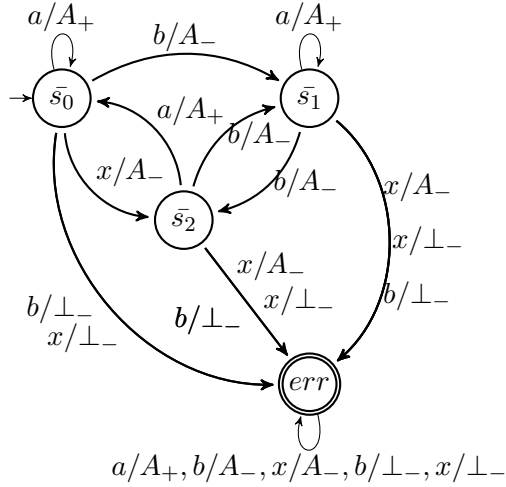

%% file: figs/vpaD.tex
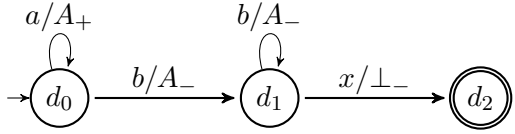
\begin{figure}[tb]
\center

\begin{tikzpicture}[node distance=1cm, auto,scale=.6,inner sep=1pt]
  \node[ initial by arrow, initial text={}, punkt] (q0) {$d_0$};
  \node[punkt, inner sep=1pt,right=2cm of q0] (q1) {$d_1$};  
    \node[punkt, accepting, inner sep=1pt,right=2cm of q1] (q2) {$d_2$};  
    
\path (q0)    edge [ pil, left=50]
                	node[pil,above]{$b/\ypop{A}$} (q1);
\path (q1)    edge [ pil, left=50]
                	node[pil,above]{$x/\ypop{\bot}$} (q2);
                	               	
\path (q0)    edge [loop above] node   {$a/\ypush{A}$} (q0);
\path (q1)    edge [loop above] node   {$b/\ypop{A}$} (q1);


\end{tikzpicture}
\caption{The VPA $\yD$ accepting $D=\{a^nb^nx | n\geq 1 \}$.}

\label{vpaD}
\end{figure}

%% file: our-fault-model.tex
\subsection{Fault Model over Visibly Pushdown Languages}\label{subsec:fault-models} 

When testing pushdown reactive systems we have notice that not only their natural reactive behavior with the environment must considered but also we need to deal with a potentially infinite memory. 
A fault model for systems of this nature must then be able to find faults in IUTs according to their corresponding specification modeled by VPTSs, \emph{i.e.},  
the a fault model must provide a fault detection in the testing process. 
A sequence of symbols that can detect faults in VPTS models is called a \emph{test case}, and a set of test cases gives the notion of a \emph{test suite}. 
So, a test suite $T$ over an alphabet $L$ is a language over $L$, \emph{i.e.} $T\ysse L^\star$, should be engineered  to detect faulty observable behavior of any given IUT, when compared to what has been determined by a specification.
In this case, $T$ can be seen as specifying a fault model, in the sense that test cases in $T$ represent faulty observable behaviors. 
In particular, if $T$ is a VPL, then it can be specified by a VPA $\yA$.
Alternatively, we could specify $T$ by a finite set of VPAs, so that the union of all the undesirable behaviors specified by these VPAs comprise the fault model.

Next we say that an implementation $\yI$ passes a test suite $T$ if no observable behavior of $\yI$ is a harmful behavior present in $T$. 
\begin{defi}\label{def:adherence}
	Let $T$ be a test suite over an alphabet $L$.
	A VPTS $\yQ$ over $L$ \emph{passes} to $T$ if  $\ysi\not\in T$ for all $\ysi\in otr(\yQ)$.
	Further, an IOVPTS $\yI$ over $L$ passes to $T$ if its underlying VPTS passes to $T$.
\end{defi}
In a testing process we also desire test suites to be sound, \emph{i.e.}, when $\yI$ passes a test suite $T$ we always have that $\yI$-visibly conforms to $\yS$. 
The opposite direction is also desirable, that is, if $\yI$-visibly conforms to $\yS$ then $\yI$ passes the test suite $T$. 
\begin{defi}\label{def:complete}
	Let $L$ be an alphabet and let $T$ be a test suite over $L$.
	Also let $\yS$ be a specification VPTS over $L$,  and let $D, F\ysse L^\star$ be languages over $L$.
	We say that:
	\begin{enumerate}
		\item  $T$ is \emph{sound} for $\yS$ and $(D,F)$  if 
		$\yI$ passes to $T$ implies $\yI \yconf{D}{F} \yS$, for all VPTS $\yI$ over $L$.
		\item $T$ is \emph{exhaustive} for $\yS$ and $(D,F)$ if $\yI \yconf{D}{F} \yS$ implies that $\yI$ passes to $T$, for all VPTS $\yI$ over $L$.
		\item $T$ is \emph{complete} for $\yS$ and $(D,F)$ if it is both sound and exhaustive 
		for $\yS$ and $(D,F)$.
	\end{enumerate}
\end{defi}
Notice that an IUT $\yI$ passes to a test suite $T$ when $otr(\yI)\cap T=\yemp$. 

%

Now we can show that the Proposition~\ref{prop:equiv-conf} can construct a test suite which is always complete, and also unique, for a given specification $\yS$. 
\begin{lemm}\label{lemm:always-complete}
	Let $\yS$ be a specification VPTS over $L$, and let $D$, $F\ysse L^\star$ be a pair of languages over $L$.
	Then, the set 
	$\big[(D\cap\ycomp{otr}(\yS))\cup(F\cap otr(\yS))\big]$ is the only complete test suite for $\yS$ and $(D,F)$. 
\end{lemm}
\begin{proof}
	Write $T=\big[(D\cap\ycomp{otr}(\yS))\cup(F\cap otr(\yS))\big]$,
	and let $\yI$ be any implementation VPTS over $L$.
	From Definition~\ref{def:adherence}, we know that $\yI$ adheres to $T$ if and only if $otr(\yI)\cap T=\yemp$.
	From Proposition~\ref{prop:equiv-conf} we get that $\yI\yconf{D}{F} \yS$ if and only if $otr(\yI)\cap T=\yemp$.
	Hence, $\yI$ adheres to $T$ if and only if $\yI\yconf{D}{F} \yS$.
	Since $\yI$ was arbitrary, from Definition~\ref{def:complete} we conclude that $T$ is a complete test suite for $\yS$ and $(D,F)$.
	
	Now, take another test suite $Z\ysse L^\star$, with $Z\neq T$.
	For the sake of contradiction, assume that $Z$ is also complete for $\yS$ and $(D,F)$.
	Fix any implementation $\yI$.
	Since $Z$ is complete, Definition~\ref{def:complete} says that 
	$\yI$ adheres to $Z$ if and only if $\yI\yconf{D}{F} \yS$.
	Using Proposition~\ref{prop:equiv-conf} we know that $\yI\yconf{D}{F} \yS$ if and only if $otr(\yI)\cap T=\yemp$.
	Hence, $\yI$ adheres to $Z$ if and only if $otr(\yI)\cap T=\yemp$.
	From Definition~\ref{def:adherence} we know that $\yI$ adheres to $Z$  if and only if  $otr(\yI)\cap Z=\yemp$.
	We conclude that $otr(\yI)\cap Z=\yemp$ if and only if $otr(\yI)\cap T=\yemp$.
	But $Z\neq T$ gives some $\ysi\in L^\star$ such that $\ysi\in T$ and $\ysi\not\in Z$.
	The case $\ysi\not\in T$ and $\ysi\in Z$ is entirely analogous. 
	We now have $\ysi\in T\cap \ycomp{Z}$.
	If we can construct an implementation VPTS $\yQ$ over $L$  with $\ysi\in otr(\yQ)$, then we have reached a contradiction 
	because we would have $\ysi\in otr(\yQ)\cap T$ and $\ysi\not\in otr(\yQ)\cap Z$.
	But that is simple. 
	Let $\ysi=x_1x_2\ldots x_k$, with $k\geq 0$ and $x_i\in L$ ($1\leq i\leq k$). 
	Define $L_c=L_r=\yemp$ and $L_i=L$, and let $\yQ=\yvpts{Q}{\{q_0\}}{L}{\yemp}{R}$, where $Q=\{q_i\yst 0\leq i\leq k\}$, and $R=\{(q_{i-1},x_i,\yait,q_i)\yst 1\leq i\leq k\}$.
	Clearly, $\ysi\in otr(\yQ)$, concluding the proof. 
\end{proof}
Lemma~\ref{lemm:always-complete} says that the test suite $T=\big[(D\cap\ycomp{otr}(\yS))\cup(F\cap otr(\yS))\big]$ is complete for the specification $\yS$ and the pair of languages $(D,F)$.
So, given an implementation $\yI$, checking if it $(D,F)$-visibly conforms to $\yS$ is equivalent  to checking if $\yI$ adheres to $T$ and, by Definition~\ref{def:adherence},  the latter is equivalent to checking that we have $otr(\yI)\cap T=\yemp$.


%% file: our-testing.tex
\subsection{Checking Visual Conformance for VPTS models}\label{sec:suites-complexity}

When checking conformance one important issue is the size of test suites, relatively to the size of the given specification.
Let $\yS=\yvptsS$  be a VPTS.
A reasonable measure of the size of $\yS$ would be the number of symbols required to write down a complete syntactic description of $\yS$.
Assume that $\yS$ has $m=\vert T\vert$ transitions, $n=\vert S\vert$ states, $\ell=\vert L\vert$ action symbols, and $p=\vert\yGa\vert$ stack symbols. 
Since any transition can be written using $\yoh{\ln (n\ell p)}$ symbols, 
the size of $\yS$ is  $\yoh{m\ln (n\ell p)}$.
We also see that $n$ is $\yoh{m}$ and, clearly, so are $\ell$ and $p$.
Thus, the size of $\yS$ is bounded by $\yoh{m\ln m}$.
If we fix the stack and action alphabets, then the size of the VPTS will be bounded by $\yoh{m}$.
In what follows, and with almost no prejudice, we will ignore the small logarithmic factor.%
\footnote{It is also customary to write $\yoh{m\ln m}$ as $\widetilde{\mathcal{O}}(m)$.
In the sequel, we can always replace $\yoh{\cdot}$ by $\widetilde{\mathcal{O}}(\cdot)$.} 

Given a specification $\yS$ and visibly pushdown languages, $D$ and $F$,  over $L$, Lemma~\ref{lemm:always-complete} says that the fault model $T$ is complete for $\yS$ and $(D,F)$, where $T=\big[(D\cap\ycomp{otr}(\yS))\cup(F\cap otr(\yS))\big]$.
Assume that $\yA_D$ and $\yA_F$ are deterministic VPAs with  $n_D$ and $n_F$ states, respectively, such that $L(\yA_D)=D$ and $L(\yA_F)=F$. 
Also, assume that $\yS$ is deterministic with $n_S$ states.
Proposition~\ref{prop:contracted-vpts} says that we may as well take $\yS$ as a contracted and deterministic VPTS. 
Then Proposition~\ref{prop:plts-pda} gives a deterministic VPA $\yA_1$  with $n_S$ states and such that $L(\yA_1)=otr(\yS)$. 
Using Proposition~\ref{prop:compl-vpa}, we can construct a deterministic VPA $\yB_1$ with $n_S +1$ states and such that $L(\yB_1)=\ycomp{L(\yA_1)}=\ycomp{otr}(\yS)$.
Using Proposition~\ref{prop:determ-complement} we can obtain a deterministic VPA $\yA_2$ with at most $n_S n_F$ states such that $L(\yA_2)=L(\yA_F)\cap L(\yA_1)=F\cap otr(\yS)$, and also can obtain a deterministic VPA $\yB_2$ with $(n_S +1) n_D$ states such that $L(\yB_2)=L(\yA_D)\cap L(\yB_1)=D\cap \ycomp{otr}(\yS)$. 
Proposition~\ref{prop:cup-vpa} then gives a deterministic VPA $\yT$ with $(n_Sn_F+1)(n_Sn_D+n_D+1)$  states and such that $L(\yT)=L(\yA_2)\cup L(\yB_2)=T$. 
Proposition~\ref{prop:cup-vpa} also says that $\yT$ is non-blocking and has no $\yeps$-moves. 
Further, Lemma~\ref{lemm:always-complete} says that $L(\yT)$ is a complete test suite for $\yS$ and $(D,F)$.
Next proposition establishes the construction of a complete test suite for a given specification $\yS$ relatively to a pair of visibly pushdown languages $(D,F)$.
\begin{prop}\label{prop:conf-poli}
	Let $\yS$ be a deterministic IOVPTS over $L$ with $n_S$ states. 
	Also let $\yA_D$ and $\yA_F$ be deterministic VPAs over $L$ with $n_D$ and $n_F$ states, respectively, such that $L(\yA_D)=D$ and  $L(\yA_F)=F$.
	Then, we can construct a deterministic, non-blocking VPA $\yT$ with at most $(n_Sn_F+1)(n_Sn_D+n_D+1)$ states and no $\yeps$-moves,  and such that $L(\yT)$ is a complete test suite for $\yS$ and $(D,F)$.
\end{prop}
\begin{proof}
	The preceding discussion gives a  deterministic and non-blocking VPA $\yT$ with at most $n_T=(n_Sn_F+1)(n_Sn_D+n_D+1)$ states and no $\yeps$-moves, and such that $L(\yT)=T=\big[(D\cap\ycomp{otr}(\yS))\cup(F\cap otr(\yS))\big].$
	%
	%
\end{proof}


Once we have a fault model at hand, which is complete for a given specification using visibly pushdown languages, then we can test whether IUTs conform to that specification under desirable and undesirable behaviors in a more general setting. 
In order to do so we will also need the notion of a \emph{balanced run}~\cite{bonifacio2022cleiej}. 
Given a VPTS $\yV$, and $p,q$ two states of $\yV$ we say that a string $\ysi\in L^\star$ induces a \emph{balanced run}~\cite{bonifacio2022cleiej} from $p$ to $q$ in $\yV$ if we have $\ytrtf{(p,\bot)}{\ysi}{\yV}{(q,\bot)}$. 

Therefore we can state the next theorem to test for visual conformance. 
\begin{theo}(\cite{bonifacio2022cleiej})\label{lemm:vpts-test}
Let $\yS=\yiovpts{S_\yS}{\{s_0\}}{L_I}{L_U}{\yDe_\yS}{T_\yS}\in\yiovp{L_I}{L_U}$ be a deterministic specification, and let $\yI=\yiovpts{S_\yI}{I_{in}}{L_I}{L_U}{\yDe_\yI}{T_\yI}\in\yiovp{L_I}{L_U}$ be a deterministic IUT.
Also let $D,F$ be VPLs with their corresponding deterministic VPAs $\yA_\yD$ and $\yA_\yF$ such that $L(\yA_\yD)= D$ and $L(\yA_\yF)= F$. 
Then we can effectively decide whether $\yI \yconf{D}{F} S$ holds.
Further,  if  $\yI \yconf{D}{F} S$ does not hold then we can find 
$\ysi\in L^\star$  that verify this condition (See Definition~\ref{def:conf}), \emph{i.e.}, $\ysi \in otr(\yI) \cap T$ showing that $otr(\yI) \cap T \neq \emptyset$. 
\end{theo}
\begin{proof}
First let $L=L_I\cup L_U$. 
The proof of Lemma~\ref{lemm:always-complete} shows that $T$ is the only complete test suite for $\yS$ and $(D,F)$. 
From Proposition~\ref{prop:conf-poli} we can obtain a deterministic non-blocking fault model $\yT=\yiovpts{S_\yT}{T_{in}}{L_U}{L_I}{\yDe_\yT}{T_\yT}$ with at most $(n_Sn_F+1)(n_Sn_D+n_D+1)$ states and no $\yeps$-moves, 
and such that $L(\yT)$. 

Since $\yI$ is deterministic, and using Propositions~\ref{prop:contracted-vpts} and~\ref{prop:plts-pda}, we can obtain a  deterministic VPA $\yA$ with at most $n_I$ states, and such that $otr(\yI)=L(\yA)$. 
Then we can construct a deterministic VPA $\yB$ with at most $n_In_T$ states, 
and such that $L(\yB)=L(\yA)\cap L(\yT)=otr(\yI)\cap T$ using Proposition~\ref{prop:determ-complement}. 
From Definition~\ref{def:adherence} we know that to check whether  $\yI \yconf{D}{F} \yS$ is equivalent to checking if $L(\yB)=\emptyset$. 
That is, $\yI \yconf{D}{F} \yS$ does not hold if and only if $\ytrut{(b_0,\bot)}{\ysi}{\yB}{(f,\yal\bot)}$ for some $\ysi\in L^\star$, 
where $b_0$ is an initial state and $f$ is a final state of $\yB$, \emph{i.e.}, $\ysi \in L(\yB)$. 
Then we also see that in order to check whether $\yI \yconf{D}{F} \yS$ does not hold, 
it suffices to check whether a configuration $(f,\yal\bot)$ is reachable from some initial configuration $(b_0,\bot)$ of $\yB$. 

Now we can apply Algorithm~1 of~\cite{bonifacio2022cleiej}, following the same steps to modify $\yB$ in order to 
guarantee that the stack is empty after reaching a final state $f$ in $\yB$, and also that any pop move on an empty stack is eliminated from $\yB$. 
To easy the notation assume $\yB$ is the already VPA with these modifications. 
We see that the problem has been reduced to find a balanced run $\ysi$ from $b_0$ to $f$ in $\yB$ if and only if $\ytrut{(b_0,\bot)}{\ysi}{\yB}{(f,\bot)}$. 
That is, we have reduced the \textbf{vconf} test to the problem of finding a string $\ysi$ that induces a balanced run from $b_0$ to $f$, where $b_0$ is the initial state of the VPA $\yB$ and $f$ is a final state of it, or indicate that such a string does not exist. \yfim
%
%
\end{proof}

\begin{exam}\label{exemplo4}
	Recall Example~\ref{exemplo3}. 
	Again, let $\yS$ be a specification depicted in Figure~\ref{iovpts2}, and $D=\{a^nb^nx: n\geq 1\}$ be the desirable language, where the VPA $\yD$, depicted at Figure~\ref{vpaD}, accepts the language $D$. 
	We also recall that the VPA $\ycomp{\yS}$ depicted in Figure~\ref{compiovpts2} accepts the language $\ycomp{otr}(\yS)$. 
	Since $D \subseteq \ycomp{otr}(\yS)$ then the VPA $\yD$ accepts the language $D\cap \ycomp{otr}(\yS)$.
	
	Now assume a faulty language $F\neq \emptyset$ such that $F=\{a^+x\}$, in this case. 
	This says that any behavior consisting of a block of $a$s (at least one) terminating by an $x$, is an undesirable behavior. 
	So the VPA $\yF$, depicted at Figure~\ref{vpaFnew}, accepts the language $F$ and that $F \subseteq otr(\yS)$. 
	Then the VPA $\yF$ accepts the language $F\cap otr(\yS)$.
	Since the VPA $\yD$ accepts the language $D\cap \ycomp{otr}(\yS)$ and the VPA $\yF$ accepts the language  $F\cap otr(\yS)$, then the language $T=(D\cap\ycomp{otr}(\yS))\cup(F\cap otr(\yS)) = L(\yD) \cup L(\yF)$.
	\input{figs/vpaFnew}
	
	Similarly, we want to check whether an implementation $\yI$ conforms to the specification $\yS$ with respect to the sets of behaviors described by $D$ and $F$.
	That is, we want to check whether $\yI\yconf{D}{F}\yS$.  

	Here also assume an implementation $\yI$  that is isomorphic to $\yS$. 
	So $\yI$ does not have the transition $\ytr{q_2}{x/\bot}{q_1}$ and then the word $aabbx$ would not be an observable behavior of $\yI$. 
	In this scenario, $aabbx \notin otr(\yI)$ and so $otr(\yI)\cap \big[(D\cap\ycomp{otr}(\yS))\big]=\yemp$.  
	The verdict means that any desirable behavior of $D$ that is present in $\yS$ must be a behavior of $\yI$. 
	In other way around if the desirable behavior of $D$ is not in $otr(\yS)$ so such desirable behavior must not be in $otr(\yI)$. 
		
	Now take the word $aax$. 
	We see that $aax$ is an observable behavior of $\yS$ and, clearly, it is also an observable behavior of $\yI$. 
	It is easy to see that $aax\notin D$, \emph{i.e.}, it is not a desirable behavior.  However, $aax\in F$ which means that this word represents an undesirable behavior. 
	Hence, in this case, $otr(\yI)\cap \big[(D\cap\ycomp{otr}(\yS))\cup(F\cap otr(\yS))\big] = otr(\yI) \cap (D \cup F) = otr(\yI)\cap T \neq\yemp$, and Proposition~\ref{prop:equiv-conf} implies that $\yI\yconf{D}{F} \yS$ does not hold, as expected.  
	Noted that, in this case, the implementation does not conform to the specification, even if they are isomorphic, since an observable fault behavior is present in the specification, and consequently, in the implementation $\yI$. \yfim
\end{exam}

%% file: figs/vpaFnew.tex
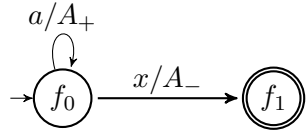
\begin{figure}[tb]
\center

\begin{tikzpicture}[node distance=1cm, auto,scale=.6,inner sep=1pt]
  \node[ initial by arrow, initial text={}, punkt] (q0) {$f_0$};
  \node[punkt, accepting, inner sep=1pt,right=2cm of q0] (q1) {$f_1$};  
        
\path (q0)    edge [ pil, left=50]
                	node[pil,above]{$x/\ypop{A}$} (q1);
                	               	
\path (q0)    edge [loop above] node   {$a/\ypush{A}$} (q0);  


\end{tikzpicture}
\caption{The VPA $\yF$ accepting $F=\{a^+x\}$.}

\label{vpaFnew}
\end{figure}

%% file: conclusions.tex
\section{Conclusion}\label{sec:conclusion}

Pushdown reactive systems are more complex because their behaviors are given by phrase structure rules in addition to their asynchronous interactions with the environment. 
Therefore testing activities are also more intricate when applied to systems of this nature. 
Several methods have been designed to treat the problem of conformance checking and test suite generation for regular reactive systems that can be specified by memoryless formal models. 
Some other previous approaches have addressed these same problems for pushdown reactive systems that have access to an infinite stack memory. 

Here we also study the latter class of reactive systems that can make use of a potentially infinite memory, but we deviated from previous works in the sense that we allow a more wide range of possibilities to define a fault model. 
In this case we make use of visibly pushdown languages to define desirable and undesirable behaviors to be checked over an implementation against to its corresponding specification. 
That is, we treated the problem of conformance checking for asynchronous system with a stack memory using the formalism of IOVPTS, and defined a more general conformance relation for these systems.
So we gave a method to generate test suites that can verify whether this more general conformance relation holds between a specification and a given IUT. 
Also we have shown that these test suites are complete, always giving a conclusive verdict, and also can be generated in polynomial time complexity. 

Our method was developed, and proved correct, using a reduction from the previous work based on {\bf ioco-like} theory, where a balanced run property is applied to obtain verdicts of conformance. 
Therefore our approach exhibits an asymptotic worst case time complexity that can be bounded by $\yoh{n^3+m^2}$, where $n$ and $m$ are proportional to the number of states and transitions from the product between implementations and specification models, respectively.

%% file: main-arxiv2023.bbl
\begin{thebibliography}{10}

\bibitem{alurm-visibly-2004}
Rajeev Alur and P.~Madhusudan.
\newblock Visibly pushdown languages.
\newblock In {\em Proceedings of the Thirty-sixth Annual ACM Symposium on
  Theory of Computing}, STOC '04, pages 202--211, New York, NY, USA, 2004. ACM.

\bibitem{ananbcc-orchestrated-2013}
Saswat Anand, Edmund~K. Burke, Tsong~Yueh Chen, John Clark, Myra~B. Cohen,
  Wolfgang Grieskamp, Mark Harman, Mary~Jean Harrold, Phil McMinn, and
  {others}.
\newblock An orchestrated survey of methodologies for automated software test
  case generation.
\newblock {\em Journal of Systems and Software}, 86(8):1978--2001, 2013.

\bibitem{bonifacio2020cleiej}
Adilson~Luiz Bonifacio and Arnaldo~Vieira Moura.
\newblock Testing asynchronous reactive systems: Beyond the ioco framework.
\newblock {\em {CLEI} Electron. J.}, 24(10), 2021.

\bibitem{bonifacio2022cleiej}
Adilson~Luiz Bonifacio and Arnaldo~Vieira Moura.
\newblock Conformance checking and pushdown reactive systems.
\newblock {\em CLEI Electronic Journal}, 25(3), November 2022.

\bibitem{stg}
Duncan Clarke, Thierry J{\'e}ron, Vlad Rusu, and Elena Zinovieva.
\newblock {STG: A Symbolic Test Generation Tool}.
\newblock In Joost-Pieter Katoen and Perdita Stevens, editors, {\em Tools and
  Algorithms for the Construction and Analysis of Systems}, pages 470--475,
  Berlin, Heidelberg, 2002. Springer Berlin Heidelberg.

\bibitem{hopcu-introduction-1979}
J.~E. Hopcroft and J.~D. Ullman.
\newblock {\em Introduction to Automata Theory, Languages, and Comutation}.
\newblock Addison Wesley, 1979.

\bibitem{uppaal}
Kim~G. Larsen, Marius Mikucionis, Brian Nielsen, and Arne Skou.
\newblock {Testing Real-time Embedded Software Using UPPAAL-TRON: An Industrial
  Case Study}.
\newblock In {\em Proceedings of the 5th ACM International Conference on
  Embedded Software}, EMSOFT '05, pages 299--306. ACM, 2005.

\bibitem{torXAkis}
Wojciech Mostowski, Erik Poll, Julien Schmaltz, Jan Tretmans, and Ronny
  Wichers~Schreur.
\newblock {Model-Based Testing of Electronic Passports}.
\newblock In Mar{\'i}a Alpuente, Byron Cook, and Christophe Joubert, editors,
  {\em Formal Methods for Industrial Critical Systems}, pages 207--209, Berlin,
  Heidelberg, 2009. Springer Berlin Heidelberg.

\bibitem{simap-generating-2014}
Adenilso Sim\~ao and Alexandre Petrenko.
\newblock Generating complete and finite test suite for ioco is it possible?
\newblock In {\em Ninth Workshop on Model-Based Testing (MBT 2014)}, pages
  56--70, 2014.

\bibitem{sipser}
Michael Sipser.
\newblock {\em Introduction to the Theory of Computation}.
\newblock International Thomson Publishing, 1996.

\bibitem{tret-model-2008}
Jan Tretmans.
\newblock Model based testing with labelled transition systems.
\newblock In {\em Formal Methods and Testing}, pages 1--38, 2008.

\end{thebibliography}
